\pdfoutput=1
\documentclass[11pt]{article}
\usepackage[a4paper, total={6.5in, 9in}]{geometry}
\usepackage{amsmath}

\usepackage{enumerate}
\usepackage{amsmath}
\usepackage{amssymb,latexsym}
\usepackage{amsthm}
\usepackage{color}
\usepackage{cancel}
\usepackage{graphicx}
\usepackage{cite}
\usepackage{changes}
\usepackage{hyperref}
\usepackage{comment}
\usepackage{amscd}
\usepackage{mathtools}

\DeclareMathOperator{\C}{\mathcal{C}}

\DeclareMathOperator{\End}{End}
\DeclareMathOperator{\Gal}{Gal}
\DeclareMathOperator{\supp}{supp}

\newtheorem{theorem}{Theorem}[section]
\newtheorem{problem}[theorem]{Problem}

\newtheorem{corollary}[theorem]{Corollary}
\newtheorem{definition}[theorem]{Definition}
\newtheorem{proposition}[theorem]{Proposition}

\newtheorem{remark}[theorem]{Remark}

\newtheorem{construction}[theorem]{Construction}

\newcommand{\cM}{{\mathcal M}}

\newcommand{\F}{{\mathbb F}}

\newcommand{\w}{{\mathrm w}}

\newcommand{\GL}{\hbox{{\rm GL}}}

\newcommand{\Lm}{\mathcal{L}_{m,\sigma}}
\newcommand{\fq}{{\mathbb F}_{q}}

\newcommand{\bfn}{\mathbf{n}}

\newcommand{\la}{\langle}
\newcommand{\ra}{\rangle}

\newcommand{\PG}{\mathrm{PG}}
\newcommand{\N}{\mathrm{N}}
\newcommand{\gcrd}{\mathrm{gcrd}}
\newcommand{\lclm}{\mathrm{lclm}}
\newcommand{\rk}{\mathrm{rk}_q}

\newcommand{\Fmnk}{[\bfn,k]_{q^m/q}}
\newcommand{\Fmnkd}{[\bfn,k,d]_{q^m/q}}

\newcommand{\Ext}{\mathrm{Ext}}

\title{On subspace designs}
\date{}

\author{Paolo Santonastaso and 
Ferdinando Zullo\thanks{The research was supported by the project ``VALERE: VAnviteLli pEr la RicErca" of the University of Campania ``Luigi Vanvitelli'' and was partially supported by the Italian National Group for Algebraic and Geometric Structures and their Applications (GNSAGA - INdAM).}}

\begin{document}

\maketitle

\begin{abstract}
Guruswami and Xing introduced subspace designs in 2013 to give the first construction of positive rate rank metric codes list decodable beyond half the distance.  
In this paper we provide bounds involving the parameters of a subspace design, showing they are tight via explicit constructions. 
We point out a connection with sum-rank metric codes, dealing with optimal codes and minimal codes with respect to this metric.
Applications to two-intersection sets with respect to hyperplanes, two-weight codes, cutting blocking sets and lossless dimension expanders are also provided.

\end{abstract}

\begin{footnotesize}
{\textbf{Keywords}: Subspace design; subspace evasive subspace; $s$-design; MSRD code; scattered subspace; sum-rank metric code; cutting design; dimension expander}\\

{\textbf{MSC2020}: Primary 51E23; 94B05. Secondary 51E22; 05B99.}

\end{footnotesize}


\section{Introduction} 

Guruswami and Xing in \cite{guruswami2013list} introduced a linear-algebraic list decoder which starts by setting the coefficients of the message into a \emph{periodic subspace}, with some more restrictions on such coefficients are needed in order to get a \emph{small} list of solutions.
This last restriction was rephrased in finding a sufficiently large set of subspaces with \emph{small} intersections with all subspaces of fixed dimension. This led to the definition of a (strong) \emph{subspace design}.

\begin{definition}(see \cite[Definition 7]{guruswami2013list})\label{def:strong}
An ordered set $(V_1,V_2,\ldots,V_t)$, where $V_i$ is an $\F_q$-subspace of $V=V(k,q)$, for any $i \in [t]$, is called a \textbf{strong} $(s,A)$-\textbf{subspace design} in $V$ if for every $\F_{q}$-subspace $W \subseteq V$ of dimension $s$, 
$$\sum_{i=1}^t \dim_{\F_{q}}(V_i \cap W) \leq A.$$
\end{definition}

Using the probabilistic method one can show the existence of strong subspace designs with large size and dimension, see e.g. \cite{guruswami2013list}. These subspace designs were used to give a randomized construction of optimal rate list-decodable codes over constant-sized
large alphabets and sub-logarithmic (and even smaller) list size. Moreover, in \cite{guruswami2016explicit}, the authors were able to construct a large \emph{explicit} strong subspace design, where a strong $(s,A)$-subspace design $(V_1,\ldots,V_t)$ in $\F_{q}^k$ is said to be explicit if there exists an algorithm which gives a basis for $(V_1,\ldots,V_t)$ in time poly$(q,k,t)$.

Using such a subspace design, they gave an explicit construction of a subcode of a Gabidulin code which has small intersection with the output of the linear-algebraic list decoder they introduced.
Hence, the code obtained in this way turns out to be efficiently list-decodable.
This was the first explicit example of an efficiently list-decodable rank metric code.

These techniques have been further investigated to obtain more rank metric codes and subspace codes which are efficiently list-decodable (see e.g.\ \cite{Guruswami2016explicitlist,guruswami2014evading}), to obtain efficiently list-decodable algebraic-geometric codes  \cite{guruswami2017optimal} and to construct explicit constant degree \emph{dimension expanders} over large fields \cite{guruswami2021lossless}.

In all of these constructions, the obtained subspace designs are large and with a small intersection with the family of subspaces considered. 
However, explicitly constructing subspace designs satisfying these properties seems challenging. 
In \cite{guruswami2016explicit} explicit constructions of strong subspace designs with parameters close to the probabilistic construction of \cite{guruswami2013list} were found, by requiring large field size, whereas in \cite{guruswami2018subspace}, using algebraic function fields, the authors constructed strong subspace designs over any field, up to some restrictions on the total intersection dimension.

More recently, in \cite{liu2021almost} the authors introduced the notion of \emph{almost affinely disjoint subspaces}, which is strongly related to the subspace designs we are going to introduce, and they have been also used to construct \emph{primitive batch codes}, see \cite{ishai2004batch}.

Our aim is to investigate subspace designs.

\begin{definition}\label{def:subdesi}
An ordered set $\mathcal{U}=(U_1,U_2,\ldots,U_t)$, where $U_i$ is an $\F_q$-subspace of $V=V(k,q^m)$, for any $i \in [t]$, is called an $(s,A)_q$-\textbf{subspace design} in $V$  if \newline $\dim_{\F_{q^m}} (\langle U_1,\ldots,U_t \rangle_{\F_{q^m}}) \geq s$ and for every $\F_{q^m}$-subspace $W \subseteq V$ of dimension $s$, 
$$\sum_{i=1}^t \dim_{\fq}(U_i \cap W) \leq A.$$
Moreover, if 
\[ \langle U_1,\ldots,U_t \rangle_{\F_{q^m}}= V \]
we say that $\mathcal{U}$ is a \textbf{non-degenerate} $(s,A)_q$-subspace design.
\end{definition}

In the above mentioned applications, the authors did not use the strong subspace design they constructed directly, but they used the intersection of such strong subspace designs with a fixed $\fq$-subspace, obtaining in this way a subspace design as in Definition \ref{def:subdesi}.
Moreover, when $t=1$, $(s,A)_q$-subspace designs coincide with the notion of $(s,A)_q$-\emph{evasive subspaces}, originally introduced in  \cite{pudlak2004pseudorandom} by Pudl{\'a}k and R{\"o}dl.

It is easy to see that for a given subspace design the following property holds (see Proposition \ref{prop:s<A}).
For an $(s,A)_q$-subspace design $(U_1,U_2,\ldots,U_t)$ in $V=V(k,q^m)$, $s\leq A$.

This paper focuses on the following two problems.

\begin{problem}\label{prob:A}
Investigate subspace designs attaining equality in Proposition \ref{prop:s<A}.
\end{problem}

\begin{problem}\label{prob:B}
Given an $(s,A)_q$-subspace design $\mathcal{U}$, what is the smallest $A'$ such that $\mathcal{U}$ is an $(s,A')_q$-subspace design?
\end{problem}

In order to study Problem \ref{prob:A}, we give the following definition.

\begin{definition}
An ordered set $(U_1,U_2,\ldots,U_t)$, where $U_i$ is an $\F_q$-subspace of $V=V(k,q^m)$, for any $i \in [t]$, is an $s$-\textbf{design} in $V$ if it is an $(s,s)_q$-subspace design in $V$.
\end{definition}

We start by providing some properties of subspace designs and showing some procedures to construct subspace designs from other subspace designs, such as by using direct sum of subspaces.
As a consequence of the bounds proved in \cite{bartoli2021evasive,blokhuis2000scattered,csajbok2021generalising}, we obtain a first bound regarding the dimension of the subspaces of a subspace design.

The $s$-designs are important examples of subspace designs.
Indeed, when the dimension of all the subspaces is at most $m$, then any $(k-1)$-design in $ V=V(k,q^m)$ (with $m \geq k$) is an \emph{optimal subspace design}, that is connected to \emph{optimal sum-rank metric codes}.
Then we prove that this is also true when considering maximum $1$-designs, i.e. $1$-designs for which the dimension of the subspaces is $km/2$ (when $km$ is even).
We investigate the problem of determining the intersection pattern of the pointsets defined by the subspaces of an $s$-design with hyperplanes. For $s=1$ we show that if such pointsets do not cover the entire space, we obtain a two-intersection set with respect to the hyperplanes.
We conclude the section by showing constructions of maximum $1$-designs.
As we will see, $s$-designs give explicit construction of lossless dimension expanders, without any restriction on the size of the field.

Subsequently, we point out a connection between sum-rank metric codes and subspace designs, already developed in \cite{neri2021geometry}.
The sum-rank metric has been recently investigated especially because of the performance of multishot network coding based on sum-rank metric codes, see \cite{nobrega2010multishot}.
This metric extends the Hamming and rank metric and it is still gaining attention.
Indeed, important and fundamental results were achieved not long ago, see e.g.\ \cite{byrne2021fundamental,martinez2019theory,Martinez2018skew,martinez2020general,neri2021twisted}.

In this paper, we characterize those subspaces defining optimal sum-rank metric codes (that is the ones satisfying equality in the Singleton bound, known as maximum sum-rank metric codes) in terms of subspace designs.
Since the sum-rank metric generalizes the Hamming and rank metric, this connection extends the well-known correspondence between linear maximum distance separable codes and arcs in projective spaces for the Hamming metric (see \cite{ball2020arcs} for a survey on this topic) and linear maximum rank distance codes and scattered subspaces/linear sets  for the rank metric (see \cite{polverino2020connections}).
Those subspace designs will be called \emph{optimal subspace designs} (see Section \ref{sec:optimal}).
By surveying the known examples of linear MSRD codes, we will show the associated subspace designs.
Another important way of constructing subspace designs from a fixed one is by using two types of dualities known as ordinary duality and Delsarte duality, the latter one is defined through the connection with sum-rank metric codes.
We will then explore how to obtain subspace designs from strong subspace designs via subspace evasive subspaces, intermediate fields, high-degree places and Cameron-Liebler sets.
Finally, we will study subspace designs with the property of capturing the structure of hyperplanes, that is their elements contain a basis of every hyperplane.
We call them \emph{cutting design}, since they represent an extension of the notion of cutting blocking sets recently introduced by Bonini and Borello in \cite{Bonini2021minimal}. We then show how to construct cutting designs from known  (linear) cutting blocking sets.
Similarly to the case of cutting blocking sets, cutting designs are in correspondence with \emph{minimal} sum-rank metric codes. Minimal codes attracted attention for their use in secret sharing schemes (see \cite{massey1993minimal}) and very recently in \cite{alfarano2022linear} this notion was extended to the rank metric. In this paper we introduce minimal sum-rank metric codes, that naturally extend those in the Hamming and rank metric.

\subsection{Organization of the paper}

The paper is organized as follows.
In Section \ref{sec:2} we briefly describe some preliminaries on linearized polynomials, linear sets, subspace evasive subspaces and scattered subspaces.
In Section \ref{sec:structures} we start by analyzing the first properties, constructions and bounds of subspace designs. 
In Section \ref{sec:maxsdes} we present constructions of maximum $s$-designs, we describe the intersection pattern between a maximum $1$-design and the hyperplanes, and then we provide some examples. Surprisingly, maximum $1$-designs give examples of two-weight codes and strongly regular graphs.
In Section \ref{sec:Singleton} we recall the connection between systems and linear sum-rank metric codes, which allows us to provide a bound on the parameters of a subspace design and to define the optimal subspace designs.
In Section \ref{sec:dualities} we describe two duality operations on a subspace design and in Section \ref{sec:moreconstr} we provide more bounds and constructions. In Section \ref{sec:fromstrong} we show some methods to obtain a subspace design from a strong subspace design.
Section \ref{sec:cutting} is devoted to the study of cutting designs and their connection with minimal sum-rank metric codes.
In Section \ref{sec:dimexp} we apply the developed theory of maximum $s$-designs to dimension expanders.
We conclude the paper with Section \ref{sec:concl} in which we list some open problems/questions.

\section{Preliminaries}\label{sec:2}

We start by fixing the following notation. Let $p$ be a prime and let $h$ be a positive integer. We fix $q=p^h$ and denote by $\fq$ the finite field with $q$ elements. Moreover, if $m$ is a positive integer then we may consider the extension field $\F_{q^m}$ of degree $m$ over $\fq$. 
Recall that for the extension $\F_{q^m}/\fq$, the \textbf{norm} of an element $\alpha \in \F_{q^m}$ is defined as
$$ \mathrm{N}_{q^m/q}(\alpha):= \prod_{i=0}^{m-1}\alpha^{q^i},$$
and the \textbf{trace} of an element $\alpha \in \F_{q^m}$ is defined as
$$ \mathrm{Tr}_{q^m/q}(\alpha):= \sum_{i=0}^{m-1}\alpha^{q^i}.$$

\noindent We list some more notation which will be repeatedly used in this paper.

\begin{itemize}
    \item $[i]=\{1,\ldots,i\}$;
    \item $S_n$ denotes the symmetric group of order $n$;
    \item $V(k,q)$ denotes a $k$-dimensional $\fq$-vector space;
    \item $\langle U \rangle_{\F_{q}}$ denotes the $\fq$-span of $U$, with $U$ a subset of a vector space $V$;
    \item $\binom{a}{b}_q=\frac{(q^a-1)\cdots (q^{a-b+1}-1)}{(q^b-1)\cdots(q-1)}$ denotes the number of $b$-dimensional $\F_q$-subspaces of $\F_q^a$;
    \item $x \cdot y:=\sum_{i=1}^kx_iy_i$, if $x=(x_1,\ldots,x_k),y=(y_1,\ldots,y_k) \in \F_{q}^k$;
    \item $\mathrm{Hom}_{\fq}(V_1,V_2)$ denotes the set of $\fq$-linear maps between two $\fq$-vector spaces $V_1$ and $V_2$;
    \item $\mathrm{End}_{\fq}(V)=\mathrm{Hom}_{\fq}(V,V)$, where $V$ is an $\fq$-vector space;
    \item $\mathrm{GL}(k,q)$ denotes the general linear group;
    \item $\mathrm{\Gamma L}(k,q)$ denotes the general semilinear group;
    \item $\fq[x]$ denotes the set of polynomials in the indeterminate $x$ with coefficients over $\fq$;
    \item $\fq[x]_{<h}$ denotes the set of polynomials in the indeterminate $x$ with coefficients over $\fq$ and degree less than $h$;
    \item $\PG(k-1,q)$ denotes the projective Desarguesian space of dimension $k$ and order $q$;
    \item $\PG(V,\fq)$, with $V$ an $\fq$-vector space, denotes the projective space obtained by $V$;
    \item $\langle S \rangle$ denotes the span of the points in $S$, with $S$ a subset of $\PG(k-1,q)$.
\end{itemize}

\subsection{Linearized polynomials}

Let $\Gal(\F_{q^m}/\F_q)$ be the Galois group of $\F_{q^m}$ over $\F_q$. Let $\sigma$ be a generator of $\Gal(\F_{q^m}/\F_q)$. 
A $\sigma$-\textbf{polynomial} (or $\sigma$-\textbf{linearized polynomial}) is a polynomial of the form 
\[
F(x)=f_0x+f_1x^{\sigma}+f_2x^{\sigma^2}+\ldots+f_dx^{\sigma^d} \in \F_{q^m}[x].
\]

The $\sigma$-degree of a nonzero $\sigma$-polynomial is defined naturally as $\max\{i:f_i \neq 0\}$ and it is denoted by $\deg_{\sigma}(F(x))$.
We denote by $\mathcal{L}_{m,\sigma}$ the set of $\sigma$-linearized polynomials and we equip it 
with the usual addition $+$ beetween polynomials of $\F_{q^m}[x]$ and the composition $\circ$ defined as 
\[
a x^{\sigma^i} \circ b x^{\sigma^j}=a \sigma^i(b) x^{\sigma^{i+j}},
\]
and then extended to $\sigma$-polynomials by associativity and distributivity. 
With these two operations and with the multiplication by elements in $\F_{q^m}$, $\mathcal{L}_{m,\sigma}$ is an $\F_{q}$-algebra and an $\F_{q^m}$-vector space. For any element $F(x) = \sum_{i=0}^{d}f_ix^{\sigma^i} \in  \mathcal{L}_{m,\sigma}$, one can consider the map
\[
\begin{array}{rccl}
\phi_F:  & \F_{q^m} & \longrightarrow & \F_{q^m} \\
& \beta & \longmapsto & F(\beta):=\sum\limits_{i=0}^{d}f_i\sigma^i(\beta).
\end{array}
\]

\noindent Then the map $F(x) \mapsto \phi_F$ is an $\F_q$-algebra epimorpishm between $\mathcal{L}_{m,\sigma}$ and $\End_{\F_q}(\F_{q^m})$. So we can identify $F(x)$ with the map $\phi_F$, and we will refer to $\rk(F(x))$ and $\ker(F(x))$ to indicate the rank over $\F_q$ of $\phi_F$ and its kernel, respectively.

For $\sigma$-polynomials we have the following bound on the number of roots. 

\begin{theorem}(see\cite[Lemma 3.2]{guralnick1994invertible} and \cite[Theorem 5]{Gow})\label{Gow}
Consider a non-zero $\sigma$-polynomial
$F(x)$.
Then 
\[ \dim_{\F_q}(\ker (F(x)))\leq \deg_{\sigma}(F(x)). \]
\end{theorem}

$\Lm$ is a right-Euclidean domain with respect to $\sigma$-degree. This implies that if $F_1(x),F_2(x)$ are nonzero $\sigma$-polynomials, then the notions of greatest common right divisor, which we denote by $\gcrd(F_1(x),F_2(x))$, and least common left multiple, denoted by $\lclm(F_1(x),F_2(x))$ are well-defined. 

Let $F(x) = f_0x+f_1x^{\sigma}+\cdots+f_d x^{\sigma^d} \in \mathcal{L}_{m,\sigma}$ and $\alpha \in \F_{q^m}^*$. Denote by $F_{\alpha}(x)$ the $\sigma$-polynomial
\[
F_{\alpha}(x):=\sum_{i=0}^d f_i \N_{\sigma}^i(\alpha)x^{\sigma^i},
\]
where $\N_{\sigma}^i(\alpha)=\prod_{j=0}^{i-1}\sigma^j(\alpha) $.

For a non-zero $\sigma$-polynomial $F(x) \in \Lm$, and $\lambda \in \F_q^*$, the $\lambda$-value for $F(x)$ is the integer
\[
d_\lambda(F(x))=\deg(\gcrd(F(x),x^{\sigma^n}-\lambda x)),
\]
see \cite[Proposition 3]{mcguire2019characterization} and \cite[Definition 2.8]{neri2021twisted}. The $\lambda$-value of a $\sigma$-polynomial $F(x)$ gives information on the number of roots of $F_{\alpha}(x)$ with $\N_{q^m/q}(\alpha)=\lambda$ as stated in the following result. 

\begin{theorem} (see \cite[Theorem 3.10, Proposition 6.1]{neri2021twisted}) \label{th:boundker}
Let $\alpha_1,\ldots,\alpha_t \in \F_{q^m}^*$. Let $\N_{q^m/q}(\alpha_i)=\lambda_i$. Suppose that $\lambda_i \neq \lambda_j$, if $i\neq j$. Let $F(x)=f_0x+f_1x^{\sigma}+\ldots+f_dx^{\sigma^d} \in \Lm$ be a non-zero $\sigma$-polynomial with $\deg_{\sigma}(F(x))=d$. We have the following:
\begin{enumerate}
    \item for any $i \in [t]$ \[\dim_{\F_q}(\ker (F_{\alpha_i}(x)))=d_{\lambda_i}(F(x));\] 
\item \[
    \sum_{i=1}^t \dim_{\F_q}( \ker (F_{\alpha_i}(x)))\leq \deg_{\sigma}(F(x));
\]
\item if $\sum_{i=1}^t \dim_{\F_q}( \ker (F_{\alpha_i}(x)))=d,$ then
\[
\N_{q^m/q}(f_0/f_d)=(-1)^{dm} \prod_{i=1}^{t} \lambda_i^{d_{\lambda_i}(F(x))}.
\]
\end{enumerate}
\end{theorem}

\subsection{Subspace evasive subspaces and scattered subspaces}

Let $V$ be any non-empty set.
In \cite{pudlak2004pseudorandom}, Pudl{\'a}k and R{\"o}dl introduced the notion of being evasive for a subset of $V$ with respect to a family $\mathcal{F}$ of subsets of $V$, in order to construct explicit Ramsey graphs. Later, this notion was adapted to the case in which $V$ is a vector space and $\mathcal{F}$ the family of all subspaces of $V$ with a fixed dimension, since they turn out to be very useful in constructing explicit list decodable codes with optimal rate, see e.g.\ \cite{dvir2012subspace,guruswami2011linear,Guruswami2016explicitlist}.
In this paper we will mainly use the following notion of subspace evasive subspace, which corresponds to \cite[Definition 1.1]{bartoli2021evasive}.

\begin{definition}
    An $\F_q$-subspace $S \subseteq V=V(k,q^m)$ is said to be $(s,r)_q$-evasive if $\dim_{\F_{q^m}}(\langle S \rangle_{\mathbb{F}_{q^m}})\geq s$ and for every $\F_{q^m}$-subspace $W \subseteq V$ of dimension $s$, we have 
    $$\dim_{\fq}(S \cap W) \leq  r.$$
\end{definition}

Subspace evasive subspaces have been recently investigated in \cite{bartoli2021evasive}, where several properties and constructions have been presented.
We list some of them, which will be useful for our purposes. 

\begin{proposition}(see \cite[Proposition 2.6]{bartoli2021evasive})
If $U$ is a $(s,r)_q$-evasive subspace in $V=V(k,q^m)$, then it is also $(s-h,r-h)_q$-evasive for any $h \in \{0,\ldots,s-1\}$.
\end{proposition}

Bounds on the dimension of a subspace evasive subspace have been provided.

\begin{theorem}\label{th:bounddimension}
Let $U$ be an $\fq$-subspace of $V=V(k,q^m)$ of dimension $d$.
Then
\begin{enumerate}
    \item  If $U$ is an $(s,s)_q$-evasive subspace and $m\geq s+1$, then
    \[
    d \leq \frac{mk}{s+1},
    \]
    (see \cite[Theorem 4.3]{blokhuis2000scattered},\cite[Theorem 2.3]{csajbok2021generalising}, \cite[Corollary 4.9]{bartoli2021evasive});
    \item If $U$ is an $(s,s)_q$-evasive subspace, $\langle U \rangle_{\F_{q^m}}=V$ and $m< s+1$ then
    \[ d\leq k, \] 
    (see \cite[Theorem 2.3]{csajbok2021generalising});
    \item If $U$ is an $(k-1,r)_q$-evasive subspace and $r<(k-1)m$, then $d\leq m+r-1$ (see \cite[Theorem 4.2]{bartoli2021evasive});
    \item If $U$ is an $(k-1,r)_q$-evasive subspace and $r<k-2+m/(k-1)$, then $d\leq m+r-1-k$ (see \cite[Theorems 4.2]{bartoli2021evasive});
    \item If $U$ is an $(s,r)_q$-evasive subspace and $r<m$, then $d \leq mk-\frac{mks}{r+1}$ (see \cite[Corollary 4.9]{bartoli2021evasive});
    \item If $U$ is an $(s,r)_q$-evasive subspace, then $|U|\leq \frac{(q^r-1)(q^{km}-1)}{q^{sm}-1}+1$ (see \cite[Theorem 4.3]{bartoli2021evasive}).
\end{enumerate}
\end{theorem}
Moreover, we recall the following results from \cite{bartoli2021evasive}.

\begin{proposition}(see \cite[Propositions 4.6 and 4.7]{bartoli2021evasive})\label{prop:constr12}
Let $V=V(k,q^m)$.
\begin{itemize}
    \item [i)] If $r \geq (k-2)(m-1)+1$, then there exists a $(k-1,r)_q$-evasive subspace in $V$ of dimension $m+r-1$;
    \item [ii)] If $km$ is even and $r\geq km/2-m+1$, then there exists a $(k-1,r)_q$-evasive subspace in $V$ of dimension $m+r-1$.
\end{itemize}
\end{proposition}

For a $(s,r)_q$-evasive subspace $U$, since $\dim_{\F_{q^m}}(\langle S \rangle_{\mathbb{F}_{q^m}})\geq s$, it is easy see that $r \geq s$. A particular class of subspace evasive subspaces is given by the $s$-scattered subspaces. 

\begin{definition}
    An $\F_q$-subspace $U$ of $V=V(k,q^m)$ is called $s$-scattered $1 \leq s \leq k-1$, if $\langle U \rangle _{\F_{q^m}}=V$ and it is also an $(s,s)_q$-evasive subspace. 
\end{definition}

Scattered subspaces were first introduced by Blokhuis and Lavrauw in \cite{blokhuis2000scattered} for $h=1$, then generalized independently by Lunardon in \cite{lunardon2017mrd} and Sheekey and Van de Voorde in \cite{sheekeyVdV} for $h=k-1$, and then in \cite{csajbok2021generalising} for general $h$.
Special attention has been paid for these subspaces especially because of their connection with maximum rank distance codes, see e.g. \cite{polverino2020connections,sheekey2016new,zini2021scattered}.

There are known constructions of $s$-scattered subspaces in $V=V(k,q^m)$ of dimension $\frac{mk}{s+1}$ in the following cases:
\begin{itemize}
    \item[a)] $mk$ is even, see \cite{ball2000linear,bartoli2018maximum,blokhuis2000scattered,csajbok2017maximum};
    \item[b)] $s+1 \mid k$ and $m\geq s+1$, see \cite{csajbok2021generalising,napolitano2021linear};
    \item[c)] $mk'$ is even, $s=m-3$ and $k=k'(m-2)/2$, see \cite{csajbok2021generalising}.
\end{itemize}

More recently, in \cite{bartoli2021evasive} was exhibited a construction of a $1$-scattered subspace of dimension $7$ in $V(3,q^5)$, where $q=p^h$ and $p \in \{2,3,5\}$.

Moreover, in \cite[Corollary 5.2]{zini2021scattered} $s$-scattered subspaces were characterized as follows. If $s+1 \mid km$ and $m\geq s+3$, then an $\fq$-subspace $U$ of dimension $km/(s+1)$ in $V(k,q^m)$ is $s$-scattered if and only if it is an $\left(k-1,\frac{km}{s+1}-m+s\right)_q$-evasive subspace.

\subsection{Linear sets}

Let $V$ be a $k$-dimensional $\F_{q^m}$-vector space and let $\Lambda=\PG(V,\F_{q^m})=\PG(k-1,q^m)$.
Let $U \neq \{0\}$ be an $\fq$-subspace of $V$ of dimension $n$, then the set of points
\[ L_U=\{\la {u} \ra_{\mathbb{F}_{q^m}} : {u}\in U\setminus \{{ 0} \}\}\subseteq \Lambda \]
is said to be an $\fq$-\textbf{linear set of rank $n$}.

The \textbf{weight} of a subspace $\mathcal{S}=\PG(W,\F_{q^m})\subseteq \Lambda$ in $L_U$ is defined as 
\[ w_{L_U}(\mathcal{S})=\dim_{\fq}(U\cap W). \]
Denote by $N_i$ the number of points of weight $i$ in $L_U$.

The following relations hold:
\begin{equation}\label{eq:card}
    |L_U| \leq \frac{q^n-1}{q-1},
\end{equation}
\begin{equation}\label{eq:pesivett}
    N_1+N_2(q+1)+\ldots+N_n(q^{n-1}+\ldots+q+1)=q^{n-1}+\ldots+q+1.
\end{equation}
When $|L_U|$ satisfies equality in \eqref{eq:card} $L_U$ is called \textbf{scattered}, or equivalently, if all the points of $L_U$ have weight one. It is easy to see that $L_U$ is scattered if and only if $U$ is $(1,1)_q$-evasive subspace, that is a $1$-scattered subspace without the assumption that it spans the entire space.

\begin{remark}\label{rk:exv}
Let $L_{U_1},\ldots,L_{U_t}$ be $t$ $\fq$-linear sets in $\PG(k-1,q^m)$ of rank $m$ and with $t<q$.
Then there exists a point $P \in \PG(k-1,q^m)\setminus (L_{U_1}\cup\ldots\cup L_{U_t})$.
This immediately follows from \eqref{eq:card}, indeed
\[ |L_{U_1}\cup\ldots\cup L_{U_t}|\leq t\frac{q^m-1}{q-1}\leq q^m-1, \]
and $q^m-1$ is less than the number of the points in $\PG(k-1,q^m)$.
As a consequence, if $U_1,\ldots,U_t$ are $t$ $\fq$-subspaces in $V=V(k,q^m)$ of dimension $m$ and with $t<q$, then there exists $v\in V\setminus\{0\}$ such that $U_i\cap \langle v \rangle_{\F_{q^m}}=\{0\}$ for every $i \in \{1,\ldots,t\}$.
\end{remark}

We refer to \cite{lavrauw2015field} and \cite{polverino2010linear} for comprehensive references on linear sets and their applications.

\section{First properties of subspace designs}\label{sec:structures}

In this section we will provide some general properties and examples of subspace designs.

\begin{proposition}\label{prop:s<A}
If $(U_1,U_2,\ldots,U_t)$ is an $(s,A)_q$-subspace design in $V=V(k,q^m)$ 
then $s\leq A$.
\end{proposition}
\begin{proof}
Let $W=\langle U_1,\ldots,U_t \rangle_{\F_{q^m}}$. Since $\dim_{\F_{q^m}}(W) \geq s$, there exist 
\[\{u_{i,j} \colon i\in\{1,\ldots,t\} \mbox{ and }j\in\{1,\ldots,j_i\}\}\subseteq W\] 
such that $u_{i,j} \in U_i$, for $i\in\{1,\ldots,t\}$, $u_{i,j}$'s are $\F_{q^m}$-linearly independent and $j_1+\ldots+j_t\geq s$. This implies that $\dim_{\F_q}(U_i \cap W) \geq j_i$ for any $i$ and so $A\geq \sum_{i=1}^t\dim_{\F_q}(U_i \cap W) \geq \sum_{i=1}^t j_i=s$.
\end{proof}

An $(s,A)_q$-subspace design of $V=V(k,q^m)$ is also an $(i,A')_q$-subspace design for any $i\leq s$ and some integer $A'\leq A$.

\begin{proposition}\label{prop:diminuzione}
If $(U_1,\ldots,U_t)$ is an $(s,A)_q$-subspace design in $V=V(k,q^m)$, 
then it is also an $(s-s',A-s')_q$-subspace design in $V$, for any $s'\in \{0,\ldots,s-1\}$. 
In particular, for $s>1$, an $s$-design is also an $i$-design for any $i \leq s$.
\end{proposition}
\begin{proof}
Let choose $s'=1$.
By contradiction, suppose that there exists an $\F_{q^m}$-subspace $H$ having dimension $s-1$ of $V$ such that $\sum_{i=1}^t \dim_{\fq}(U_i \cap H) \geq A$. Since $\dim_{\F_{q^m}} (\langle U_1,\ldots,U_t \rangle_{\F_{q^m}}) \geq s$, $U_1 \cup \cdots \cup U_t$ is not contained in $H$. Hence, there exists $u \in (U_1 \cup \cdots \cup U_t) \setminus H$ so that
\[
\sum_{i=1}^t \dim_{\fq}(U_i \cap \langle {u}, H \rangle_{\F_{q^m}})\geq
\sum_{i=1}^t(\dim_{\fq}(U_i \cap H) +\dim_{\fq}(U_i \cap \langle {u}\rangle_{\F_{q^m}})) \geq A+1,
\]
a contradiction.
The assertion for a general $s'$ follows by repeating the previous argument $s'$ times.
\end{proof}

As a direct consequence, the elements of an $(s,A)_q$-subspace design are related to $(s-s',A')_q$-evasive subspaces for any $s'$ as follows.

\begin{corollary}\label{cor:fromdestoevasive}
If $(U_1,\ldots,U_t)$ is an $(s,A)_q$-subspace design in $V=V(k,q^m)$, 
then for every $\F_{q^m}$-subspace $H$ having dimension $s-s'$ and for any $i \in[t]$ 
\[ \dim_{\fq}(H\cap U_i)\leq A-s'. \]
In particular, if $\dim_{\F_{q^m}}(\langle U_i \rangle_{\F_{q^m}})\geq s-s'$, then $U_i$ is an $(s-s',A-s')_q$-evasive subspace.
\end{corollary}

In the next proposition, we show how to extend the subspaces of an $(s,A)_q$-subspace design and how this reflects on the parameter $A$.

\begin{proposition}
Suppose that there exists an $(s,A)_q$-subspace design $(U_1,\ldots, U_t)$ where $\dim_{\fq}(U_i)=n_i$ for every $i\in \{1,\ldots,t\}$ in $V=V(k,q^m)$. Let $j_1,\ldots,j_t \in \mathbb{N}$ such that $0 \leq j_i \leq mk-n_i$ for every $i\in \{1,\ldots,t\}$, and $\sum_{i=1}^t j_i=s'$. 
Then there exists an $(s,A+s')_q$-subspace design $(U_1',\ldots,U'_t)$ in $V$ such that $\dim_{\fq}(U'_i)=n_i+j_i$, for every $i\in [t]$.
\end{proposition}
\begin{proof}
Without loss of generality, suppose that $mk-n_1>0$. Then there exists ${w} \in V \setminus U_{1}$. Replacing $U_{1}$ by $U'_1=\langle U_1, {w} \rangle_{\fq}$, we obtain that $(U'_1,U_2,\ldots,U_t)$ is an $(s,A+1)_q$-subspace design in $V$. 
Indeed, suppose for the contrary  that there exist ${u}_{i,1},\ldots {u}_{i,\ell_i} \in U_i$ such that $\ell_1+\ldots+\ell_t=A+2$ and the sets
\[
D_1=\{{w}+{u}_{1,1},\ldots,{w}+{u}_{1,\ell_1} \}  \subseteq U'_1
\]
and
\[
D_i=\{{u}_{i,1},\ldots,{u}_{i,\ell_i}  \} \subseteq U_i,\,\,\, 2\leq i \leq t
\]
are sets of $\F_q$-linearly independent elements contained in the same $s$-dimensional $\F_{q^m}$-subspace $H$ of $V$, for every $i$. Note that $\ell_1>0$, otherwise \[\sum_{i=1}^t\dim_{\fq}(U_i \cap H)= \sum_{i=2}^t\dim_{\fq}(U_i \cap H) \geq A+2> A,\] that is a contradiction to  $(U_2,\ldots,U_t)$ being an $(s,A)_q$-subspace design.
Then $D'_1=\{{u}_{1,1}- {u}_{1,\ell_1},\ldots,  {u}_{1,\ell_1-1}- {u}_{1,\ell_1}\}\subseteq U_1$ is a set of $\ell_1-1$ $\F_q$-linearly independent elements of $U_1$ in $H$ and $\sum_{i=1}^t\dim_{\fq}(U_i \cap H) \geq A+1$, a contradiction. The result follows by repeating the previous argument $s'$ times.
\end{proof}

We can construct a subspace design by extending a subspace design lying in a hyperplane of the entire space.

\begin{proposition}
Let $(U_1,\ldots,U_t)$ be a $(k-1,A)_q$-subspace design in $V=V(k,q^m)$ with $\dim_{\fq}(U_i)=n_i$. Let $\ell$ be a positive integer such that $\sum_{i=1}^t n_i-A \leq \ell \leq m$, then there exists a $(k,A+\ell)_q$-subspace design in $V'=V(k+1,q^m)$ with $\dim_{\fq} (U_i)=n_i+e_i$ for every $i$, such that $\sum_{i=1}^t e_i=\ell$.
\end{proposition}

Another way to construct subspace designs is via the direct sum of subspace designs.

\begin{theorem}\label{th:directsum}
Let $V=V_1 \oplus V_2$ where $V_i=V(k_i,q^m)$ and $V=V(k,q^m)$. If $(U'_1,\ldots,U'_t)$ is an $(s,A_1)_q$-subspace design in $V_1$ and $(U''_1,\ldots,U''_t)$ is an $(s,A_2)_q$-subspace design in $V_2$, then $(U_1,\ldots,U_t)$, where $U_i=U'_i \oplus U''_i$ for every $i \in [t]$, is an $(s,A_1+A_2-s)_q$-subspace design in $V$ with $\dim_{\F_q}(U_i)=\dim_{\F_q}(U_i')+\dim_{\F_q}(U_i'')$ for every $i$. 
\end{theorem}
\begin{proof}
It is easy to see that
\[ \dim_{\F_{q^m}} (\langle U_1,\ldots U_t \rangle_{\F_{q^m}})\geq s. \]
Now, suppose that there exists an $s$-dimensional $\F_{q^m}$-subspace $W$ of $V$ such that 
\begin{equation} \label{eq:contradictdess}
\sum_{i=1}^t \dim_{\fq}(W\cap U_i) \geq A_1+A_2-s+1.
\end{equation}
Note that by Proposition \ref{prop:s<A}, $A_1+A_2-s+1>A_1$ and $A_1+A_2-s+1>A_2$.
Clearly, $W$ cannot be contained in $V_1$ since $(U'_1,\ldots,U'_t)$ is an $(s,A_1)_q$-subspace design in $V_1$. 
Let $W_1:=W \cap V_1$ and $h:=\dim_{\F_{q^m}} (W_1)$. Then $h < s$ and by Proposition \ref{prop:diminuzione}, the ordered set of $\F_q$-subspaces $(U_1',\ldots,U'_t)$ is a $(h,A_1-s+h)_q$-subspace design in $V_1$. 
Let denote $\langle U'_i,W \cap U_i \rangle_{\fq}$ by $\overline{U}_i$, then the Grassmann's formula and \eqref{eq:contradictdess} imply
\begin{equation} \label{eq:differencedims}
    \sum_{i=1}^t \dim_{\fq}(\overline{U}_i)-\sum_{i=1}^t \dim_{\fq}(U'_i) \geq A_1+A_2-s+1-(A_1-s+h)=A_2+1-h.
\end{equation}
Consider the quotient space $V/V_1$ (which is isomorphic to $V_2$) and consider the subspace $T:=W+V_1$ of $V/V_1$. Then $\dim_{\F_{q^m}}(T)=s-h$ and $T$ contains the $\F_q$-subspaces
\[
M_i:=\overline{U}_i+V_1,
\]
for every $i \in [t]$.
Since $M_i$ is also contained in the $\F_q$-subspace $U_i+V_1=U''_i+V_1$ for any $i$, the ordered set $(M_1,\ldots,M_t)$ is an $(s,A_2)_q$-subspace design in $V/V_1$ and hence by Proposition \ref{prop:diminuzione} $(M_1,\ldots,M_t)$ is also an $(s-h,A_2-h)_q$-subspace design in $V/V_1$. \\
On the other hand, \\
\[
\sum_{i=1}^t\dim_{\fq}(M_i \cap T) =\sum_{i=1}^t \dim_{\fq}(M_i)=\sum_{i=1}^t \dim_{\fq}(\overline{U}_i)-\sum_{i=1}^t \dim_{\fq} (\overline{U}_i \cap V_1) 
\]
\[
\geq \sum_{i=1}^t \dim_{\fq}(\overline{U}_i)-\sum_{i=1}^t \dim_{\fq} (U_i \cap V_1) =\sum_{i=1}^t \dim_{\fq}(\overline{U}_i)-\sum_{i=1}^t \dim_{\fq} (U'_i),
\]
and hence, \eqref{eq:differencedims} implies
\[
\sum_{i=1}^t\dim_{\fq}(M_i \cap T) \geq A_2-h+1,
\]
a contradiction to the fact that $(M_1,\ldots,M_t)$ is also an $(s-h,A_2-h)_q$-subspace design in $V/V_1$.
\end{proof}

The above result can be generalized as follows.

\begin{corollary}\label{cor:directsum}
Let $V=V_1 \oplus \cdots \oplus V_\ell$ where $V_i=V(k_i,q^m)$ and $V=V(k,q^m)$. If $(U_{i,1},\ldots,U_{i,t})$ is an $(s,A_i)$-subspace design in $V_i$, for every $i\in [\ell]$, then $(U_1,\ldots,U_t)$, where $U_i=U_{1,i} \oplus \cdots \oplus U_{\ell,i}$, is an $(s,\sum_{i=1}^{\ell} A_i-(\ell-1)s)_q$-subspace design in $V$ with $\dim_{\F_q}(U_i)=\sum_{j=1}^{\ell}\dim_{\F_q}(U_{j,i})$.
\end{corollary}

As a consequence of Theorem \ref{th:bounddimension}, we obtain the following bounds on the dimensions of the subspaces of a subspace design.

\begin{corollary}\label{cor:trivsubdes}
Suppose that $(U_1,\ldots,U_t)$ is an $(s,A)_q$-subspace design in $V=V(k,q^m)$.
Let $n_i=\dim_{\fq}(U_i)$ for every $i \in [t]$ and suppose that $\dim_{\F_{q^m}}(\langle U_i \rangle_{\F_{q^m}} )\geq s$ for every $i$.
\begin{enumerate}
    \item If $s=k-1$ and $A<(k-1)m$, then $n_i\leq m+A-1$ for any $i$, and hence \[n_1+\ldots+n_t\leq t(m+A-1);\]
    \item If $s=k-1$ and $A<k-2+m/(k-1)$, then $n_i\leq m+A-1-k$ and hence
    \[n_1+\ldots+n_t\leq t(m+A-1-k);\]
    \item If $A<m$, then $n_i\leq mk-\frac{mks}{A+1}$ and hence
    \[n_1+\ldots+n_t\leq  t\left(mk-\frac{mks}{A+1}\right);\]
    \item More generally, $n_i \leq \frac{(q^A-1)(q^{km}-1)}{q^{sm}-1}+1$ and hence
    \[ n_1+\ldots+n_t\leq t\left( \frac{(q^A-1)(q^{km}-1)}{q^{sm}-1}+1\right). \]
\end{enumerate}
\end{corollary}

However, some of the bounds provided in the above corollary still hold when considering a subspace design with no assumptions on the $\F_{q^m}$-span of its elements.
The proof follows the same arguments of \cite[Theorem 2.3]{csajbok2021generalising} (see also \cite[Corollary 4.9]{bartoli2021evasive}) with the aid of Corollary \ref{cor:fromdestoevasive}.

\begin{theorem} \label{th:boundelementsdesign}
Suppose that $(U_1,\ldots,U_t)$ is an $s$-design in $V=V(k,q^m)$.
Then for every $i \in [t]$
\[ \dim_{\fq}(U_i)\leq \begin{cases}
k, & \mbox{ if } m < s+1,\\
\frac{mk}{s+1}, & \mbox{ if } m \geq s+1.
\end{cases} \]
\end{theorem}
\begin{proof}
Let $U$ be in $\{U_1,\ldots,U_t\}$ and denote by $n$ its dimension.
Since $(U_1,\ldots,U_t)$ is an $s$-design, then 
\[ \dim_{\mathbb{F}_q}(U_i \cap W)\leq s,\]
for every $s$-dimensional $\F_{q^m}$-subspace $W$ of $V$.
If $\langle U \rangle_{\mathbb{F}_{q^m}}=V$ then $U$ is an $s$-scattered $\fq$-subspace and the result follows from 1. and 2. of Theorem \ref{th:bounddimension}. Suppose that $\langle U \rangle_{\mathbb{F}_{q^m}}\ne V$, then we can consider $U$ as an $s$-scattered subspace in $\langle U \rangle_{\mathbb{F}_{q^m}}$ and apply again Theorem \ref{th:bounddimension} to get the assertion.
\end{proof}

\begin{remark}
The above theorem extends \cite[Theorem 2.3]{csajbok2021generalising} to any subspace which is part of an $s$-design. This condition is weaker than the one of being $s$-scattered since we do not require that the subspace spans over $\F_{q^m}$ the entire space or an $\F_{q^m}$-subspace of dimension $s$ as done in \cite{bartoli2021evasive}.
\end{remark}

Subspaces attaining the bound of the previous theorem have the property that their span over $\F_{q^m}$ coincides with the ambient space.

\begin{proposition} \label{prop:spanentirespace}
Let $s$ be a positive integer with $s<k$. Let $U$ be an $\F_q$-subspace of $V=V(k,q^m)$ that 
\[
\dim_{\F_q}(U \cap W) \leq s,
\]
for every $s$-dimensional $\F_{q^m}$-subspace $W$ of $V$. Suppose that
\[ \dim_{\fq}(U)= \begin{cases}
k, & \mbox{ if } m < s+1,\\
\left\lfloor \frac{mk}{s+1} \right\rfloor, & \mbox{ if } m \geq s+1.
\end{cases} \]Then $\langle U \rangle_{\F_{q^m}}=V$. 
\end{proposition}
\begin{proof}
Assume that $m<s+1$ and $\dim_{\F_q}(U)=k$. Since $k>s$, first we note that $U$ is not contained in any $s$-dimensional $\F_{q^m}$-subspace of $V$. Now, suppose by contradiction that $\langle U\rangle_{\F_{q^m}}=T$, where $T$ is an $\F_{q^m}$-subspace of $V$ of dimension $k'<k$. Then by Theorem \ref{th:bounddimension}, since $U$ is an $(s,s)_q$-evasive subspace contained in $T$ and $m<s+1$, we get
\[
k=\dim_{\F_q}(U) \leq k',
\]
a contradiction as $k'<k$, so this implies that $\langle U\rangle_{\F_{q^m}}=V$. Assume now that $m\geq s+1$ and $\dim_{\F_q}(U)=\left\lfloor\frac{mk}{s+1}\right\rfloor$. Since $m\geq s+1$, we get $\left\lfloor\frac{mk}{s+1}\right\rfloor>s$ and we get that $U$ is not contained in any $s$-dimensional $\F_{q^m}$-subspace of $V$.
Suppose by contradiction that $\langle U\rangle_{\F_{q^m}}=T$, where $T$ is an $\F_{q^m}$-subspace of $V$ of dimension $k'<k$. Then by Theorem \ref{th:bounddimension}, since $U$ is an $(s,s)_q$-evasive subspace contained in $T$ and $m \geq s+1$, we get
\[
\frac{mk-s}{s+1}\leq \left\lfloor \frac{mk}{s+1} \right\rfloor=\dim_{\F_q}(U) \leq \frac{mk'}{s+1} \Rightarrow s \geq m(k-k')\geq m,
\]
a contradiction to the fact that $s<m$.
\end{proof}

Applying the bound of Theorem \ref{th:boundelementsdesign} to the elements of an $s$-design we obtain the following.

\begin{corollary}\label{cor:trivhdes}
Suppose that $(U_1,\ldots,U_t)$ is an $s$-design in $V=V(k,q^m)$.
Let $n_i=\dim_{\fq}(U_i)$ for every $i \in [t]$.
If $m\geq s+1$ then
\begin{equation} \label{eq:trivhdesignsub}
 n_1+\ldots+n_t\leq t\frac{mk}{s+1}. 
 \end{equation}
\end{corollary}

As we will see in the next sections, the bounds provided in Corollary \ref{cor:trivsubdes} for subspace designs are not sharp in general and we will see how to improve them (at least in some cases).
Whereas, for the bound of Corollary \ref{cor:trivhdes} we will provide some constructions when $m\geq s+1$ satisfying equality, i.e. we will construct examples of maximum $s$-design.

\begin{definition}
Let $(U_1,\ldots,U_t)$ be an $s$-design in $V=V(k,q^m)$ such that $m\geq s+1$. If
\begin{equation*} 
 \sum_{i=1}^t \dim_{\fq}(U_i)= t\frac{mk}{s+1},
 \end{equation*}
then $(U_1,\ldots,U_t)$ will be called a \textbf{maximum $s$-design}.
\end{definition}

Since each subspace of a maximum $s$-design has dimension $\frac{mk}{s+1}$, from Proposition \ref{prop:spanentirespace} we get the following.

\begin{proposition}
Every maximum $s$-design is non-degenerate.
\end{proposition}

\section{Maximum \texorpdfstring{$s$}{Lg}-designs}\label{sec:maxsdes}

This section is mainly devoted to the study of Problem \ref{prob:A} under certain assumptions. Indeed,
we investigate maximum $s$-designs, i.e. $s$-designs attaining equality in Corollary \ref{cor:trivsubdes}. 
Their interest is also related to the fact they can be seen as an extension of $s$-scattered subspaces in the framework of subspace designs.

\subsection{Constructions}

We start by constructing examples of $(k-1)$-designs by partitioning an $\F_{q^m}$-basis of the ambient space.

\begin{proposition} \label{prop:partitionofbasis}
Let $\alpha_1,\ldots,\alpha_k$ be an $\F_{q^m}$-basis of $V=V(k,q^m)$. 
Let consider $\mathcal{P}=\{X_1,\ldots,X_t\}$ a partition of the set $\{\alpha_1,\ldots,\alpha_k\}$, with $t \geq 2$. Let $U_i$ be the $\F_q$-subspace of $V$ generated by the elements of $X_i$ over $\F_q$, for any $i \in [t]$. Then $(U_1,\ldots,U_t)$ is a $(k-1)$-design. 
\end{proposition}
\begin{proof}
Let $W \subseteq V$ a $(k-1)$-dimensional $\F_{q^m}$-subspace and assume that $\sum_{i=1}^t \dim_{\fq}(U_i \cap W) \geq k$. Since
$$k \leq \sum_{i=1}^t \dim_{\fq}(U_i \cap W) \leq \sum_{i=1}^t\dim_{\fq}(U_i)=k,$$
it follows that $\dim_{\fq}(U_i \cap W)=\dim_{\fq}(U_i)$ and then $U_i \subseteq W$ for every $i$, that is $\alpha_1,\ldots,\alpha_k \in W$, a contradiction.
\end{proof}

A more involved construction is the following, which turns out to be also maximum.

\begin{theorem}\label{th:example2(n-1)des}
     Let $t<q$, $\alpha_1,\ldots,\alpha_t \in \F_{q^m}^*$ and let $\N_{q^m/q}(\alpha_i)=\lambda_i$, for any $i$. Suppose that $\lambda_i\neq \lambda_j$ for every $i \neq j$. Let $\eta \in \F_{q^m}$ be such that $\N_{q^m/q}(\eta)(-1)^{km} \notin G$, with $G$ the multiplicative subgroup of $\F_{q}^*$ generated by $\{\lambda_1,\ldots,\lambda_t\}$. Let $S_i$ be an $n_i$-dimensional $\F_q$-subspace of $\F_{q^m}$ for any $i \in [t]$. If $k\leq \sum_i n_i$ then $(U_1,\ldots,U_t)$ is a $(k-1)$-design of $\F_{q^m}^k$, where
    \[
    U_i=\{ (x+\eta \N_{\sigma}^k(\alpha_i)\sigma^k(x),\sigma(x)\N_{\sigma}^1(\alpha_i),\sigma^2(x)N_{\sigma}^2(\alpha_i),\ldots, \sigma^{k-1}(x)\N_{\sigma}^{k-1}(\alpha_i)) \colon x \in S_i\}
    \]
    and $\dim_{\fq}(U_i)=n_i$ for every $i \in [t]$.
    If $n_i=m$ for any $i$ and $k\leq m$ then it is a maximum $(k-1)$-design.
\end{theorem}
\begin{proof}
We divide the proof into two steps.\\
\textbf{Step 1:} We prove that $ \sum_{i=1}^t \dim_{\F_q}(H \cap U_i)\leq k-1$, for each $\F_{q^m}$-hyperplane $H$ of $\F_{q^m}^k$.\\ 
Let $H$ be such a hyperplane and suppose that $H$ is defined by the equation $h_0x_0+\ldots+h_{k-1}x_{k-1}=0$. Let $F^H(x)=h_0 x+h_1 \sigma(x)+\ldots+h_{k-1}\sigma^{k-1}(x)+h_0 \eta \sigma^k(x) \in \Lm$. Observe that for any $i$ we have
\begin{align*}
    H \cap U_i = &\{(x+\eta \N_{\sigma}^k(\alpha_i)\sigma^k(x),\sigma(x)\N_{\sigma}^1(\alpha_i),\sigma^2(x)N_{\sigma}^2(\alpha_i),\ldots, \sigma^{k-1}(x)\N_{\sigma}^{k-1}(\alpha_i)) \colon \\ & \hskip 0.5 cm h_0(x+\eta\N_{\sigma}^k(\alpha_i)\sigma^k(x))+\ldots+h_{k-1}\sigma^{k-1}(x)\N_{\sigma}^{k-1}(\alpha_i)=0, x \in S_i\}. 
\end{align*}
In particular, $\dim_{\fq}(H \cap U_i)=\dim_{\fq}(\ker(F^H_{\alpha_i}(x)) \cap S_i)$.
Hence, 
\begin{align*}
    \sum_{i=1}^t \dim_{\F_q}(H \cap U_i)& =\sum_{i=1}^t\dim_{\F_q}(\ker(F^H_{\alpha_i}(x)) \cap S_i)  \leq \deg_{\sigma}(F^H(x)),
\end{align*}
where the last inequality follows by (2) of Theorem \ref{th:boundker}. If $h_0=0$ then $\deg_{\sigma}(F^H_{\alpha_i}(x)) \leq k-1$ and we have the desired inequality. Now, suppose that $h_0 \neq 0$ and so $\deg_{\sigma}(F^H_{\alpha_i}(x))=k$. If $\sum_{i=1}^t \dim_{\F_q}(\ker(F^H_{\alpha_i}(x)))= k$, then by (3) of Theorem \ref{th:boundker} we have
\[
\N_{q^m/q}(\eta)=(-1)^{km} \prod_{i=1}^{t} \lambda_i^{d_{\lambda_i}(F^H(x))},
\]
a contradiction on the assumptions on $\eta$.\\
\textbf{Step 2:} We show that $\langle U_1,\ldots,U_t \rangle_{\F_{q^m}}=\F_{q^m}^k$. Let $\{\beta_{i,1},\ldots,\beta_{i,n_i}\}$ be an $\F_q$-basis of $S_i$ for every $i \in [t]$. Then 
\[
Z_i=\{(\beta_{i,j}+\eta \N_{\sigma}^k(\alpha_i)\sigma^k(\beta_{i,j}), \sigma(\beta_{i,j})\N_{\sigma}^1(\alpha_i) , \ldots, \sigma^{k-1}(\beta_{i,j})\N_{\sigma}^{k-1}(\alpha_i)) \colon j \in [n_i]\}
\]
is an $\F_q$-basis of $U_i$, for every $i \in [t]$ and so $\langle U_1,\ldots,U_t \rangle_{\F_{q^m}}=\langle Z_1,\ldots,Z_t\rangle_{\F_{q^m}}$. Then we have the assertion if the row span (over $\F_{q^m}$), or equivalently, the column span (over $\F_{q^m}$) of the following matrix has dimension over $\F_{q^m}$ equal to $k$
\[
Z=\left(
\begin{array}{c c c c }
\beta_{1,1}+\eta \N_{\sigma}^k(\alpha_1)\sigma^k(\beta_{1,1})& \sigma(\beta_{1,1})\N_{\sigma}^1(\alpha_1) & \cdots& \sigma^{k-1}(\beta_{1,1})\N_{\sigma}^{k-1}(\alpha_1)\\
\beta_{1,2}+\eta \N_{\sigma}^k(\alpha_1)\sigma^k(\beta_{1,2})& \sigma(\beta_{1,2})\N_{\sigma}^1(\alpha_1) & \cdots& \sigma^{k-1}(\beta_{1,2})\N_{\sigma}^{k-1}(\alpha_1)\\
\vdots & \vdots & \ddots & \vdots \\
\beta_{1,n_1}+\eta \N_{\sigma}^k(\alpha_1)\sigma^k(\beta_{1,n_1})& \sigma(\beta_{1,n_1})\N_{\sigma}^1(\alpha_1) & \cdots& \sigma^{k-1}(\beta_{1,n_1})\N_{\sigma}^{k-1}(\alpha_1) \\
\hline \\
\vdots & \vdots & \ddots & \vdots \\
\hline \\
\beta_{t,1}+\eta \N_{\sigma}^k(\alpha_t)\sigma^k(\beta_{t,1})& \sigma(\beta_{t,1})\N_{\sigma}^1(\alpha_t) & \cdots& \sigma^{k-1}(\beta_{t,1})\N_{\sigma}^{k-1}(\alpha_t)\\
\beta_{t,2}+\eta \N_{\sigma}^k(\alpha_t)\sigma^k(\beta_{t,2})& \sigma(\beta_{t,2})\N_{\sigma}^1(\alpha_t) & \cdots& \sigma^{k-1}(\beta_{t,2})\N_{\sigma}^{k-1}(\alpha_t)\\
\vdots & \vdots & \ddots & \vdots \\
\beta_{t,n_t}+\eta \N_{\sigma}^k(\alpha_t)\sigma^k(\beta_{t,n_t})& \sigma(\beta_{t,n_t})\N_{\sigma}^1(\alpha_t) & \cdots& \sigma^{k-1}(\beta_{t,n_t})\N_{\sigma}^{k-1}(\alpha_t)
\end{array}
\right).
\]

Suppose there exists an $\F_{q^m}$-linear combination of the columns of the matrix $Z$ which equals the zero vector, with coefficients $h_0,h_1,\ldots,h_{k-1}\in \F_{q^m}$ and not all of them are zero. Let $F'(x)=h_0x+h_1x^{\sigma}+\ldots+h_{k-1}x^{\sigma^{k-1}}+h_{0}\eta x^{\sigma^{k}}$. This means that $F'_{\alpha_i}(\beta_{i,j})=0$ for every $i \in [t]$ and $j \in [n_i]$. This implies that $S_i \subseteq \ker(F'_{\alpha_i}(x))$ for every $i$ and so $n_i \leq \dim_{\F_q}( \ker(F'_{\alpha_i}(x)))$. Since $h_i$'s are not all zero, $F'(x)$ is not the zero polynomial and by (2) of Theorem \ref{th:boundker}, we have 
\[
\sum_{i=1}^{t}n_i \leq \deg_{\sigma}(F'(x)). 
\]
If $\deg_{\sigma}(F'(x))\leq k-1$, we have a contradiction on the assumptions on $k$. If $\deg_{\sigma}(F'(x))=k$, then $\sum_{i=1}^{t}n_i=k$, and by (3) of Theorem \ref{th:boundker} we have
\[
\N_{q^m/q}(\eta)=(-1)^{km} \prod_{i=1}^{t} \lambda_i^{d_{\lambda_i}(F'(x))},
\]
a contradiction on the assumptions on $\eta$. This means that $h_0=h_1=\ldots=h_{k-1}=0$ and so the columns of $Z$ are $\F_{q^m}$-linearly independent.
\end{proof}

\begin{remark}
The matrix $Z$ which appears in Step 2 of the above proof when $\eta=0$, $t=1$ and $\alpha_1=1$ coincides with a \emph{Moore matrix}, which is strongly related to MRD codes; see \cite{bartoli2020asymptotics,csajbok2018maximum}. When $\eta\ne 0$, $t=1$ and $\alpha_1=1$, the matrix $Z$ still has similar properties of the Moore matrix; see \cite[Section 3]{bartoli2021linear}.
\end{remark}

\begin{remark}
The condition $(-1)^{km}\N_{q^m/q}(\eta) \notin G$ in Theorem \ref{th:example2(n-1)des} produces a restriction on $t$. Indeed, since the multiplicative group $\F_{q}^*$ is a cyclic group then there exists a subgroup of order $(q-1)/r$, for every positive divisor $r$ of $q-1$. For this reason, we can construct subspace designs with $\eta \neq 0$ if $t$ is not larger than the size of the biggest proper subgroup of $\F_{q}^*$, that is $(q-1)/r'$, where $r'$ is the smallest prime diving $q-1$.
If $G=\F_{q}^*$, the only possibly choice for $\eta$ is zero.
\end{remark}

\begin{remark} As seen in the above remark the maximum number of subspaces in the construction of Theorem \ref{th:example2(n-1)des} is obtained when $\eta=0$. For this reason we explicitly describe them.
Let $t<q$ be an integer. Let $\alpha_1,\ldots,\alpha_t \in \F_{q^m}^*$ and let $\N_{q^m/q}(\alpha_i)=\lambda_i$, and suppose that $\lambda_i\neq \lambda_j$ if $i \neq j$. Let $S_i$ be an $n_i$-dimensional $\F_q$-subspace of $\F_{q^m}$ for any $i \in [t]$. By Theorem \ref{th:example2(n-1)des} if $k\leq \sum_i n_i$ then $(U_1,\ldots,U_t)$, where
\[
    U_i=\{ (x,\sigma(x)\N_{\sigma}^1(\alpha_i),\sigma^2(x)N_{\sigma}^2(\alpha_i),\ldots, \sigma^{k-1}(x)\N_{\sigma}^{k-1}(\alpha_i)) \colon x \in S_i\},
    \]
    is a $(k-1)$-design of $\F_{q^m}^k$ and $\dim_{\fq}(U_i)=n_i$ for every $i \in [t]$. If $n_i=m$ for any $i$ and $k \leq m$, then it is a maximum $(k-1)$-design.
\end{remark}

\begin{remark}
The above example is the geometric description of the MSRD codes found by Neri in \cite{neri2021twisted} known as linearized twisted Reed-Solomon codes, which extends the family of linearized Reed-Solomon codes found by Mart{\'\i}nez-Pe{\~n}as in \cite{Martinez2018skew}.
\end{remark}

As a consequence of Theorem \ref{th:directsum} and Corollary \ref{cor:directsum}, we obtain a way to construct $s$-designs via the direct sum of subspaces.

\begin{theorem}\label{th:directsumdesigns}
Let $V=V_1 \oplus \cdots \oplus V_\ell$ where $V_i=V(k_i,q^m)$ and $V=V(k,q^m)$. If $(U_{i,1},\ldots,U_{i,t})$ is an $s$-design in $V_i$, for every $i\in [\ell]$, then $(U_1,\ldots,U_t)$, where $U_j=U_{1,j} \oplus \cdots \oplus U_{\ell,j}$, for $j \in [t]$, is an $s$-design in $V$.
Moreover, if $(U_{i,1},\ldots,U_{i,t})$ is a maximum $s$-design in $V_i$ for every $i$, then $(U_1,\ldots,U_t)$ is a maximum $s$-design in $V$.
\end{theorem}

As a consequence of Theorems \ref{th:example2(n-1)des} and \ref{th:directsumdesigns}, we prove the existence of a maximum $s$-subspace design in $V(k,q^m)$ when $s+1 \mid k$ and $t\leq q-1$.

\begin{theorem} \label{th:incollamento}
    If $s+1$ divides $k$, $m \geq s+1$ and $t\leq q-1$, then there exists a maximum $s$-design in $V=V(k,q^m)$.
\end{theorem}
\begin{proof}
Suppose that $k=\ell(s+1)$ and let $V=V_1\oplus \ldots\oplus V_{\ell}$ in such a way that $\dim_{\F_{q^m}}(V_i)=s+1$ for every $i \in [\ell]$.
In each of this $V_i$ consider a maximum $s$-design $U_{i,1},\ldots,U_{i,t}$ whose existence is guaranteed by Theorem \ref{th:example2(n-1)des} (all the subspaces involved have dimension $m$).
Theorem \ref{th:directsumdesigns} implies that $(U_1,\ldots,U_t)$, where $U_j=U_{1,j} \oplus \cdots \oplus U_{\ell,j}$, for $j \in [t]$, is an $s$-design in $V$.
Moreover,
\[ \dim_{\fq}(U_i)=lm=\frac{km}{s+1}, \]
for every $i$ and so $(U_1,\ldots,U_t)$ is a maximum $s$-design in $V$.
\end{proof}

In the next section we will construct examples in which $s+1$ does not divide $k$.

\subsection{Intersection with hyperplanes}

We give some lower and upper bounds on the sum of the dimensions of intersection with hyperplanes for a maximum $s$-design.
Let $(U_1,\ldots,U_t)$ be a maximum $s$-design in $V(k,q^m)$, that is $m \geq s+1$ and $\dim_{\fq}(U_i)=\frac{km}{s+1}$. 

Let $H$ be a hyperplane of $V$. Then, by Grassmann's formula we have 
\[\dim_{\fq}(U_i \cap H) \geq \frac{mk}{s+1}-m,\]
so that 
\[\sum_{i=1}^t \dim_{\fq}(U_i \cap H) \geq t\left(\frac{mk}{s+1}-m\right).\]

Recall now the following result from \cite{blokhuis2000scattered} and \cite{csajbok2021generalising}.

\begin{theorem}(see \cite[Theorem 4.2]{blokhuis2000scattered},\cite[Theorem 2.7]{csajbok2021generalising})\label{th:boundhyperscatt}
If $U$ is an $s$-scattered $\fq$-subspace of $V=V(k,q^m)$ of dimension $km/(s+1)$, then for any $(k-1)$-dimensional $\F_{q^m}$-subspace $W$ of $V$ we have
\[ \frac{mk}{s+1}-m\leq \dim_{\fq}(U\cap W)\leq \frac{mk}{s+1}-m+s. \]
\end{theorem}


\begin{remark}
In \cite{1930-5346_2019_0_108,zini2021scattered}, the distribution of the intersections between an $s$-scattered $\fq$-subspace of dimension $km/(s+1)$ and the hyperplanes of $V$ have been determined.
\end{remark}

Since each element of a maximum $s$-design in $V(k,q^m)$ satisfies the assumptions of Proposition  \ref{prop:spanentirespace}, by applying Theorem \ref{th:boundhyperscatt} we can provide the following bounds on the dimension of intersection of an $s$-design with the hyperplanes.

\begin{corollary}\label{cor:boundwait}
Let $(U_1,\ldots,U_t)$ be a maximum $s$-design in $V=V(k,q^m)$. Then for every hyperplane $W$ of $V$ 
\[ t\left(\frac{mk}{s+1}-m\right)\leq \sum_{i=1}^t\dim_{\fq}(U_i\cap W)\leq t\left(\frac{mk}{s+1}-m+s\right). \]
\end{corollary}

However, the above upper bound is too large in general.
We can improve it for the case $s=1$, following the techniques developed in \cite{blokhuis2000scattered}.

\begin{theorem}\label{th:boundhyper1scatt}
Let $(U_1,\ldots,U_t)$ be a maximum $1$-design in $V=V(k,q^m)$. Then 
\[
\frac{tm(k-2)}{2} \leq \sum_{i=1}^t \dim_{\fq}(U_i \cap H) \leq  \frac{tm(k-2)}{2}+1,
\]
for every hyperplane $H$ of $V$.
\end{theorem}
\begin{proof}
By Corollary \ref{cor:boundwait}, we have that
\[
    \frac{tm(k-2)}{2} \leq \sum_{i=1}^t \dim_{\fq}(U_i \cap H) \leq \frac{tm(k-2)}{2}+t,
\]
for every hyperplane $H$ of $V$.

Let $L_{U_i} \subseteq \PG(k-1,q^m)=\PG(V,\F_{q^m})$ be the $\fq$-linear set associated with $U_i$, for $i\in [t]$. Since $(U_1,\ldots,U_t)$ is a $1$-design, it follows that
\[
\sum_{i=1}^t \dim_{\fq}(U_i \cap \langle  {v} \rangle_{\F_{q^m}})=\sum_{i=1}^t w_{L_{U_i}}(P) \leq 1,
\]
for every $P=\langle  {v} \rangle_{\F_{q^m}} \in \PG(k-1,q^m)$. In particular, the $L_{U_i}$'s are pairwise disjoint. Now let $j \in [t]$ and denote by $h_j$ the number of hyperplanes $\mathcal{H}=\PG(H,\F_{q^m})$ of $\PG(k-1,q^m)$ such that 
\[
\sum_{i=1}^t\dim_{\F_q}(U_i \cap H)=\sum_{i=1}^t w_{L_{U_i}}(\mathcal{H}) = \frac{tm(k-2)}{2}+j,
\]
and note that this exactly corresponds to $q^m-1$ multiplied by the number of hyperplanes of $V$ whose sum of the dimensions of the intersections with the $U_i$'s is $\frac{tm(k-2)}{2}+j$.
Our aim is to prove that $h_j=0$ if $j>1$. For this purpose, let give some relations involving the $h_j$'s.
Clearly, 
\begin{equation}\label{eq:numhyper}
\sum_{j=1}^t h_j=\frac{q^{mk}-1}{q^m-1}.
\end{equation}
Now, let double count the following sets 
\[
S_1= \{(\mathcal{H},P)  \colon \mathcal{H} \mbox{ hyperplane of } \PG(k-1,q^m), \ P \in \mathcal{H} \cap (L_{U_1} \cup \ldots \cup L_{U_t})\},
\]
\[
S_2=\{(\mathcal{H},(P_1,P_2)) \colon \mathcal{H} \mbox{ hyperplane of } \PG(k-1,q^m), P_i \in \mathcal{H} \cap (L_{U_1} \cup \ldots \cup L_{U_t}), \ P_1 \neq P_2 \}.
\]
First observe that if a hyperplane $\mathcal{H}$ is such that 
\[
\sum_{i=1}^t w_{L_{U_i}}(\mathcal{H}) = \frac{tm(k-2)}{2}+j, 
\]
then there exist $\ell_1,\ldots, \ell_{t-i} \in \mathbb{N}$ such that \[w_{L_{U_h}}(\mathcal{H})=\frac{tm(k-2)}{2},\,\, \text{for}\,\, h \in \{\ell_1,\ldots,\ell_{t-i}\}\] 
and 
\[w_{L_{U_h}}(\mathcal{H})=\frac{tm(k-2)}{2}+1\,\, \text{for}\,\, h \in [t] \setminus \{\ell_1,\ldots,\ell_{t-i}\}.\]

Then by double counting $S_1$ we have
\begin{equation}\label{eq:counthyper2}
\sum_{j=0}^t h_j \left( (t-j) \frac{q^{\frac{m(k-2)}{2}}-1}{q-1}+j \frac{q^{\frac{m(k-2)}{2}+1}-1}{q-1}
\right)= \left( t\frac{q^\frac{mk}{2}-1}{q-1} \right) \left(\frac{q^{m(k-1)}-1}{q^m-1} \right),
\end{equation}
and by double counting $S_2$ we obtain
\begin{small}
\[
\sum_{i=0}^t h_j \left( (t-j) \frac{q^{\frac{m(k-2)}{2}}-1}{q-1}+j \frac{q^{\frac{m(k-2)}{2}+1}-1}{q-1}
\right) \left( (t-j) \frac{q^{\frac{m(k-2)}{2}}-1}{q-1}+j \frac{q^{\frac{m(k-2)}{2}+1}-1}{q-1}
-1 \right)
\]
\begin{equation}\label{eq:counthyper3}
=\left( t\frac{q^\frac{mk}{2}-1}{q-1} \right) \left( t\frac{q^\frac{mk}{2}-1}{q-1} -1 \right) \left( \frac{q^{m(k-2)}-1}{q^m-1} \right).
\end{equation}\end{small}

Let $\gamma_j= (t-j) \frac{q^{\frac{m(k-2)}{2}}-1}{q-1}+j \frac{q^{\frac{m(k-2)}{2}+1}-1}{q-1}$ for any $j \in [t]$. Note that $\gamma_0<\gamma_1<\gamma_2<\ldots$. Then  \eqref{eq:numhyper}, \eqref{eq:counthyper2} and \eqref{eq:counthyper3} give the following system whose unknowns are the $h_j$'s
\begin{equation} \label{eq:relationhi}
\begin{cases}
\sum_{j=1}^t h_j=\frac{q^{mk}-1}{q^m-1}, \\
\sum_{j=0}^t h_j \gamma_j= \left( t\frac{q^\frac{mk}{2}-1}{q-1} \right) \left(\frac{q^{m(k-1)}-1}{q^m-1} \right), \\
\sum_{j=0}^t h_j \gamma_j(\gamma_j-1)
=\left( t\frac{q^\frac{mk}{2}-1}{q-1} \right) \left( t\frac{q^\frac{mk}{2}-1}{q-1} -1 \right) \left( \frac{q^{m(k-2)}-1}{q^m-1} \right).
\end{cases}
\end{equation}

Let 
\begin{equation} \label{eq:proveA}
 A=\sum_{j=0}^t h_j (\gamma_j-\gamma_0)(\gamma_j-\gamma_1).
 \end{equation}
Since $\gamma_0<\gamma_1<\gamma_2<\ldots$, to get the assertion it is enough to prove that $A=0$.

The coefficient of $h_j$ in \eqref{eq:proveA} can be written as \[\gamma_j(\gamma_j-1)-(\gamma_0+\gamma_1-1)\gamma_j+\gamma_0\gamma_1,\]
so that
\[ A=\sum_{j=0}^t \gamma_j(\gamma_j-1) h_j-(\gamma_0+\gamma_1-1)\sum_{j=0}^t \gamma_j h_j+\gamma_0\gamma_1 \sum_{j=0}^t h_j.  \]

Then, using  \eqref{eq:relationhi}, \eqref{eq:proveA} can be rewritten as follows
\[
A=\left( t\frac{q^\frac{mk}{2}-1}{q-1} \right) \left( t\frac{q^\frac{mk}{2}-1}{q-1} -1 \right) \left( \frac{q^{m(k-2)}-1}{q^m-1} \right)-
\]
\[
-\left( t \frac{q^{\frac{m(k-2)}{2}}-1}{q-1}+(t-1) \frac{q^{\frac{m(k-2)}{2}}-1}{q-1} +  \frac{q^{\frac{m(k-2)}{2}+1}-q}{q-1}
\right)\left( t \frac{q^{\frac{mk}{2}}-1}{q-1} \right)\left( \frac{q^{m(k-1)}-1}{q^m-1} \right)
\]
\[
+\left( t\frac{q^{\frac{m(k-2)}{2}}-1}{q-1}\right) \left((t-1) \frac{q^{\frac{m(k-2)}{2}}-1}{q-1} +  \frac{q^{\frac{m(k-2)}{2}+1}-1}{q-1}
\right) \left( \frac{q^{mk}-1}{q^m-1} \right).
\]

Straightforward computations yield $A=0$. 
Since 
\[ A=\sum_{j=2}^t h_j (\gamma_j-\gamma_0)(\gamma_j-\gamma_1)\]
and $\gamma_j >\gamma_0$ and $\gamma_j>\gamma_0$, for any $j>1$, it follows that $h_j=0$ for every $j>1$.
\end{proof}

We can determine exactly the weight distributions of the hyperplanes with respect to a maximum $1$-design.

\begin{theorem}\label{th:projproperties}
Let $(U_1,\ldots,U_t)$ be a maximum $1$-design in $V=V(k,q^m)$.
For any hyperplane $H$ of $V$ there exists a positive integer $j$ such that
\begin{equation}\label{eq:weighthyper2}
\sum_{i=1}^t \dim_{\fq}(U_i\cap H) = \frac{tm(k-2)}{2}+j.
\end{equation}
For every positive integer $j$, denote by $h_j$ the number of hyperplanes $H$ of $V$ such that \eqref{eq:weighthyper2} holds.
Then
\begin{itemize}
    \item $h_j=0$ for any $j\geq 2$;
    \item $h_0$ and $h_1$ are uniquely determined and 
    \begin{equation}\label{eq:valuehi} \left\{
    \begin{array}{ll}
    h_1=t\frac{q^{\frac{mk}2}-1}{q-1},\\
    h_0=\frac{q^{mk}-1}{q^m-1}-h_1.
    \end{array}
    \right.\end{equation}
    \item $h_1$ is non-zero;
    \item $h_0=0$ if and only if $t=(q-1)\frac{q^{\frac{mk}2}+1}{q^m-1}$.
\end{itemize}
\end{theorem}
\begin{proof}
The fact that the sum of the dimensions of the intersections between the $U_i$'s and any hyperplane has form \eqref{eq:weighthyper2} and the fact that $h_j=0$ for any $j\geq 2$ are consequences of Theorem \ref{th:boundhyper1scatt}.
As pointed out in the proof of Theorem \ref{th:boundhyper1scatt}, $h_0$ and $h_1$ must satisfy System \eqref{eq:relationhi} and, in particular, the first two equations
\begin{equation}\label{eq:systemh0h1}
    \left\{
    \begin{array}{ll}
         h_0+h_1=\frac{q^{mk}-1}{q^m-1}, \\
         \gamma_0 h_0+\gamma_1 h_1= \left( t\frac{q^\frac{mk}{2}-1}{q-1} \right) \left(\frac{q^{m(k-1)}-1}{q^m-1} \right),
    \end{array}
    \right.
\end{equation}
where $\gamma_0=t \frac{q^{\frac{m(k-2)}{2}}-1}{q-1}$ and $\gamma_1=(t-1) \frac{q^{\frac{m(k-2)}{2}}-1}{q-1}+\frac{q^{\frac{m(k-2)}{2}+1}-1}{q-1}$.
Since $\gamma_1>\gamma_0$, the System \eqref{eq:systemh0h1} is determined and hence $h_0$ and $h_1$ are uniquely determined as in \eqref{eq:valuehi}.
If $h_1$ would be zero, then System \eqref{eq:systemh0h1} implies that
\[
\left\{
    \begin{array}{ll}
         h_0=\frac{q^{mk}-1}{q^m-1}, \\
         \gamma_0 h_0= \left( t\frac{q^\frac{mk}{2}-1}{q-1} \right) \left(\frac{q^{m(k-1)}-1}{q^m-1} \right),
    \end{array}
    \right.
\]
that is
\[
q^{m(k-2)/2}+q^{mk}=q^{mk/2}+q^{m(k-1)},
\]
which cannot happen.
Now, suppose that $h_0=0$, then 
\eqref{eq:valuehi} implies that
\[
\left\{
    \begin{array}{ll}
         h_1=\frac{q^{mk}-1}{q^m-1}, \\
         \gamma_1 h_1= \left( t\frac{q^\frac{mk}{2}-1}{q-1} \right) \left(\frac{q^{m(k-1)}-1}{q^m-1} \right),
    \end{array}
    \right.
\]
that is
\[ (t-1)(q^\frac{m(k-2)}2-1)(q^{mk}-1)+(q^{mk}-1)(q^{\frac{m(k-2)}2+1}-1)=t(q^{\frac{mk}2}-1)(q^{m(k-1)}-1), \]
which can be rewritten as follows
\[ q^{\frac{m(k-2)}2}(q^{\frac{mk}2}-1)[(q-1)q^{\frac{mk}2}-tq^m+q+t-1]=0, \]
and it happens if and only if 
\[ t=(q-1)\frac{q^{\frac{mk}2}+1}{q^m-1}. \]
This concludes the proof.
\end{proof}

As a corollary we obtain examples of two-intersection sets with respect to hyperplanes.

\begin{corollary}\label{cor:8.9}
Let $(U_1,\ldots,U_t)$ be a maximum $1$-design in $V=V(k,q^m)$.
Consider the associated $\fq$-linear sets $L_{U_1},\ldots,L_{U_t}$.
Then for any hyperplane $\mathcal{H}$ of $\PG(k-1,q^m)=\PG(V,\F_{q^m})$ we have
\[ \left| \mathcal{H} \cap (L_{U_1}\cup \ldots\cup L_{U_t}) \right| \in \left\{t \frac{q^{\frac{m(k-2)}{2}}-1}{q-1},(t-1) \frac{q^{\frac{m(k-2)}{2}}-1}{q-1}+\frac{q^{\frac{m(k-2)}{2}+1}-1}{q-1}\right\}. \]
In particular, 
\begin{itemize}
\item $L_{U_1}\cup \ldots\cup L_{U_t}$ is a two-intersection set with respect to hyperplanes if and only if $t\ne (q-1)\frac{q^{\frac{mk}2}+1}{q^m-1}$;
\item if $t= (q-1)\frac{q^{\frac{mk}2}+1}{q^m-1}$ then $L_{U_1}\cup \ldots\cup L_{U_t}=\PG(k-1,q^m)$;
\item $t \leq (q-1)\frac{q^{\frac{mk}2}+1}{q^m-1}$.
\end{itemize}
\end{corollary}
\begin{proof}
Since $(U_1,\ldots,U_t)$ is a maximum $1$-design then the $L_{U_i}$'s are scattered and pairwise disjoint. This means that
\[ \left| \mathcal{H} \cap (L_{U_1}\cup \ldots\cup L_{U_t}) \right|= \left| \mathcal{H} \cap L_{U_1}\right|+ \ldots+ \left|\mathcal{H} \cap L_{U_t} \right|. \]
Theorem \ref{th:projproperties} implies that for any hyperplane $\mathcal{H}$ of $\PG(k-1,q^m)$ we have the following two possibilities:
\begin{itemize}
    \item $|L_{U_i}\cap \mathcal{H}|=\frac{q^{\frac{m(k-2)}2}-1}{q-1}$ for any $i \in \{1,\ldots,t\}$;
    \item $|L_{U_i}\cap \mathcal{H}|=\frac{q^{\frac{m(k-2)}2}-1}{q-1}$ for any $i \in \{1,\ldots,t\}$ except for a $j \in \{1,\ldots,t\}$ for which $|L_{U_j}\cap \mathcal{H}|=\frac{q^{\frac{m(k-2)}2+1}-1}{q-1}$.
\end{itemize}
Therefore the first part of the assertion is proved.
Moreover, if $t= (q-1)\frac{q^{\frac{mk}2}+1}{q^m-1}$ then
\[  \left| L_{U_1}\cup  \ldots \cup L_{U_t} \right|= \left|  L_{U_1}\right|+ \ldots+ \left| L_{U_t} \right|=t\frac{q^{\frac{mk}2}-1}{q-1}, \]
which equals to $\frac{q^{mk}-1}{q^m-1}$ and $L_{U_1}\cup  \ldots \cup L_{U_t}=\PG(k-1,q^m)$. Hence, if $t> (q-1)\frac{q^{\frac{mk}2}+1}{q^m-1}$ then at least two distinct $L_{U_i}$'s meets in at least one point, a contradiction.
\end{proof}

We are going to characterize the maximum $1$-designs reaching the bound on $t$ provided in the last item of the above corollary. We start by showing that $(q-1)\frac{q^{\frac{mk}2}+1}{q^m-1}$ is an integer if and only if $k$ is odd and $m=2$.

\begin{proposition}
Let $m$ and $k$ be two positive integers such that $mk$ is even and let $q$ be a prime power. Then $(q-1)\frac{q^{\frac{mk}2}+1}{q^m-1}$ is an integer if and only if $k$ is odd and $m=2$.
\end{proposition}
\begin{proof}
Let $r$ be the greatest common divisor of $q^{\frac{mk}2}+1$ and $\frac{q^m-1}{q-1}=q^{m-1}+\ldots+q+1$. Clearly, $(q-1)\frac{q^{\frac{mk}2}+1}{q^m-1}$ is an integer if and only if $r=\frac{q^m-1}{q-1}$. Since $\frac{q^m-1}{q-1}>2$, if $r \in \{1,2\}$, it results that $(q-1)\frac{q^{\frac{mk}2}+1}{q^m-1}$ is not an integer. So suppose that $r>2$ and $r=r_1^{e_1} \cdot r_2^{e_2} \cdots r_{g}^{e_g}$ is the decomposition in prime factor of $r$ with $r_1,\ldots,r_g$ distinct prime numbers and $e_1,\ldots,e_g$ positive integers. Let start by observing that, since $r$ divides $q^{\frac{mk}2}+1$, the $r_i$'s are different from $p$ otherwise $r$ divides 1. This means that $r$ and $q^j$ are coprime for every positive integer $j$. Since $r$ divides $q^{\frac{mk}2}+1$ and $q^{m-1}+\ldots+q+1$, then $r$ divides their difference, i.e. $q^{\frac{mk}2}-q^{m-1}-\ldots-q=q( q^{\frac{mk}2-1}-q^{m-2}-\ldots-1)$ and so $r$ divides $q^{\frac{mk}2-1}-q^{m-2}-\ldots-1$. Again since $r$ divides $q^{m-1}+\ldots+q+1$, we get that $r$ divides their sum, i.e. $q^{\frac{mk}2-1}-q^{m-1}=q^{m-1}\left(q^{\frac{m(k-2)}2}-1\right)$, which implies that $r$ also divides $q^{\frac{m(k-2)}2}-1$. Repeating $i$ times this argument, we get that $r$ also divides 
\begin{equation}\label{eq:divIkmgen}
q^{\frac{m(k-2i)}2}-1 
\end{equation}
when $i\leq \lfloor k/2 \rfloor$. If $k$ is even, then for $i=k/2$, we have that $r$ divides $2$, a contradiction. Therefore, $k$ is odd and as a consequence $m$ is even. In this case, when $i=(k-1)/2$ in \eqref{eq:divIkmgen} we get that $r$ divides $q^{\frac{m}{2}}+1$. In this way, $r \leq q^{\frac{m}2}+1 \leq \frac{q^m-1}{q-1}$, and so $(q-1)\frac{q^{\frac{mk}2}+1}{q^m-1}$ is an integer if and only if $r=q^{\frac{m}2}+1 =\frac{q^m-1}{q-1}$, which is possible if and only if $m=2$.
\end{proof}

In particular, this means that all the elements of a maximum $1$-design having $\frac{q^{k}+1}{q+1}$ subspaces define canonical subgeometries. As we will see in the next section, such subspace designs always exist under the assumptions that $m=2$ and $k$ is odd.

\subsection{Examples of \texorpdfstring{$1$}{Lg}-design}\label{sec:ex1des}

We now show some examples of maximum $1$-design. From a geometric point of view, to construct $1$-subspace design is equivalent to construct a certain number of pairwise disjoint scattered $\fq$-linear sets.
The first example can be also seen as a consequence of Theorems \ref{th:example2(n-1)des} and \ref{th:directsumdesigns}.

\begin{proposition}\label{prop:ex1desnorm}
Let $s$ and $m$ be two positive integers such that $\gcd(s,m)=1$.
Let $\mu_1,\ldots,\mu_{q-1} \in \F_{q^m}^*$ such that $\N_{q^m/q}(\mu_i) \neq \N_{q^m/q}(\mu_j)$, if $i \neq j$.\\ 
Let $U_i=\{(x_1,\mu_ix_1^{q^s},\ldots,x_r,\mu_ix_r^{q^s})\colon x_1,\ldots,x_r \in \F_{q^m}\}\subseteq \F_{q^m}^{2r}$, for $i\in[q-1]$. 
Then $(U_1,\ldots,U_{q-1})$ is a maximum $1$-design in $\F_{q^m}^{2r}$.
\end{proposition}
\begin{proof}
Clearly, $\dim_{\fq}(U_i\cap \langle {v} \rangle_{\F_{q^m}})\leq 1$ for every $ {v}\in \F_{q^m}^{2r} \setminus\{ {0}\}$.
Suppose by contradiction that there exist $i,j \in [q-1]$ and $ {v}\in \F_{q^m}^{2r} \setminus\{ {0}\}$ such that $\dim_{\fq}(U_i\cap \langle {v} \rangle_{\F_{q^m}})=\dim_{\fq}(U_j\cap \langle {v} \rangle_{\F_{q^m}})=1$.
Hence, there exists $\rho \in \F_{q^m}$ such that
\[\rho(x_1,\mu_ix_1^{q^s},\ldots,x_r,\mu_ix_r^{q^s})=(y_1,\mu_iy_1^{q^s},\ldots,y_r,\mu_iy_r^{q^s}),\]
for some $x_1,\ldots,x_r,y_1,\ldots,y_r\in \F_{q^m}$ such that $(x_1,\ldots,x_r),(y_1,\ldots,y_r)$ are not the zero vector.
Suppose that $x_\ell \ne 0$, then $y_\ell \ne 0$ and
\[ \begin{cases}
\rho x_\ell =y_\ell,\\
\rho \mu_i x_\ell^{q^s}=\mu_j y_\ell^{q^s}.
\end{cases} \]
Putting together the equations of the above system we obtain
\[ \mu_i x_{\ell}^{q^s-1}=\mu_j y_{\ell}^{q^s-1}, \]
that is $\N_{q^m/q}(\mu_i)= \N_{q^m/q}(\mu_j)$, a contradiction.
So, 
\[ \sum_{i=1}^{q-1} \dim_{\fq}(U_i \cap \langle  {v} \rangle_{\F_{q^m}})\leq 1 \]
for every $ {v}\in \F_{q^m}^{2r}\setminus\{ {0}\}$ and $\dim_{\fq}(U_i)=rm$, i.e. $(U_1,\ldots,U_{q-1})$ is a maximum $1$-subspace design.
\end{proof}

As a consequence of Corollary \ref{cor:8.9}, we obtain examples of two-intersection sets with respect to hyperplanes.

\begin{corollary}\label{lem:disjointpseudospace}
Let $s$ and $m$ be two positive integers such that $\gcd(s,m)=1$.
Let $\mu_1,\ldots,\mu_{q-1} \in \F_{q^m}^*$ such that $\N_{q^m/q}(\mu_i) \neq \N_{q^m/q}(\mu_j)$, if $i \neq j$.\\ 
Let $U_i=\{(x_1,\mu_ix_1^{q^s},\ldots,x_r,\mu_ix_r^{q^s})\colon x_1,\ldots,x_r \in \F_{q^m}\}\subseteq \F_{q^m}^{2r}$, for $i\in[q-1]$. 
The $L_{U_i}$'s are pairwise disjoint scattered $\fq$-linear sets of rank $rm$ in $\PG(2r-1,q^m)$.
Moreover, their union forms a two-intersection set with respect to hyperplanes and
\[\lvert \mathcal{H} \cap (L_{U_1} \cup \ldots \cup L_{U_{q-1}}) \rvert \in \{q^{m(r-1)}-1,2q^{m(r-1)}-1\},\]
for every hyperplane $\mathcal{H}$ of $\PG(2r-1,q^m)$.
\end{corollary}

\begin{remark}
The pointset described in Corollary \ref{lem:disjointpseudospace} is the union of special type of linear sets known as linear sets of pseudoregulus type, which were introduced in \cite{marino2007fq} and studied in \cite{lavrauw2013scattered,lunardon2014maximum}.
\end{remark}

We will now use subgeometries to construct $1$-designs which turn out to be maximum (and also optimal, see Section \ref{sec:optimal}) in some cases.

A \textbf{canonical subgeometry} of $\PG(k-1,q^m)$ is any $\PG(k-1,q)$ which is embedded in $\PG(k-1,q^m)$.
If $m=2$, we call it a \textbf{Baer subgeometry}.
Under some assumptions on $k$ and $m$ we can partition $\PG(k-1,q^m)$ into canonical subgeometries.

\begin{theorem}(see \cite[Theorem 4.29]{hirschfeld1998projective})\label{th:partionsubg}
There exists a partition of $\PG(k-1,q^m)$ into canonical subgeometries if and only if $\gcd(k,m)=1$.
\end{theorem}

We also refer to the paper of Eisfeld \cite{eisfeld1996some} for more constructions of large sets of large mutually disjoint canonical subgeometries.
In particular, we recall the following.

\begin{theorem}(see \cite[Theorem 2.2.6]{eisfeld1996some})\label{th:(m,n)=2}
If $\gcd(k,m)=2$, there exist mutually disjoint canonical subgeometries of $\PG(k-1,q^m)$ covering the points of $\PG(k-1,q^m)$ except for the points of two disjoint $\left(\frac{k-2}{2}\right)$-dimensional subspaces of $\PG(k-1,q^m)$.
\end{theorem}

In \cite{ebert1998partitioning} (see also \cite{baker2000baer}) it has been described a different way to construct a partition of $\PG(k-1,q^2)$ by Baer subgeometries when $k$ is odd (i.e. not by taking orbits under a certain subgroup of a Singer cycle).

Canonical subgeometries can be characterized as those $\fq$-linear sets $L_U$ of $\PG(k-1,q^m)=\PG(V,\F_{q^m})$ such that $U$ is an $\fq$-subspace of $V$ with $\dim_{\fq}(U)=k$ and $\langle U \rangle_{\F_{q^m}}=V$.
In particular, canonical subgeometries are scattered linear sets of rank $k$.

\begin{theorem}\label{th:subgeom1design}
If $\gcd(k,m)=1$ there exists a $1$-design of $V=V(k,q^m)$ consisting of $\frac{(q^{mk}-1)(q-1)}{(q^k-1)(q^m-1)}$ subspaces. If $m=2$ and $k$ is odd, then there exists a maximum $1$-design of $V=V(k,q^2)$ consisting of $\frac{q^{k}+1}{q+1}$ subspaces.
\end{theorem}
\begin{proof}
By Theorem \ref{th:partionsubg}, there exist $\ell=\frac{(q^{mk}-1)(q-1)}{(q^k-1)(q^m-1)}$ distinct canonical subgeometries $\Sigma_1=L_{U_1},\ldots,\Sigma_\ell=L_{U_\ell}$ of $\PG(k-1,q^m)$ partitioning $\PG(k-1,q^m)$. 
Therefore, for every point $P \in \PG(k-1,q^m)$ we have
\[ \sum_{i=1}^\ell w_{L_{U_i}}(P)=1, \]
and hence for every $\langle  {v}\rangle_{\F_{q^m}}$ with $ {v}\in V\setminus\{ {0}\}$  
\[ \sum_{i=1}^\ell \dim_{\fq}(U_i \cap \langle  {v}\rangle_{\F_{q^m}})=1. \]
When $m=2$, since the subspaces $U_i$ associated with  $\Sigma_i$ have dimension $k$, then $(U_1,\ldots,U_t)$ is a maximum $1$-design.
\end{proof}

The $1$-design in Theorem \ref{th:subgeom1design} is the only possible construction that maximizes the value $t$ of Corollary \ref{cor:8.9}.

\begin{corollary}
Let $(U_1,\ldots,U_t)$ be a maximum $1$-design in $V=V(k,q^m)$.
Consider the associated $\fq$-linear sets $L_{U_1},\ldots,L_{U_t}$. Then 
\begin{itemize}
    \item $t \leq (q-1)\frac{q^{\frac{mk}{2}}+1}{q^m-1}$;
    \item $t = (q-1)\frac{q^{\frac{mk}{2}}+1}{q^m-1}$ if and only if $m=2$, $k$ is odd, $L_{U_i}$ is a Baer subgeometry of $\PG(k-1,q^2)$ for any $i$ and $L_{U_1},\ldots,L_{U_t}$ is a partition of $\PG(k-1,q^2)$;
    \item there exists a maximum 1-design $(U_1,\ldots,U_t)$ in $\F_{q^2}^k$ with $t =\frac{q^{k}+1}{q+1}$, for any $k$ odd.
\end{itemize}
\end{corollary}

\noindent Similarly, one can prove the following by making use of Theorem \ref{th:(m,n)=2}.

\begin{theorem}
If $\gcd(k,m)=2$, there exists a $1$-design of $V=V(k,q^m)$ consisting of \\ $\frac{\left(q^{mk}-2q^{\frac{(k-2)m}{2}}+1\right)(q-1)}{(q^m-1)(q^k-1)}$ subspaces.
\end{theorem}

\subsection{Linear codes with two weights}\label{sec:2weight}

We refer to \cite{vladut2007algebraic} for properties of Hamming metric codes and the well-known connection between projective systems and linear codes.
We just recall that an $[n,k,d]_q$-linear code (or $[n,k]_q$-linear code) is an $\mathbb{F}_q$-linear code $\C$ of length $n$, dimension $k$ and minimum distance $d$ and a multiset of points in $\mathrm{PG}(k-1,q)$ is denoted by a pair $(\mathcal{P},\mathrm{m})$, where $\mathcal{P}$ is a subset of $\mathrm{PG}(k-1,q)$ and $\mathrm{m}$ is a multiplicity function.

\subsubsection{Linear codes and strongly regular graphs from maximum \texorpdfstring{$1$}{Lg}-design}

We can now consider the linear codes associated with the linear sets defined by a maximum $1$-design.
These codes turn out to very interesting, indeed we can completely determine their weight enumerator and they are two-weight codes (except for a special case).
So, by Theorem \ref{th:projproperties}, Corollary \ref{cor:8.9} and the well-known correspondence between linear codes and projective systems (see \cite{vladut2007algebraic}), the following holds.

\begin{corollary}
Let $(U_1,\ldots,U_t)$ be a maximum $1$-design in $V=V(k,q^m)$.
Consider the associated $\fq$-linear sets $L_{U_1},\ldots,L_{U_t}$.
Let consider $(\mathcal{P},\mathrm{m})$ the multiset where $\mathcal{P}=L_{U_1}\cup \ldots\cup L_{U_t}$ and with multiplicity function $\mathrm{m}$ as
\[
\mathrm{m}(P)=1, \text{for every}\,\,P \in \mathcal{P}.
\]
Then $(\mathcal{P},\mathrm{m})$ defines an $[N,k,d]_q$-linear code $\mathcal{C}$ in $\F_q^N$ and it has weight enumerator
\[ 1+(q^m-1)h_1z^{N-w_1}+(q^m-1)h_0z^{N-w_0}, \]
where
\begin{itemize}
    \item $w_0=t \frac{q^{\frac{m(k-2)}{2}}-1}{q-1}$;
    \item $w_1=(t-1) \frac{q^{\frac{m(k-2)}{2}}-1}{q-1}+\frac{q^{\frac{m(k-2)}{2}+1}-1}{q-1}$;
    \item $N=t\frac{q^{\frac{km}2}-1}{q-1}$;
    \item $h_1=t\frac{(q^{\frac{mk}2}-1)(q^{m(k-1)}-1)-(q^{\frac{m(k-2)}2}-1)(q^{mk}-1)}{(q^m-1)(q-1)q^{\frac{m(k-2)}2}}$;
    \item $h_0=\frac{q^{mk}-1}{q^m-1}-h_1$;
    \item $d=N-w_1$.
\end{itemize}
In particular, $\C$ is a two-weight code if and only if 
$t\ne (q-1)\frac{q^{\frac{mk}2}+1}{q^m-1}$.
Moreover, if $t= (q-1)\frac{q^{\frac{mk}2}+1}{q^m-1}$ then $\C$ is a one-weight code.
\end{corollary}

\begin{remark}
One-weight codes correspond to the simplex codes in the Hamming metric, as proved by Bonisoli in \cite{bonisoli1983every}.
\end{remark}

It is also well-known the connection between two-intersection sets with respect to hyperplanes and strongly regular graphs, see e.g.\ \cite{calderbank1986geometry}.

A regular graph $\mathcal{G}$ (that is a graph in which every vertex has the same number of neighbors) with $v$ vertices and degree $K$ is said to be \textbf{strongly regular} if there exist two integers $\lambda$ and $\mu$ such that every two adjacent vertices have $\lambda$ common neighbours and every two non-adjacent vertices have $\mu$ common neighbours. In this case we say that $(v,K,\lambda,\mu)$ are its parameters. Now, we describe the connection between a two-intersection set with respect to hyperplanes and a strongly regular graph.

Let $\mathcal{P}$ be a set of points in $\PG(k-1,q)$ such that $\langle \mathcal{P} \rangle=\PG(k-1,q)$.
Embed $\PG(k-1,q)$ as a hyperplane $\mathcal{H}$ of $\PG(k,q)$.
Define $\Gamma(\mathcal{P})$ as the graph whose vertices are the points in $\PG(k,q)\setminus \mathcal{H}$ and two vertices $P$ and $Q$ are adjacent if the line $PQ$ meets $\mathcal{H}$ in a point of $\mathcal{P}$.
When $\mathcal{P}$ is a two-intersection set with respect to hyperplanes, then $\Gamma(\mathcal{P})$ is a strongly regular graph and its parameters are described in the following result.

\begin{theorem}(see \cite{calderbank1986geometry})
Let $\mathcal{P}$ be a set of points in $\PG(k-1,q)$ of size $N$ such that every hyperplane meets $\mathcal{P}$ in either $w_0$ or $w_1$ points and there exist at least two hyperplanes $\mathcal{H}_0$ and $\mathcal{H}_1$ of $\PG(k-1,q)$ with the property that
\[ |\mathcal{H}_0\cap \mathcal{P}|=w_0\,\,\,\mbox{and}\,\,\,|\mathcal{H}_1\cap \mathcal{P}|=w_1. \]
If $\langle\mathcal{P}\rangle=\PG(k-1,q)$, then $\Gamma(\mathcal{P})$ is a strongly regular graph whose parameters are $(v,K,\lambda,\mu)$, where
\begin{itemize}
    \item $v=q^k$;
    \item $K=N(q-1)$;
    \item $\lambda=K^2+3K-q(2N-w_0-w_1)-Kq(2N-w_0-w_1)+q^2(N-w_0)(N-w_1)$;
    \item $\mu=\frac{q^2(N-w_0)(N-w_1)}{q^k}=K^2+K-Kq(2N-w_0-w_1)+q^2(N-w_0)(N-w_1)$.
\end{itemize}
\end{theorem}

\noindent When applying the above result to a maximum $1$-design we obtain the following.

\begin{corollary}\label{cor:srg}
Let $q$ be a prime power, $m,k,t$ positive integers such that $mk$ is even and $t\ne (q-1)\frac{q^ {\frac{mk}{2}}+1}{q^m-1}$ and let $(U_1,\ldots,U_t)$ be a maximum $1$-design in $V=V(k,q^m)$.
Let $N=t\frac{q^{\frac{km}2}-1}{q-1}$, $w_0=t \frac{q^{\frac{m(k-2)}{2}}-1}{q-1}$ and $w_1=(t-1) \frac{q^{\frac{m(k-2)}{2}}-1}{q-1}+\frac{q^{\frac{m(k-2)}{2}+1}-1}{q-1}$. Then $\Gamma(L_{U_1}\cup\ldots\cup L_{U_t})$ is a strongly regular graph with parameters $(v,K,\lambda,\mu)$, where
\begin{itemize}
    \item $v=q^{mk}$;
    \item $K=N(q^m-1)$;
    \item $\lambda=K^2+3K-q^m(2N-w_1-w_0)-Kq^m(2N-w_1-w_0)+q^{2m}(N-w_1)(N-w_0)$;
    \item $\mu=\frac{q^{2m}(N-w_1)(N-w_0)}{q^{mk}}$.
\end{itemize}
\end{corollary}

\begin{remark}
By the constructions in Section \ref{sec:ex1des}, we always have examples of maximum $1$-design when $k$ is even, and hence a strongly regular graph with parameters as in Corollary \ref{cor:srg} always exist for every even $k$.
\end{remark}

\section{Singleton bound for sum-rank metric codes and optimal subspace designs}\label{sec:Singleton}

In this section we will first recall some definitions and results on sum-rank metric codes and then we will  identify the geometric counterparts of those which are optimal.
We will call them \textbf{optimal subspace designs}.

\subsection{Sum-rank metric codes}

The sum-rank metric has been recently investigated especially because of the performance of multishot network coding based on sum-rank metric codes, see \cite{nobrega2010multishot}.
In the following we will recall some results from  \cite{alfarano2021sum,byrne2021fundamental,martinez2019theory,Martinez2018skew,martinez2020general,moreno2021optimal,neri2021geometry}.

Let $t$ be a positive integer. Let $m_1,\ldots,m_t,n_1,\ldots,n_t$ be positive integers.
We consider the product of $t$ matrix spaces 
$$\Pi:= \bigoplus_{i=1}^t \F_q^{m_i \times n_i}.$$
The \textbf{sum-rank distance} is the function 
\[
d : \Pi \times \Pi \longrightarrow \mathbb{N}  
\]
defined by
\[
d(X,Y) = \sum_{i=1}^t \mathrm{rk}(X_i - Y_i),
\]
where $X=(X_1 \lvert \ldots \lvert X_t)$, $Y=(Y_1 \lvert \ldots \lvert Y_t)$, with $X_i,Y_i \in \F_q^{m_i \times n_i}$. 

We define the \textbf{sum-rank weight} of an element $X=(X_1|\ldots|X_t) \in \Pi$ as 
$$w(X):=\sum_{i=1}^t \mathrm{rk}(X_i).$$
Clearly, $d(X,Y)= w(X-Y)$, for every $X,Y \in \Pi$. \\
A (linear) \textbf{sum-rank metric code} $\mathrm{C}$ is an $\F_q$-linear subspace of $\Pi$ endowed with the sum-rank distance.
The \textbf{minimum sum-rank distance} of a sum-rank metric code $\mathrm{C}$ is defined as usual via $$d(\mathrm{C})=\min\{w(X): X \in \mathrm{C}, X \neq 0\}.$$ 

However, for our purposes we will need the description of sum-rank metric codes in terms of vectors.

\subsection{Vector sum-rank metric codes and the associated subspace designs}\label{sec:vectsumrank}

In this section we set
$\mathbf{n}=(n_1,\ldots,n_t) \in \mathbb{N}^t$ an ordered tuples with $n_1 \geq n_2 \geq \ldots \geq n_t$ and we set $N \coloneqq n_1+\ldots+n_t$. We use the following compact notations for the direct sum of vector spaces 
$$\F_{q^m}^\bfn\coloneqq\bigoplus_{i=1}^t\F_{q^m}^{n_i}.$$
Furthermore, we follow the notation used in \cite[Section 3.3]{alfarano2021sum}.

Let start by recalling that the rank of a vector $v=(v_1,\ldots,v_n) \in \F_{q^m}^n$ is the dimension of the
vector space generated over $\F_q$ by its entries, i.e. $\rk(v)=\dim_{\fq} (\langle v_1,\ldots, v_n\rangle_{\fq})$. The sum-rank weight of an element $x=(x_1 ,\ldots, x_t) \in \F_{q^m}^\bfn$ is 
$$ \mathrm{w}(x)=\sum_{i=1}^t \rk(x_i).$$

A (linear vector) sum-rank metric code $\mathrm{C} $ is an $\F_{q^m}$-subspace of $\F_{q^m}^{\bfn}$ endowed with the sum-rank distance defined as
\[
\mathrm{d}(x,y)=\mathrm{w}(x-y)=\sum_{i=1}^t \rk(x_i-y_i),
\]
where $x=(x_1 , \ldots , x_t), y=(y_1 , \ldots, y_t) \in \F_{q^m}^\bfn$. 
Let $\mathrm{C} \subseteq \F_{q^m}^\bfn$ be a sum-rank metric code. We will write that $\mathrm{C}$ is an $[\bfn,k,d]_{q^m/q}$ code (or $[\bfn,k]_{q^m/q}$ code) if $k$ is the $\F_{q^m}$-dimension of $\mathrm{C}$ and $d$ is its minimum distance, that is 
\[
d=d(\mathrm{C})=\min\{\mathrm{d}(x,y) \colon x, y \in \mathrm{C}, x \neq y  \}.
\]

The matrix and vector settings for the sum-rank metric described above are related in the following way.
For every $r \in [t]$, let $\Gamma_r=(\gamma_1^{(r)},\ldots,\gamma_m^{(r)})$ be an ordered $\fq$-basis of $\F_{q^m}$, and let $\Gamma=(\Gamma_1,\ldots,\Gamma_t)$. Given   $x=(x_1, \ldots ,x_t) \in \F_{q^m}^\bfn$, with $x_i \in \F_{q^m}^{n_i}$, define the element $\Gamma(x)=(\Gamma_1(x_1), \ldots, \Gamma_t(x_t)) \in \Pi$, where
$$x_{r,i} = \sum_{j=1}^m \Gamma_r (x_r)_{ij}\gamma_j^{(r)}, \qquad \mbox{ for all } i \in [n_r].$$
In other words, the $r$-th block of $\Gamma(x)$ is the matrix expansion of the vector $x_r$ with respect to the $\fq$-basis $\Gamma_r$ of $\F_{q^m}$.
As already noted in \cite[Theorem 2.7]{neri2021geometry}, the map 
    $$\Gamma: \F_{q^m}^\bfn \longrightarrow \Pi$$
is an $\fq$-linear isometry between the metric spaces $(\F_{q^m}^\bfn, \mathrm{d})$ and $(\Pi,d)$. 
For more details on this two settings, see \cite{byrne2021fundamental} and \cite{neri2021twisted}. \\
It is possible to define a Singleton-like bound for a sum-rank metric code. It has been stated and proved in the matrix setting (see \cite[Theorem 3.2]{byrne2021fundamental}), but in the next we adapt it for the vector setting.
Let $\mathbf{n}=(n_1,\ldots,n_t) \in \mathbb{N}^t$ with $n_1 \geq n_2 \geq \ldots \geq n_t$ and let $d$ be a positive integer such that $d \leq \sum_{i=1}^t \min \{m,n_i\}$. Let $j$ be the minimum positive integer such that $d\leq \sum_{i=1}^j \min \{m,n_i\}$. Then there exists a unique non-negative integer $\delta$ such that $d=\sum_{i=1}^{j-1} \min\{m,n_i\}+\delta+1$ with $0 \leq \delta\leq \min\{m,n_i\}-1$. This way of writing $d$ allows us to find the maximum possible size for an $[\mathbf{n},k,d]_{q^m/q}$ code.

\begin{theorem} \label{th:sbadapted}
    Let $\mathrm{C} \subseteq \F_{q^m}^\bfn$ be an $[\bfn,k,d]_{q^m/q}$ code.
    Let $j$ and $\delta$ be the unique integers satisfying \[d-1=\sum_{i=1}^{j-1}\min\{m,n_i\}+\delta \mbox{ and }0 \leq \delta \leq \min\{m,n_j\}-1.\]
Then 
\begin{equation}\label{eq:boundgen}
\lvert \mathrm{C} \rvert \leq q^{m\sum_{i=j}^tn_i-\max\{m,n_j\} \delta}.\end{equation}
In particular, 
\begin{enumerate}
        \item If $m \geq n_1$, then $d-1=\sum_{i=1}^{j-1}n_i+\delta \mbox{ with }0 \leq \delta \leq n_j-1$ and 
        \[
        \lvert \mathrm{C} \rvert \leq q^{m(N-d+1)}.
        \]
        \item If $n=n_1=\ldots=n_t \geq m$ then $d-1=m(j-1)+\delta, \,\,0 \leq \delta \leq m-1$ and  
        \[
        \lvert \mathrm{C} \rvert \leq q^{n(tm-d+1)}.
        \]
    \end{enumerate}
\end{theorem}


A $[\bfn,k,d]_{q^m/q}$ code is called \textbf{Maximum Sum-Rank Distance code} (or shortly \textbf{MSRD code}) if its size attains the bound of Theorem \ref{th:sbadapted}.

The map 
\[(x,y)\in\F_{q^m}^\bfn\times \F_{q^m}^\bfn \mapsto \sum_{i=1}^t x_i \cdot y_i \in \F_{q^m}\] 
is a symmetric and non-degenerate bilinear form on $\F_{q^m}^\bfn$.
The dual of an $[\bfn,k,d]_{q^m/q}$ code $\mathrm{C}$ is defined as 
\begin{equation}\label{eq:dualsrmcodes}
\mathrm{C}^{\perp} =\left\{
y=(y_1 ,\ldots, y_t) \in \F_{q^m}^\bfn \colon \sum_{i=1}^t x_i \cdot y_i=0 \mbox{ for all } x=(x_1 \lvert \ldots \lvert x_t) \in \mathrm{C}
\right\}.
\end{equation}

\noindent The dual of a sum-rank code $\mathrm{C}$ is also a sum-rank code and
\begin{equation} \label{eq:dimperpvector}
\dim_{\F_{q^m}}(\mathrm{C}^{\perp})=N-\dim_{\F_{q^m}}(\mathrm{C}).
\end{equation}

\noindent Also, the MSRD property is invariant under duality.

\begin{theorem}\label{th:dualMSRD}(see \cite[Theorem 6.1]{byrne2021fundamental},\cite[Theorem 5]{martinez2019theory})
Let $\mathrm{C}$ be an $[\bfn,k,d]_{q^m/q}$ code. 
Then, $$d(\mathrm{C}^{\perp})\leq 
     \begin{cases}
     N-d+2, \mbox{ if }n_1 \leq m, \\
     tm-d+2, \mbox{ if }m \leq n_1\mbox{ and }n_1=\ldots=n_t,
     \end{cases}
     $$
    and equality holds when $\mathrm{C}$ is an MSRD code. In particular, if $\mathrm{C}$ is an MSRD code then $\mathrm{C}^\perp$ is an MSRD code as well.
\end{theorem}

\noindent The second part of the above result follows by Theorem \ref{th:sbadapted}.

Now, we recall the definition of equivalence between sum-rank metric codes. An $\F_{q^m}$-linear isometry $\phi$ of $\F_{q^m}^{\bfn}$ is an $\F_{q^m}$-linear map of $\F_{q^m}^{\bfn}$ that preserves the distance, i.e. $\mathrm{w}(x)=\mathrm{w}(\phi(x))$, for every $x \in  \F_{q^m}^{\bfn}$, or equivalently $\mathrm{d}(x,y)=\mathrm{d}(\phi(x),\phi(y))$, for every $x,y \in  \F_{q^m}^{\bfn}$.
The $\F_{q^m}$-linear isometries of $\F_{q^m}^\bfn$ equipped with sum-rank metric have been classified in \cite[Theorem 2]{martinezpenas2021hamming}, see also \cite[Theorem 3.7]{alfarano2021sum}. 

In order to present such a classification, we introduce the following notation. 
Let $\mathcal{N}(\mathbf{n}):=\{n_1,\ldots,n_t\}$, and let $\ell:=\lvert \mathcal{N}(\mathbf{n}) \rvert $. Let $n_{i_1},\ldots,n_{i_{\ell}}$ be the distinct elements of $\mathcal{N}(\mathbf{n})$.
By $\lambda(\mathbf{n})\in \mathbb{N}^{\ell}$ we will denote the vector whose entries are
\[
\lambda_j:=\lvert \{k \colon n_k=n_{i_j} \} \rvert, \ \ \ \ \mbox{for each }j=1,\ldots,\ell.
\]
So, $\lambda_j$ denotes the number of blocks that have the same length $n_{i_j}$.
For a vector $\mathbf{v}=(v_1,\ldots, v_{\ell})\in \mathbb{N}^{\ell}$, we define $$S_{\mathbf{v}}=S_{v_{1}} \times \cdots \times S_{v_\ell},$$
where $S_i$ is the symmetric group of order $i$. 
Similarly, we denote by $\GL(\mathbf{v}, \F_q)$ the direct product of the general linear groups of degree $v_i$ over $\F_q$, i.e.
$$ \GL(\bfn, \F_q) = \GL(n_1, \F_q)\times \ldots\times \GL(n_t, \F_q).$$

\begin{theorem} (see \cite[Theorem 2]{martinezpenas2021hamming}) \label{th:isometry}
The group of $\F_{q^m}$-linear isometries of the space $\F_{q^m}^{\mathbf{n}}$ endowed with the sum-rank metric is
$$((\F_{q^m}^\ast)^{t} \times \GL(\mathbf{n}, \F_q)) \rtimes S_{\lambda(\mathbf{n})},$$
which acts as 
  \begin{equation*}( {a},M_1,\ldots, M_{t},\pi)\cdot (c^{(1)} \mid \ldots \mid c^{(t)})\longmapsto (a_1c^{(\pi^{-1}(1))}M_1 \mid \ldots \mid a_{t}c^{(\pi^{-1}(t))}M_{t}).\end{equation*}
\end{theorem}

This means that the $\F_{q^m}$-linear isometries of $\F_{q^m}^{\mathbf{n}}$ are the $\F_{q^m}$-linear maps that multiply each block for a nonzero scalar in $\F_{q^m}^*$ and for invertible matrices with coefficients in $\F_q$ and permute blocks of the same length.

We will say that two sum-rank metric codes $\mathrm{C}_1$ and $\mathrm{C}_2$ are equivalent if there exists an isometry $\varphi$ as in Theorem \ref{th:isometry} such that $\varphi(\mathrm{C}_1)=\mathrm{C}_2$.
We denote the set of equivalence classes of $[\mathbf{n},k,d]_{q^m/q}$ sum-rank metric codes by $\mathfrak{C}[\mathbf{n},k,d]_{q^m/q}$.

\begin{remark}\label{rk:equivcodes}
It is easy to see that two sum-rank metric codes $\mathrm{C}_1$ and $\mathrm{C}_2$ are equivalent if and only if their dual codes $\mathrm{C}_1^{\perp}$ and $\mathrm{C}_2^{\perp}$ are equivalent
Indeed, if $\mathrm{C}_1$ and $\mathrm{C}_2$ are equivalent then there exists $( {a},M_1,\ldots, M_{t},\pi)\in(\F_{q^m}^\ast)^{t} \times \GL(\mathbf{n}, \F_q)\rtimes S_{\lambda(\mathbf{n})}$ such that $(a,M_1,\ldots, M_{t},\pi) \mathrm{C}_1=\mathrm{C}_2$.
The codes $\mathrm{C}_1^\perp$ and $\mathrm{C}_2^\perp$ are equivalent via the isometry defined by $(a',(M_1^\top)^{-1},\ldots, (M_{t}^\top)^{-1},\pi)$, where $ {a}'=(a_1^{-1},\ldots,a_t^{-1})$.
The converse clearly holds.
\end{remark}

Let $\mathrm{C} \subseteq \F_{q^m}^{\mathbf{n}}$ be a linear sum-rank metric code. Let $G=(G_1\lvert \ldots \lvert G_t) \in \F_{q^m}^{k \times N}$ be a generator matrix of $\mathrm{C}$, with $G_1,\ldots,G_t \in \F_{q^m}^{k \times {n_i}}$. We say that $\C$ is \textbf{non-degenerate} if the columns of $G_i$ are $\fq$-linearly independent for $i\in \{1,\ldots,t\}$; see \cite[Definition 2.11, Proposition 2.13]{neri2021geometry}.

We will recall now some results from \cite{neri2021geometry}, on the connections between sum-rank metric codes and some sets of subspaces, which will be then rephrased in terms of subspace designs.

\begin{theorem}(see \cite[Theorem 3.1]{neri2021geometry}) \label{th:connection}
Let $\mathrm{C}$ be a non-degenerate $[\bfn,k,d]_{q^m/q}$ code. Let $G=(G_1\lvert \ldots \lvert G_t)$ be an its generator matrix.
Let $U_i \subseteq \F_{q^m}^k$ be the $\F_q$-span of the columns of $G_i$, for $i\in \{1,\ldots,t\}$.
The sum-rank weight of an element $x G \in \mathrm{C}$, with $x \in \F_{q^m}^k$, is
\[
\w(x G) = N - \sum_{i=1}^t \dim_{\fq}(U_i \cap x^{\perp}),\]
where $x^{\perp}=\{y \in \F_{q^m}^k \colon x \cdot y=0\}.$ In particular,
\begin{equation} \label{eq:distancedesign}
d=N- \max\left\{ \sum_{i=1}^t \dim_{\fq}(U_i \cap H)  \colon H\mbox{ is an } \F_{q^m}\mbox{-hyperplane of }\F_{q^m}^k  \right\}.
\end{equation}
\end{theorem}

The equivalence classes of Hamming-metric non-degenerate codes are in one-to-one correspondence with equivalence classes of projective systems. Recently, in \cite{Randrianarisoa2020ageometric} it has been shown that equivalence classes of non-degenerate rank metric codes are in one-to-one correspondence with equivalence classes of $q$-systems, where the latter constitute the $q$-analogue of projective systems. 
The following definition extends the notions of projective systems and $q$-systems.

\begin{definition}
Let $\mathbf{n}=(n_1,\ldots,n_t) \in \mathbb{N}^{t}$, with $n_1\geq \cdots \geq n_t$. An $[\mathbf{n},k,d]_{q^m/q}$-\textbf{system} $\mathrm{U}$ is an ordered set $(U_1,\cdots,U_t)$, where, for any $i\in [t]$, $U_i$ is an $\F_q$-subspace of $\F_{q^m}^k$ of dimension $n_i$, such that
$ \langle U_1, \ldots, U_t \rangle_{\F_{q^m}}=\F_{q^m}^k$ and 
$$ d=N-\max\left\{\sum_{i=1}^t\dim_{\F_q}(U_i\cap H) \mid H \textnormal{ is an $\F_{q^m}$-hyperplane of }\F_{q^m}^k\right\}.$$
Moreover, two $[\mathbf{n},k,d]_{q^m/q}$-systems $(U_1,\ldots,U_t)$ and $(V_1,\ldots, V_t)$ are \textbf{equivalent} if there exists an isomorphism $\varphi\in\GL(k,\F_{q^m})$ and a permutation $\sigma\in\mathcal{S}_t$, such that for every $i\in[t]$
$$ \varphi(U_i) = a_iV_{\sigma(i)}.$$
We denote the set of equivalence classes of $[\mathbf{n},k,d]_{q^m/q}$-systems by $\mathfrak{U}[\mathbf{n},k,d]_{q^m/q}$.
\end{definition}

Clearly, systems naturally defines subspace designs with respect to hyperplanes as follows.

\begin{proposition} \label{prop:relationdesignsystem}
Let $(U_1,\ldots,U_t)$ be an ordered set of $\fq$-subspaces in $\F_{q^m}^k$. Then $(U_1,\ldots,U_t)$ is an $[\mathbf{n},k,d]_{q^m/q}$-system if and only if $(U_1,\ldots,U_t)$ form an $(k-1,N-d)_q$-subspace design of $\F_{q^m}^k$ such that 
\begin{itemize}
    \item $\dim_{\F_q} (U_i)=n_i$,
    \item $\sum_{i=1}^t \dim_{\F_q}(U_i \cap H)=N-d$, for some $\F_{q^m}$-hyperplane $H$ of $\F_{q^m}^k$,
    \item $ \langle U_1, \ldots, U_t \rangle_{\F_{q^m}}=\F_{q^m}^k$.
\end{itemize}
\end{proposition}

In \cite{neri2021geometry}, it has been proved that there is a one-to-one  correspondence  between  equivalence  classes  of  sum-rank non-degenerate $[\mathbf{n},k,d]_{q^m/q}$ code and equivalence classes of $[\mathbf{n},k,d]_{q^m/q}$-systems.
This correspondence can be formalized by the following two maps
\begin{align*}
    \Psi :  \mathfrak{C}[\mathbf{n},k,d]_{q^m/q} &\to\mathfrak{U}[\mathbf{n},k,d]_{q^m/q} \\
    \Phi : \mathfrak{U}[\mathbf{n},k,d]_{q^m/q} &\to \mathfrak{C}[\mathbf{n},k,d]_{q^m/q}.
\end{align*}
Let $[\mathrm{C}]\in\mathfrak{C}[\mathbf{n},k,d]_{q^m/q}$ and let $\overline{\mathrm{C}}$ be a representative of $[\mathrm{C}]$. Let $\mathbf{G}=(G_1\lvert \ldots \lvert G_t)$ be a generator matrix for $\overline{\mathrm{C}}$. Then $\Psi([\mathrm{C}])$ is the equivalence class of $[\mathbf{n},k,d]_{q^m/q}$-systems $[\mathrm{U}]$, where $\mathrm{U}=(U_1,\ldots,U_t)$ and $U_i$ is the $\F_q$-span of the columns of $G_i$ for every $i\in[t]$. In this case $U$ is also called the \textbf{system associated with} $\overline{\mathrm{C}}$. Viceversa, given $[(U_1,\ldots,U_t)]\in\mathfrak{U}[\mathbf{n},k,d]_{q^m/q}$, for every $i\in[t]$. Define $G_i$ as the matrix whose columns are an $\F_q$-basis of $U_i$. Then $\Phi([(U_1,\ldots,U_t)])$ is the equivalence class of the sum-rank metric codes generated by $G=(G_1\lvert \ldots \lvert G_t)$. $\Psi$ and $\Phi$ are well-posed and they are inverse of each other.

\begin{remark} \label{rk:particolarsystem}
As observed in \cite[Remark 3.6]{neri2021geometry}, when $t=1$, the definition of $[\mathbf{n},k,d]_{q^m/q}$-system coincides with the definition of $[n,k,d]_{q^m/q}$-system introduced in \cite{Randrianarisoa2020ageometric} and the correspondence $(\Psi,\Phi)$ gives us a one-to-one correspondence between classes of non-degenerate  $[n,k,d]_{q^m/q}$ rank metric codes and classes of $[n,k,d]_{q^m/q}$-system, for more details see \cite{Randrianarisoa2020ageometric} and \cite{alfarano2022linear}.
When $n_1=\ldots=n_t=1$, the definition of $[\mathbf{1},k,d]_{q^m/q}$-system does not immediately coincide with that of projective system (see \cite{vladut2007algebraic}), but we can still identify classes of projective systems and classes of $[\mathbf{1},k,d]_{q^m/q}$-system.
\end{remark}

\subsection{Optimal subspace designs}\label{sec:optimal}

We are now ready to give an answer to Problem \ref{prob:B} when considering subspace designs with respect to hyperplanes. The strategy regards the use of Proposition \ref{prop:relationdesignsystem}, which points out a connection between subspace designs with respect to hyperplanes and systems, and the use of the connection shown in the previous section between systems and linear sum-rank metric codes. This allows us to use the Singleton bound to get a lower bound on the parameter $A$ of a $(k-1,A)$-subspace design which is non-degenerate.

\begin{theorem}\label{cor:MSRDoptimalsd}
Let $\mathrm{U}=(U_1,\ldots,U_t)$ be a non-degenerate $(k-1,A)_q$-subspace design in $\F_{q^m}^k$, with $n_i=\dim_{\fq}(U_i)$ for every $i$, and $n_1 \geq \ldots \geq n_t$ and let 
\[
M=\max\left\{ \sum_{i=1}^t \dim_{\fq}(U_i \cap H)  \colon H\mbox{ is an } \F_{q^m}\mbox{-hyperplane of }\F_{q^m}^k  \right\}
\]
and $d=N-M$, where $N=\sum_{i=1}^tn_i$.
Let $j$ and $\delta$ be the unique integers satisfying \[d-1=\sum_{i=1}^{j-1}\min\{m,n_i\}+\delta \mbox{ and }0 \leq \delta \leq \min\{m,n_j\}-1.\]
Then
\begin{equation} \label{eq:Abysystem}
A \geq N-\left(
\frac{m \sum_{i=j}^{t}n_i-mk}{\max\{m,n_j\}}+\sum_{i=1}^{j-1}\min\{m,n_j\}+1\right).
\end{equation}
In particular, 
\[
A \geq \begin{cases}
k-1, & \mbox{ if } n_1 \leq m, \\
N-tm+\frac{m}{n}k-1, & \mbox{ if } m \leq n=n_1=\ldots=n_t.
\end{cases}
\]
Moreover, a sum-rank metric code
 $\mathrm{C}\in \Phi([\mathrm{U}])$ is an MSRD code if and only if
 \[
M \leq N-\left(
\frac{m \sum_{i=j}^{t}n_i-mk}{\max\{m,n_j\}}+\sum_{i=1}^{j-1}\min\{m,n_j\}+1\right).
\]
In particular, $\mathrm{C}$ is an MSRD code if and only if 
\begin{equation} \label{eq:optimalmsrd}
M \leq \begin{cases}
k-1, & \mbox{ if } n_1 \leq m, \\ \\
N-tm+\frac{m}{n}k-1, & \mbox{ if } m \leq n=n_1=\ldots=n_t.
\end{cases}
\end{equation}
\end{theorem}
\begin{proof}
First, we observe that $\mathrm{U}=(U_1,\ldots,U_t)$ is an $\Fmnkd$-system. Let $\mathrm{C}\in \Phi([\mathrm{U}])$, then $\mathrm{C}$ is an $\Fmnkd$ code and by \eqref{eq:boundgen} we have that 
\[
mk \leq m \sum_{i=j}^t n_i-\max\{m,n_j\}\delta=m \sum_{i=j}^t n_i-\max\{m,n_j\}(d-1-\sum_{i=1}^{j-1}\min\{m,n_i\}).
\]
Since $(U_1,\ldots,U_t)$ is a $(k-1,A)_q$-subspace design, then $M \leq A$ and $d \geq N-A$. Then \eqref{eq:Abysystem} follows. The remaining part follows by the definition of MSRD code.
\end{proof}

The subspace designs as in Theorem \ref{cor:MSRDoptimalsd} satisfying equality in \eqref{eq:Abysystem} will be called \textbf{optimal subspace designs}, i.e. when the associated sum-rank metric code is an MSRD code.

\begin{remark}
Let $(U_1,\ldots,U_t)$ be a non-degenerate $(k-1,A)_q$-design in $\F_{q^m}^k$ with $\dim_{\fq}(U_i)\leq m$ for every $i \in \{1,\ldots,t\}$.
Then $(U_1,\ldots,U_t)$ is an optimal subspace design if and only if $(U_1,\ldots,U_t)$ is an $(k-1)$-design.
\end{remark}

We will see some examples of optimal subspace designs  in the next section.

\subsection{Examples of optimal subspace designs}\label{sec:exoptimalsubdes}

In order to construct optimal subspace designs, we can consider a construction of MSRD codes in \cite{byrne2021fundamental}, relying on \textbf{MRD codes}, that is rank metric codes whose parameters satisfy equality in the Singleton bound for rank metric codes, and rephrased in the vector framework.

The study of MSRD codes started with Mart{\'\i}nez-Pe{\~n}as \cite{Martinez2018skew}, and since then many other papers analyzed
their structure \cite{martinezpenas2019reliable,martinezpenas2019universal,ott2021bounds} from the vector point of view. In this framework, the linearized Reed-Solomon codes represented the first family of MSRD codes constructed.
Neri in \cite{neri2021twisted} extended the family of generalized twisted Gabidulin codes to a new family of MSRD codes.
In Theorem \ref{th:example2(n-1)des}, we already constructed examples of $(k-1)$-designs in $V=V(k,q^m)$ which turned out to be optimal subspace designs.

Other examples can be provided in \cite[Construction 7.2 (a)]{byrne2021fundamental}. Let $m \geq n_1 \geq n_2 \geq \cdots \geq n_t$ and let embed $\F_{q^m}^{n_i}$ in $\F_{q^m}^{n_1}$ and consider an $\F_{q^m}$-linear MRD-code in $\F_{q^m}^{n_1}$. It is possible to construct an MSRD code whose first block is obtained by considering the elements of the fixed MRD code.

\begin{theorem}
Let $m \geq n_1 \geq n_2 \geq \ldots \geq n_t$ and let $N=n_1+\ldots+n_t$. 
Let $\hat{\mathrm{C}} \subseteq \F_{q^m}^{n_1}$ be an MRD code with minimum distance $2$, i.e. $\hat{\mathrm{C}}$ is an $[n_1,n_1-1,2]_{q^m/q}$ code. Let consider the map $\phi_i:\F_{q^m}^{n_i} \rightarrow \F_{q^m}^{n_1}$ defined by adding $n_1-n_i$ zero entries to the given vector. Define the sum-rank metric code as
\[\mathrm{C}=\left\{\left( {v}_1-\sum_{i=2}^t\phi_i( {v}_i),  {v}_2, \ldots ,  {v}_t\right): {v}_1 \in \hat{\mathrm{C}},  {v}_i \in \F_{q^m}^{n_i}\right\} \subseteq \F_{q^m}^{N}\]
Then $\mathrm{C}$ is an $[\mathbf{n},N-1,2]_{q^m/q}$-code, and hence an MSRD code.
\end{theorem}

Thanks to the previous construction we can define the following optimal design. This can be obtained by plugging together an $(n_1-1)$-scattered subspace of $\F_{q^m}^{n_1}$ and subgeometries of $\F_{q^m}^{n_i}$.

\begin{corollary}
Let $m,n_1,\ldots,n_t$ be positive integers such that $m \geq n_1 \geq \cdots \geq n_t$ and let $N=n_1+\ldots+n_t$. Let $U$ be an $[n_1,n_1-1,2]_{q^m/q}$-system.  Let $ {g}_1=(g_{1,1},\ldots,$ $g_{1,n_1-1}),\ldots,$ $ {g}_{n_1}=(g_{n_1,1},\ldots,g_{n_1,n_1-1})$ be a basis of $U$. 
Let 
\[
U_1=\{
(\alpha_1 g_{1,1}+\ldots+\alpha_{n_1}g_{n_1,1},\ldots, \alpha_1 g_{1,n_1-1}+\ldots+\alpha_{n_1}g_{n_1,n_1-1}\lvert  \]\[ -\alpha_1,\ldots,-\alpha_{n_2} \lvert \ldots \lvert -\alpha_1,\ldots,-\alpha_{n_t}):\alpha_i \in \F_q
\} \subseteq \F_{q^m}^{N-1}
\]
and 
\[
U_i= 0\F_{q}^{n_1-1}  \oplus_{j=2}^t \delta_{i,j}\F_{q}^{n_j},
\]
for any $i\in\{2,\ldots,t\}$ and with $\delta_{i,j}=1$ if $i=j$ and $\delta_{i,j}=0$ otherwise.
Then $(U_1,\ldots,U_t)$ is an $[\mathbf{n},N-1,2]_{q^m/q}$-system with $\mathbf{n}=(n_1,\ldots,n_t)$. Hence, $(U_1,\ldots,U_t)$ is an $(N-2)$-design in $\F_{q^m}^{N-1}$ such that $\dim_{\fq} (U_i) =n_i$.
\end{corollary}

\noindent Other examples arise from the maximum $1$-design investigated in Section \ref{sec:maxsdes}.

\begin{corollary}
Let $(U_1,\ldots,U_t)$ be a maximum $1$-design in $V=V(k,q^m)$. Then it is an optimal subspace design.
\end{corollary}
\begin{proof}
By Theorem \ref{th:boundhyper1scatt}, the ordered set $(U_1,\ldots,U_t)$ is also an $\left(k-1,\frac{tm(k-2)}{2}+1\right)_q$-subspace design.
Since the dimensions of the $U_i$'s is $\frac{mk}2$, direct computations show that equality in \eqref{eq:optimalmsrd} holds and hence it is an optimal subspace design.
\end{proof}

\begin{remark}
In Theorem \ref{th:subgeom1design}, the linear sets $L_{U_1},\ldots,L_{U_t}$ associated with the subspace exhibited for $m=2$ cover the entire space, and hence by Corollary \ref{cor:8.9} for every hyperplane $H$ if follows that
\[ \sum_{i=1}^t \dim_{\fq}(U_i\cap H)=(t-1) \frac{q^{\frac{m(k-2)}{2}}-1}{q-1}+\frac{q^{\frac{m(k-2)}{2}+1}-1}{q-1}. \]
Moreover, the system $(U_1,\ldots,U_t)$ defines a one-weight MSRD code because of Theorem \ref{th:connection}.
\end{remark}

\section{Dualities}\label{sec:dualities}

\subsection{Ordinary dual}\label{sec:orddual}

Let $\sigma : V \times V \longrightarrow \F_{q^m}$ be a non-degenerate reflexive sesquilinear form on $V=V(k,q^m)$ and define 
\[
\sigma':( {u}, {v}) \in V \times V \rightarrow \mathrm{Tr}_{q^m/q}(\sigma( {u}, {v})) \in \F_q.
\]
It is known that $\sigma'$ is a non-degenerate reflexive sesquilinear form on $V$, when $V$ is regarded as an $mk$-dimensional vector space over $\F_q$. Let $\tau$ and $\tau'$ be the orthogonal complement maps defined by $\sigma$ and $\sigma'$ on the lattices of the $\F_{q^m}$-subspaces and $\F_q$-subspaces of $V$, respectively. Recall that if $W$ is an $\F_{q^m}$-subspace of $V$ and $U$ is an $\F_q$-subspace of $V$ then $W^\tau$ is an $\F_{q^m}$-subspace of $V$, $U^{\tau'}$ is an $\F_q$-subspace of $V$, $\dim_{\F_{q^m}}( W^{\tau})+\dim_{\F_{q^m}}( W)=k$ and $\dim_{\fq}( U^{\tau'})+\dim_{\fq}( U)=mk$. It is easy to see that $\tau$ and $\tau'$ coincide when applied to $\F_{q^m}$-subspaces, that is $W^{\tau}=W^{\tau'}$ for each $\F_{q^m}$-subspace $W$ of $V$. Also, $U^{\tau'}$ is called the \textbf{dual} of $U$ (with respect to $\tau'$). The dual of an $\F_q$-subspace of $V$ does not depend on the choice of the non-degenerate reflexive sesquilinear forms $\sigma$ and $\sigma'$ on $V$. More precisely, consider two non-degenerate reflexive sesquilinear forms $\sigma_1$ and $\sigma_2$ and then consider $\tau_1'$ and $\tau_2'$ as above, then $U^{\tau_1'}$ and $U^{\tau_2'}$ are $\mathrm{\Gamma L}(k,q^m)$-equivalent, for any $\fq$-subspace $U$ of $V$ (see \cite[Section 2]{polverino2010linear}). If $W$ is an $s$-dimensional $\F_{q^m}$-subspace of $V$ and $U$ is a $h$-dimensional $\F_q$-subspace of $V$, then 
\begin{equation}\label{eq:dualord}
\dim_{\fq}(U^{\tau'} \cap W^{\tau}) -\dim_{\fq}(U \cap W) =km-h-ms.
\end{equation}
For more details see \cite[section 7]{taylor1992geometry}.

Let consider an ordered set $(U_1,\ldots,U_t)$ of $\F_q$-subspaces in $V=V(k,q^m)$.
The ordered set $(U_1^{\tau'},\ldots,U_t^{\tau'})$ will be called the \textbf{dual subspaces} of $(U_1,\ldots,U_t)$.
Furthermore, note that the dual subspaces of $(U_1^{\tau'},\ldots,U_t^{\tau'})$ coincide with the ordered set $(U_1,\ldots,U_t)$.

From Equation \eqref{eq:dualord} the next result immediately follows.

\begin{proposition}\label{prop:orddual}
Suppose that $(U_1,\ldots,U_t)$ is an $(s,A)_q$-subspace design in $V=V(k,q^m)$, with $\dim_{\fq} (U_i) =n_i$ for any $i$ and suppose that \[\dim_{\F_{q^m}}(\langle U_1^{\tau'},\ldots,U_t^{\tau'}\rangle_{\F_{q^m}})\geq k-s.\] 
Then the dual subspaces $(U_1^{\tau'},\ldots,U_t^{\tau'})$ of $(U_1,\ldots,U_t)$ is an $(k-s,A+t(k-s)m-N)_q$-subspace design where $N=\sum_{i=1}^t n_i$ and $\dim_{\fq} (U_i^{\tau'})=mk-n_i$. 
\end{proposition}

When the assumptions of the above proposition are satisfied, we say that $(U_1^{\tau'},\ldots,U_t^{\tau'})$ is the \textbf{dual subspace design} of $(U_1,\ldots,U_t)$.

Now we show that the property of being maximum $1$-design is preserved via duality.

\begin{theorem}\label{prop:dual1des}
Let $(U_1,\ldots,U_t)$ be a  maximum $1$-design in $V=V(k,q^m)$. Then the dual subspace design of $(U_1,\ldots,U_t)$ is a maximum $1$-design in $V=V(k,q^m)$.
\end{theorem}
\begin{proof}
By Theorem \ref{th:boundhyper1scatt}, it follows that
\[ \sum_{i=1}^t \dim_{\fq}(U_i\cap H)\in \left\{ \frac{tm(k-2)}2,\frac{tm(k-2)}2+1 \right\},\]
for every hyperplane $H$ of $V=V(k,q^m)$, i.e. $(U_1,\ldots,U_t)$ is also an  $\left(k-1,\frac{tm(k-2)}2+1\right)_q$-subspace design.
Since $\dim_{\F_{q^m}}(\langle U_1^{\tau'},\ldots,U_t^{\tau'}\rangle_{\F_{q^m}})\geq 1$ (with the above notation), by Proposition \ref{prop:orddual} the dual subspace design $(U_1^{\tau'},\ldots,U_t^{\tau'})$ of $(U_1,\ldots,U_t)$ (seen as $\left(k-1,\frac{tm(k-2)}2+1\right)_q$-subspace design) is a $(1,\frac{tm(k-2)}2+1+tm-t\frac{mk}2)_q$-subspace design, i.e. the dual subspace design of $(U_1,\ldots,U_t)$ is a $1$-design.
Moreover, since $\dim_{\fq}(U_i^{\tau'})=\frac{mk}2$, $(U_1^{\tau'},\ldots,U_t^{\tau'})$ is a maximum $1$-design.
\end{proof}

As a corollary we get a characterization of maximum $1$-design as a $(k-1,A)_q$-subspace design.

\begin{corollary}\label{cor:charactmax1desintermsofhyper}
Let $(U_1,\ldots,U_t)$ be an ordered set of $\fq$-subspaces in $V=V(k,q^m)$ having dimension $\frac{mk}2$.
Then $(U_1,\ldots,U_t)$ is a maximum $1$-design in $V$ if and only if it is a $\left(k-1,\frac{tm(k-2)}2+1\right)_q$-subspace design in $V$.
\end{corollary}
\begin{proof}
The if part follows by Proposition \ref{prop:spanentirespace} and Theorem \ref{th:boundhyper1scatt}.
For the converse, suppose that $(U_1,\ldots,U_t)$ is a $\left(k-1,\frac{tm(k-2)}2+1\right)_q$-subspace design, then its dual subspace design $(U_1',\ldots,U_t')$ of $(U_1,\ldots,U_t)$ is a maximum $1$-design because of Proposition \ref{prop:orddual}.
Now, Theorem \ref{prop:dual1des} implies the assertion.
\end{proof}

Corollary \ref{cor:charactmax1desintermsofhyper} gives a one-to-one correspondence between maximum $1$-designs and MSRD codes with certain parameters, i.e. maximum $1$-designs are optimal subspace designs.

\begin{theorem}\label{th:charmax1desMSRDwithpar}
Let $m,k$ be two integers such that $mk$ is even with $k\geq 2$.
Let $\mathrm{U}=(U_1,\ldots,U_t)$ be an ordered set of $\fq$-subspaces in $\F_{q^m}^k$ of dimension $\frac{mk}2$ such that \newline $\langle U_1,\ldots,U_t \rangle_{\F_{q^m}}=\F_{q^m}^k$. Then every code in $\Phi([\mathrm{U}])$ is an $[\mathbf{n},k,d]_{q^m/q}$ MSRD code with $\mathbf{n}=(mk/2,\ldots,mk/2)$ and 
\[ d=t\frac{mk}2-t\frac{m(k-2)}2-1=mt-1 \]
if and only if $(U_1,\ldots,U_t)$ is a maximum $1$-design.
\end{theorem}
\begin{proof}
Suppose that $\mathrm{U}$ defines an MSRD code, then by Theorem \ref{cor:MSRDoptimalsd} it follows that
\[ \sum_{i=1}^t \dim_{\fq}(U_i\cap H)\leq t\frac{mk}2-tm+1, \]
for every hyperplane $H$ of $\F_{q^m}^k$. This implies that $(U_1,\ldots,U_t)$ is a $\left(k-1,\frac{tm(k-2)}2+1\right)_q$-subspace design and hence by Corollary \ref{cor:charactmax1desintermsofhyper} it is also a maximum $1$-subspace design.
Conversely, if $(U_1,\ldots,U_t)$ is a maximum $1$-subspace design, then again applying Corollary \ref{cor:charactmax1desintermsofhyper} and Theorem \ref{cor:MSRDoptimalsd} we obtain that every code in $\Phi([\mathrm{U}])$ is an $[\mathbf{n},k,d]_{q^m/q}$ MSRD code.
\end{proof}

We recall an upper bound on the number of blocks of an MSRD code.
\begin{theorem} (see \cite[Theorem 6.12]{byrne2021fundamental}) \label{th:lenghtMSRD}
Let $\mathrm{C}$ be an $[\bfn,k,d]_{q^m/q}$ MSRD code with $n=n_1=\ldots=n_t \geq m$ and $d \geq 3$. Then
\begin{equation} \label{eq:boundont}
t \leq \left\lfloor \frac{d-3}{m} \right\rfloor +\left\lfloor \frac{q^{m}-q^{m\left\lfloor \frac{d-3}{m} \right\rfloor+m-d+3}+(q-1)(q^n+1)}{q^m-1} \right\rfloor.
\end{equation}
\end{theorem}
Theorem \ref{th:charmax1desMSRDwithpar} gives a new bound on the number of blocks that an MSRD code associated with a maximum $1$-design can have. Indeed, using the connection between maximum $1$-design and MSRD code, the number of blocks corresponds to the number of subspaces in a maximum $1$-design. This latter number has been bounded in Corollary \ref{cor:8.9} and we showed that it is sharp for some set of parameters (see Theorem \ref{th:subgeom1design}). 

\begin{corollary}
Let $\mathrm{C}$ be an $[\bfn,k,d]_{q^m/q}$ MSRD code with $\bfn=(mk/2,\ldots,mk/2)$ and $d=tm-1$. Then $t \leq \left\lfloor (q-1) \frac{q^{\frac{mk}{2}}+1}{q^m-1} \right\rfloor$. Moreover, if $t\geq 3$ this bound improves \eqref{eq:boundont} of Theorem \ref{th:lenghtMSRD}.
\end{corollary}
\begin{proof}
Let $(U_1,\ldots,U_t)$ be a system associated with $\mathrm{C}$. By Theorem \ref{th:charmax1desMSRDwithpar}, $(U_1,\ldots,U_t)$ is a maximum $1$-design and so by Corollary \ref{cor:8.9}, we get $t \leq (q-1) \frac{q^{\frac{mk}{2}}+1}{q^m-1}$ that proves the first part. Moreover, since $t \geq 3$ we have $\left\lfloor \frac{tm-4}{m} \right\rfloor \geq 1$ and so it follows that 
\[
\begin{array}{c l}
\left\lfloor \frac{(q-1)(q^{\frac{mk}{2}}+1)}{q^m-1} \right\rfloor &   \leq \left\lfloor \frac{tm-4}{m} \right\rfloor +\left\lfloor \frac{q^{m}-q^{m\left\lfloor \frac{tm-4}{m} \right\rfloor+m-tm+4}+(q-1)(q^{\frac{mk}{2}}+1)}{q^m-1} \right\rfloor,
\end{array}
\]
where the last quantity corresponds to the right hand side of the bound \eqref{eq:boundont}. This proves the last part.
\end{proof}

As a consequence, this bound shows that the bound in Theorem \ref{th:lenghtMSRD} is not sharp for these parameters.

\subsection{Delsarte dual}

In this section we provide another duality acting on subspace designs by using their connection with sum-rank metric codes.

Let $\mathrm{U}=(U_1,\ldots,U_t)$ be a non-degenerate $(s,A)_q$-subspace design in $V=V(k,q^m)$.
Suppose that $n_i=\dim_{\fq}(U_i)$ for every $i \in[t]$ and let $\mathbf{n}=(n_1,\ldots,n_t)$ and $N=n_1+\ldots+n_t$, and w.l.o.g. suppose that $n_1 \geq \cdots \geq n_t$.
Up to coordinatize $V$, we may assume that $V=\F_{q^m}^k$.
Let consider the class of sum-rank metric codes $[\mathrm{C}]=\Phi([\mathrm{U}])$. Let consider $\overline{\mathrm{C}} \in [\mathrm{C}]$. Let $\overline{\mathrm{C}}^{\perp}$ be the dual of $\overline{\mathrm{C}}$ defined as in \eqref{eq:dualsrmcodes} and suppose that $\overline{\mathrm{C}}^{\perp}$ is non-degenerate. We call $[\mathrm{U}']=\Psi([\overline{\mathrm{C}}^{\perp}])$ the \textbf{Delsarte dual class} of $[\mathrm{U}]$ and any ordered set $(U_1',\ldots,U'_t)$ such that $(U_1',\ldots,U'_t)\in \Psi([\overline{\mathrm{C}}^{\perp}])$ will be called a \textbf{Delsarte dual design} of $(U_1,\ldots,U_t)$.
This definition is well-posed because of Remark \ref{rk:equivcodes}.

In order to give a relation on the parameters of a subspace design and its dual we need the following property.

\begin{proposition}
Suppose that $(U_1,\ldots,U_t)$ is a non-degenerate subspace design in $V=V(k,q^m)$, with $\dim_{\fq} (U_i) =n_i$, for any $i$.\\
Let $M=\max\{\sum_{i=1}^t \dim_{\fq}(U_i\cap H) \colon H\mbox{ is an }\F_{q^m}\mbox{-hyperplane of }V\}$.
Then the Delsarte dual subspace design $(U_1',\ldots,U_t')$ of $(U_1,\ldots,U_t)$ has the following properties:
\begin{itemize}
    \item there exists $\sigma \in S_t$ such that $\dim_{\fq} (U_i')=n_{\sigma(i)}$ for any $i$;
    \item if $M'=\max\{\sum_{i=1}^t \dim_{\fq}(U_i'\cap H) \colon H\mbox{ is an }\F_{q^m}\mbox{-hyperplane of }V\}$ then
    \[ M' \geq \left\{
    \begin{array}{ll}
    N-M-2 & \text{if}\,\,m\geq \max\{n_1,\ldots,n_t\},\\
    2N-tm-M-2 & \text{if}\,\,m\leq n_1=\ldots=n_t,\\
    \end{array}
    \right.\]
    where $N=n_1+\ldots+n_t$.
\end{itemize}
\end{proposition}
\begin{proof}
Without loss of generality, we may suppose that $n_1 \geq \cdots \geq n_t$. The assertion follows by considering a sum-rank metric code $\mathrm{C}$ in $\Phi([(U_1,\ldots,U_t)])$ and then apply Theorem \ref{th:connection}, so that $(U_1,\ldots,U_t)$ can be seen as an $[\mathbf{n},k,d]_{q^m/q}$-system with $d=N-M$.
Clearly, $\mathrm{C}^\perp \in \Phi([(U_1',\ldots,U_t')])$ and then apply Theorem \ref{th:connection}, so that $(U_1',\ldots,U_t')$ can be seen as an $[\mathbf{n},k,d']_q$-system with $d'=N-M'$.
The assertion then follows by Theorem \ref{th:dualMSRD}.
\end{proof}

As a consequence we obtain the following corollary, also taking into account Theorem \ref{th:dualMSRD}.

\begin{corollary}
Suppose that $(U_1,\ldots,U_t)$ is a non-degenerate $(k-1,A)_q$-subspace design in $V=V(k,q^m)$, with $\dim_{\fq} (U_i) =n_i$ for any $i$, and suppose that \[A=\max\left\{\sum_{i=1}^t \dim_{\fq}(U_i\cap H) \colon H\mbox{ is an }\F_{q^m}\mbox{-hyperplane of }V\right\}.\]
Then the Delsarte dual subspace design $(U_1',\ldots,U_t')$ of $(U_1,\ldots,U_t)$ is an $(k-1,A')$-subspace design with
 \[ A' \geq \left\{
    \begin{array}{ll}
    N-A-2 & \text{if}\,\,m\geq \max\{n_1,\ldots,n_t\},\\
    2N-tm-A-2 & \text{if}\,\,m\leq n_1=\ldots=n_t,\\
    \end{array}
    \right.\]
    where $N=n_1+\ldots+n_t$. In particular, the Delsarte dual of an optimal subspace design is an optimal subspace design.
\end{corollary}

\section{More bounds and constructions}\label{sec:moreconstr}

In this section we will further explore bounds and constructions of subspace designs whose parameters satisfy equality in the provided bounds.

As a consequence of the definition of subspace evasive subspaces we obtain that the ordered set of subspace evasive subspaces form a subspace design with the following parameters.

\begin{proposition}
Let $(U_1,\ldots, U_t)$ be an ordered set of $\F_q$-subspaces of $V=V(k,q^m)$ such that $U_i$ is an $(s,r_i)_q$-evasive subspace for every $i \in \{1,\ldots,t\}$. Then $(U_1,\ldots, U_t)$ is an $(s,A)_q$-subspace design with $A=\sum_{i=1}^tr_i$.
\end{proposition}

Under certain assumptions, the bounds in Corollary \ref{cor:trivsubdes} may be improved. 

\begin{theorem}\label{th:boundparMSRD}
Let $(U_1,\ldots,U_t)$ be a $(k-1,A)_q$-subspace design of $V=V(k,q^m)$. Let $n= \dim_{\fq}(U_i)$ for any $i$. If $n \geq m$ then \
\begin{enumerate}
    \item if $A < tm(k-1)$ then $n \leq m + \frac{A}{t} -1$;
    \item if $A < \frac{tm}{k-1}+k-2$ then $tn \leq tm+A-k+1$. 
\end{enumerate}
\end{theorem}
\begin{proof}
By definition, $\dim_{\F_{q^m}}(\langle U_1,\ldots,U_t\rangle_{\F_{q^m}}) \geq k-1$. First suppose that \\ $\dim_{\F_{q^m}}(\langle U_1,\ldots,U_t\rangle_{\F_{q^m}}) = k-1$. Then $H=\langle U_1,\ldots,U_t\rangle_{\F_{q^m}}$ is an $\F_{q^m}$-hyperplane of $\F_{q^m}^k$ and so 
\[
tn = \sum_{i=1}^t \dim_{\F_q}(U_i \cap H) \leq A \leq A+tm-t,
\]
proving 1. of the assertion. Since $A \geq k-1$, by the assumption in 2. we have that $tm>k-1$ and so $tn \leq A <A+tm-k+1$ proving also 2. of the assertion. Assume now that $\langle U_1,\ldots,U_t\rangle_{\F_{q^m}}=\F_{q^m}^k$, then
$\mathrm{U}=(U_1,\ldots,U_t)$ is an $[\mathbf{n},k,d]_{q^m/q}$-system, with $\mathbf{n}=(n,\ldots,n)$ and
\[
d=N-\max\left\{\sum_{i=1}^t \dim_{\fq}(H \cap U_i) : H \ \F_{q^m}\mbox{-hyperplane of }\F_{q^m}^k\right\}, 
\]
where $N=tn$.

Since $(U_1,\ldots,U_t)$ is a $(k-1,A)_q$-subspace design, it follows that $d \geq tn-A$.

Let $[\mathrm{C}]=\Phi([\mathrm{U}])$. We can apply the Singleton bound (2. of Theorem \ref{th:sbadapted}) on the parameters of $[\mathrm{C}]$. 
Indeed, we have
\[
mk \leq n(tm-d+1),
\]
and so
\begin{equation} \label{eq:bounddesignhyper}
mk \leq n(tm-tn+A+1).
\end{equation}

To prove the first part, suppose for the contrary $n \geq m + \frac{A}{t}$. Substituting in \eqref{eq:bounddesignhyper} gives $mk \leq n$. This implies that $U_i=\F_{q^m}^k$, for every $i$ and then $tm(k-1) \leq A$, contradicting our assumption on $A$.

To prove the second part, first we observe that \eqref{eq:bounddesignhyper} can be rewritten as
\begin{equation} \label{eq:bounddesignhyperaggiustato}
mk+n^2t \leq -n(-tm-A-1)
\end{equation}
to the contrary suppose that $tn \geq tm+A-k+2$. Substituting in \eqref{eq:bounddesignhyperaggiustato}, we get

\[
mk+t\left(m+\frac{A-k+2}{t}\right)^2 \leq \left(-m-\frac{A-k+2}{t}\right)(-tm-A-1)
\]
which yields
\[
mk+\left(m+\frac{A-k+2}{t}\right)\left(-k+1\right) \leq 0,
\]
and so
\[
tm+\left(k-2\right)\left(k-1\right) \leq A(k-1),
\]
from which we obtain 
\[
\frac{tm}{k-1}+k-2 \leq A,
\]
contradicting our assumptions. 
\end{proof}

In the next we will see some constructions of subspace designs satisfying equality in the bounds of Theorem \ref{th:boundparMSRD}.

We start by constructing examples of subspace designs satisfying equality in 1. of Theorem \ref{th:boundparMSRD} when $A=t \alpha$, with $\alpha \geq (k-2)(m-1)+1$.

\begin{proposition}
If $A=t \alpha$, with $\alpha \geq (k-2)(m-1)+1$, then in $V(k,q^m)$ there exists a $(k-1,A)_q$-subspace design whose elements have dimension $m+\alpha-1$.
\end{proposition}
\begin{proof}
By i) of Proposition \ref{prop:constr12}, if $\alpha \geq (k-2)(m-1)+1$, in $V(k,q^m)$ there exists a $(k-1,\alpha)_q$-evasive subspace of dimension $m + \alpha -1$. The subspace design whose subspaces are $t$ copies of such subspace evasive subspace will give the subspace design with the desired parameters. 
\end{proof}

When $mk$ is even we can prove the sharpness of the first bound of Theorem \ref{th:boundparMSRD} also
for smaller value of $A$. 

\begin{proposition}
If $km$ is even and $A=t\alpha$ with $\alpha \geq \frac{km}{2}-m+1$ then there exists a $(k-1,A)_q$-design $(U_1,\ldots,U_t)$ in $V(k,q^m)$ such that $\dim_{\fq} (U_i)=m+\alpha-1$.
\end{proposition}
\begin{proof}
By ii) of Proposition \ref{prop:constr12}, when $\alpha \geq \frac{km}{2}-m+1$ there exists a $(k-1,\alpha)_q$-evasive subspace $S$ with $\dim_{\fq}(S)=m+\alpha-1$. The subspace design whose elements are $t$ copies of such subspace evasive subspace will give the subspace design with the desired parameters. 
\end{proof}

\begin{remark}
The existence of the above subspace evasive subspace relies on the the existence of scattered subspaces provided in \cite{ball2000linear,bartoli2018maximum,blokhuis2000scattered,csajbok2017maximum}.
\end{remark}

For the second bound in Theorem \ref{th:boundparMSRD}, we have the following construction.

\begin{theorem}
Let $t<q$. If $k-1 \leq A < \frac{tm}{k-1}+k-2$ and $k \leq m$ then there exists a $(k-1,A)_q$-subspace design $(U_1,\ldots,U_t)$ in $V=V(k,q^m)$ such that $\dim_{\fq} (U_i)= m+ \left\lfloor\frac{A-k+1}{t}\right\rfloor $.
\end{theorem}
\begin{proof}
From the assumptions on $A$, we get $A \leq tm +k-1$. Let $(U'_1,\ldots,U'_t)$ be a $(k-1)$-design such that $\dim_{\fq}(U'_i)=m \geq k$, cf.\ Theorem \ref{th:example2(n-1)des}. Since $\frac{A-k+1}{t} \leq m$, we can consider an $\F_q$-subspace $W'$ of dimension $\lfloor{\frac{A-k+1}{t}\rfloor}$ contained in a $1$-dimensional $\F_{q^m}$-subspace $\langle  {v} \rangle_{\F_{q^m}}$ of $V$, with $\langle  {v} \rangle_{\F_{q^m}} \cap U_i = \{ {0}\}$ for any $i\in \{1,\ldots,t\}$.
Such a $ {v}$ exists because of Remark \ref{rk:exv}, since the union of the related linear sets does not cover the entire space. So, $(U'_1 \oplus W',\ldots,U'_t \oplus W')$ is a $(k-1,A)_q$-subspace design such that $\dim_{\fq} (U'_i \oplus W')=m+\lfloor{\frac{A-k+1}{t}\rfloor}$.
\end{proof}

\section{Constructions from strong subspace designs}\label{sec:fromstrong}

In most of the applications strong subspace designs have been used to construct subspace designs. 
So, in this section we describe how to obtain a subspace design from a strong subspace design using different tricks. 

\subsection{Using subspace evasive subspaces}\label{sec:fromstrong1}

The first, and probably the most used, subspace design employs subspace evasive subspaces.

\begin{theorem} \label{teo:designplusevasive}
   Consider a strong $(s,A)$-subspace design $(V_1,\ldots,V_t)$ in $V=V(k,q^m)$.  
   Let $k_i=\dim_{\F_{q^m}}(V_i)$, for every $i\in\{1,\ldots,t\}$. 
    Let $S \subseteq V$ be an $\F_q$-subspace of dimension $d$ with the property that $S$ is a $(h,ch)_q$-evasive subspace for every $h \leq s$, with $c$ a positive number. Let $U_i=V_i \cap S$, for every $i$. If $\dim_{\F_{q^m}}(\langle U_1,\ldots,U_t \rangle_{\F_{q^m}})\geq s$, then $(U_1,\ldots,U_t)$ is an $(s,cA)_q$-subspace design and each $U_i$ has dimension at least $mk_i-km+d$.
\end{theorem}
\begin{proof}
By applying Grassmann's formula to the $U_i$'s we have that $\dim_{\fq}(U_i) \geq mk_i-km+d$. 
Now, let $W$ be a $\F_{q^m}$-subspace of dimension $s$ in $V$. By the assumptions we have
\begin{equation}\label{eq:strong}
\sum_{i=1}^t \dim_{\F_{q^m}}(V_i \cap W) \leq A.
\end{equation}
For each $i$, we have that $\dim_{\F_{q^m}}(W \cap V_i)=s_i \leq s$. Then by the assumptions on $S$, we have that $$\dim_{\fq}(W \cap U_i)=\dim_{\fq}((W \cap V_i) \cap S) \leq cs_i.$$ 
So,
$$\sum_{i=1}^t \dim_{\fq}(W \cap U_i ) \leq c\sum_{i=1}^t\dim_{\F_{q^m}}(W \cap V_i) \leq cA,$$
and the assertion follows by \eqref{eq:strong}.
\end{proof}

More generally and following the above proof one gets the following result.
 
\begin{theorem} 
   Consider a strong $(s,A)$-subspace design $(V_1,\ldots,V_t)$ in $V=V(k,q^m)$.  
   Let $k_i=\dim_{\F_{q^m}}(V_i)$, for every $i\in\{1,\ldots,t\}$. 
    Let $S \subseteq V$ be an $(s,r)_q$-evasive subspace. Let $U_i=V_i \cap S$, for every $i$. If $\dim_{\F_{q^m}}(\langle U_1,\ldots,U_t \rangle_{\F_{q^m}})\geq s$, then $(U_1,\ldots,U_t)$ is an $(s,t(r-s)+A)_q$-subspace design and each $U_i$ has dimension at least $mk_i-km+d$.
\end{theorem}

\begin{remark}
When $c=1$ in Theorem \ref{teo:designplusevasive}, $S$ is an $s$-scattered subspace in $V$. 
\end{remark}

\subsection{Intermediate fields}\label{sec:fromstrong2}

Another way to get subspace designs from strong subspace designs regards the use of intermediate fields \cite{guruswami2021lossless}.

To see this, first observe that $\F_{q}^k$ is an $s$-scattered subspace in $\F_{q^m}^k$ for any $s$.

\begin{proposition} \label{prop:intermediate}
Let $c$ be an integer such that $m \lvert c$. Let $(V_1,\ldots,V_t)$ be a strong $(s,A)$-subspace design in $\F_{q^m}^k$. Then $(V_1,\ldots,V_t)$ is an $(s,mA)_q$-subspace design in $\F_{q^c}^k$.
\end{proposition}
\begin{proof}
Let $W$ be an $s$-dimensional $\F_{q^c}$-subspace of $\F_{q^c}^k$. Since $V_i \subseteq \F_{q^m}^k$ we have that $W \cap V_i=(W \cap \F_{q^m}^k) \cap V_i$. Let $W'=W \cap \F_{q^m}^k$. Since $\F_{q^m}^k$ is an $s$-scattered $\mathbb{F}_{q^m}$-subspace of $\F_{q^c}^k$, then $\dim_{\F_{q^m}} (W') \leq s$ and so
\[
\sum_{i=1}^t\dim_{\fq}(V_i \cap W) =\sum_{i=1}^tm\dim_{\F_{q^m}}(V_i \cap W) =\sum_{i=1}^tm\dim_{\F_{q^m}}(V_i \cap W') \leq mA.
\]
\end{proof}


\subsection{High-degree places}\label{sec:fromstrong3}

To construct subspace designs we can also use high-degree places, as done by in \cite[Section 4.2]{guruswami2021lossless} for a fixed strong subspace design.
Consider a real number $\delta \in (0,1)$, and let $h=\delta k m$, with $m \leq q-1$. Let $(V_1,\ldots,V_t)$ be a strong $(r,A(r))$-subspace design (here, we explicitly write the dependence of $A$ on $r$ by $A(r)$) in $\F_q^h$, with $\dim_{\F_q}(V_i)=n_i$. 
We will work in the following framework.\\
There is an $\F_q$-linear isomorphism from $\F_q^h$ to $\F_q[x]_{<h}$, so that we can see $V_1,\ldots,V_t$ as $\F_q$-subspaces of $\F_q[x]_{<h}$.  
Now, let
\begin{itemize}
    \item $\zeta$ be a primitive root of the finite field $\F_q$;
    \item $\tau$ be the $\F_q$-automorphism of the function field $\F_q(x)$ mapping $x$ to $\zeta x$;
    \item $p(x)$ be an irriducible polynomial of degree $d$ such that $p,\tau p,\ldots,\tau^{k-1} p$ are pairwise coprime, where $\tau^i p:=\tau^i(p(x))=p(\zeta^i x)$.
\end{itemize}

Note that $(\tau^j p)$ is a maximal ideal of $\F_q[x]$ (\textbf{places} of the function field $\F_q(x)$) then $\F_q[x]/(\tau^{j} p)\cong \F_{q^m}$ and $\F_q[x]/(p)\times\ldots\times \F_q[x]/(\tau^{k-1} p)\cong \F_{q^m}^k$ (as $\fq$-vector spaces). 
Consider the $\F_q$-linear map 
\[
\begin{array}{lccl}
    \pi \colon & \F_q[X]_{<h} & \longrightarrow & \F_q[x]/(p)\times\ldots\times \F_q[x]/(\tau^{k-1}p ) \\
      & f(x) & \longmapsto & (f(p),\ldots,f(\tau^{k-1} p)),
\end{array}
\]
where $f(\tau^{j} p)$ is the residue of $f$ in the residue field $\F_q[x]/(\tau^{j} p)$, i.e. $f(\tau^{j} p)$ is the lateral of $f$ in $\F_q[x]/(\tau^{j} p)$.
Let $f(x) \in \F_q[x]_{<h}$ such that $\pi(f(x))=0$. Then $\tau^i p$ divides $f(x)$, for any $i\in\{0,\ldots,k-1\}$. 
Since the ideals $(p),(\tau p),\ldots,(\tau^{k-1}p)$ are pairwise coprime, and $\deg(f(x)) < h <km$, we have that $f(x)$ is the zero polynomial and so $\pi$ is injective. We define 
\begin{equation} \label{eq:placesstrong}
U_i=\pi(V_i)=\{(f(p),\ldots,f(\tau^{k-1} p)) \colon f \in V_i\}\subseteq \F_{q^m}^k,
\end{equation}
for $i\in [t]$.
Since $\pi$ is an injective $\F_q$-linear map, it follows that $\dim_{\F_q}(U_i)=n_i$.
Now, following the proof of \cite[Proposition 4.6]{guruswami2021lossless} and by replacing the considered strong subspace design with any strong subspace design, we obtain the following result.

\begin{theorem} \label{th:highdegree}
Let $\delta$ be a real number in $(0,1)$.  Let $(V_1,\ldots,V_t)$ be a strong $(r,A(r))$-subspace design of $\F_q^h$. Suppose that there exists a positive integer $s$ such that $r=\left\lfloor \frac{s}{1-\delta} \right\rfloor+1$. Define $(U_1,\ldots,U_t)$ as in \eqref{eq:placesstrong} and suppose that $\dim_{\F_{q^m}}(\langle U_1,\ldots,U_t\rangle_{\F_{q^m}})\geq s$, then $(U_1,\ldots,U_t)$ forms an $\left(s,A\left(\left\lfloor \frac{s}{1-\delta} \right\rfloor+1\right)\right)_q$-subspace design in $\F_{q^k}^m$, with $\dim_{\F_q}(U_i)=n_i$. In particular, if $(V_1,\ldots,V_t)$ is a strong $(r,A(r))$-subspace design for any $r \leq r'$, then $(U_1,\ldots,U_t)$ is an $\left(s,A\left(\left\lfloor \frac{s}{1-\delta} \right\rfloor+1\right)\right)_q$-subspace design for any $s<(1-\delta)r'$.
\end{theorem}

\begin{remark}
In the statement of \cite[Proposition 4.6]{guruswami2021lossless}, $r=\left\lfloor \frac{s}{1-\delta} \right\rfloor$, but when $\frac{s}{1-\delta}$ is a positive integer, then Equation (6) in the proof of \cite[Proposition 4.6]{guruswami2021lossless} does not hold.
However, the asymptotics is the same and hence this does not effect their results.
\end{remark}

\subsection{Cameron-Liebler sets}

We conclude this section by describing a way to obtain strong subspace design from well-known objects in finite geometry known as Cameron-Liebler sets of $n$-dimensional projective subspace of $\PG(k,q)=\PG(V,\fq)$ introduced in \cite{rodgers2018cameron}, generalizing Cameron-Liebler sets of lines in the projective space $\PG(3,q)$ originally introduced by Cameron and Liebler in \cite{cameron1982tactical}.
Indeed, they introduced specific line classes $\mathcal{L}$ of size $x(q^2+q+1)$ in $\PG (3,q)$ with the property that every line spread in $\PG (3,q)$ has $x$ lines in common with $\mathcal{L}$; such a $x$ is called the \emph{parameter} of $\mathcal{L}$.

In \cite{blokhuis2019cameron} are proven several equivalent properties that defines Cameron-Liebler sets. We choose one that is useful for us.

\begin{definition}(see \cite[Theorem 2.9]{blokhuis2019cameron})
A set of $n$-dimensional projective subspaces $\mathcal{L}$ of $\PG(k,q)$ with $k \geq 2n+1$ is a \textbf{Cameron-Liebler set} of $\PG(k,q)$ with parameter $x=\lvert \mathcal{L} \rvert \binom{k}{n}_q^{-1}$ if and only if 
for a given $i \in \{1,\ldots,n+1\}$ and any $n$-dimensional projective subspaces $\pi$ of $\PG(k,q)$, the numbers of elements of $\mathcal{L}$, meeting $\pi$ in a $(n-i)$-dimensional projective subspace is given by:
\[
\begin{cases}
w_i=\left( (x-1) \frac{q^{n+1}-1}{q^{n-i+1}-1}+q^i\frac{q^{k-n}-1}{q^{i}-1} \right)q^{i(i-1)} \binom{k-n-1}{i-1}_q\binom{n}{i}_q & \mbox{ if }\pi \in \mathcal{L}, \\
w'_i=x\binom{k-n-1}{i-1}_q\binom{n+1}{i}_q q^{i(i-1)} & \mbox{ if }\pi \notin \mathcal{L} .
\end{cases}
\]
\end{definition}

Since there is a one-to-one correspondence between $i$-dimensional projective subspace of $\PG(k,q)$ and the $(i+1)$-dimensional $\fq$-subspaces of $V$, for any $i\in\{0,\ldots,k\}$, we can revise the definition of Cameron-Liebler sets from a (strong) subspace design point of view.

\begin{theorem} \label{th:camerondesign}
Let $\mathcal{L}=\{\mathcal{S}_1,\ldots,\mathcal{S}_t\}$ be a Cameron-Liebler set of $n$-dimensional projective subspaces of $\PG(k,q)=\PG(V,\fq)$ with parameter $x$ and $k\geq 2n+1$. 
Let $V_1,\ldots,V_t$ be the $\fq$-subspaces such that $\mathcal{S}_i=\PG(V_i,\fq)$. 
Then $(V_1,\ldots,V_t)$ is a strong $(n+1,A)$-subspace design of $V$ with \[A=n+1+\sum_{i=1}^{n+1}w_i(n-i+1),\] and $\dim_{\fq}(V_i)=n+1$ for any $i \in\{1,\ldots,t\}$.
\end{theorem}
\begin{proof}
Let $W$ be an $(n+1)$-dimensional subspace of $V$ and let $\mathcal{W}=\PG(W,\F_q)$. We have
\[
\sum_i^t \dim_{\fq} (V_i \cap W)=\sum_i^t (\dim (\mathcal{S}_i \cap  \mathcal{W})+1).
\]
In $\mathcal{L}$ there are $w_i$ subspaces that intersect $ \mathcal{W}$ in a $(n-i)$-dimensional projective subspace in the case in which $\mathcal{W} \in \mathcal{L}$, otherwise there are $w'_i$ subspaces that intersect $ \mathcal{W}$ in a $(n-i)$-dimensional projective subspace.
This means that 
\[
\sum_{i=1}^t \dim_{\fq} (V_i \cap W) = \begin{cases}
\sum_{i=1}^{n+1}w_i (n+1-i) +n+1, &  \mbox{ if } \mathcal{W} \in \mathcal{L}, \\
\sum_{i=1}^{n+1}w'_i (n+1-i), &\mbox{ if }\mathcal{W} \notin \mathcal{L}.
\end{cases}
\]
Since $w'_i \leq w_i$, for $i \geq 1$, the assertion follows.
\end{proof}

\begin{remark}
Note that the $(n+1,A)_q$-subspace design given by Theorem \ref{th:camerondesign}, is not an $(n+1,A')_q$-subspace design if $A' < A$.
\end{remark}

Now we recall some examples of Cameron-Liebler sets (see e.g. \cite{rodgers2018cameron}): let $k=2n+1$ then
\begin{itemize}
    \item the set of all $n$-dimensional projective subspaces of $\PG(k,q)$ containing a fixed point of $\PG(k,q)$ is a Cameron-Liebler set with parameter $1$.
    \item The set of all $n$-dimensional projective subspaces of $\PG(k,q)$ contained in a fixed hyperplane of $\PG(k,q)$ is a Cameron-Liebler sets with parameter $1$.
\end{itemize}

More examples can be found in e.g. \cite{rodgers2018cameron,feng2021cameron} and also there known non-existence results, see e.g. \cite{beule2022cameron,de2022degree,metsch2010non}.

\section{Cutting designs and minimal sum-rank metric codes}\label{sec:cutting}

In this section we study subspace designs with a special property that involves the hyperplanes of the ambient space, which we will call cutting designs. The name arises from \emph{cutting blocking sets}, recently introduced by Bonini and Borello in \cite{Bonini2021minimal}, with the aim of constructing minimal codes. 
This structure has been also studied before the paper \cite{Bonini2021minimal}, under other names such as strong blocking sets and generator sets, in connection with saturating sets and covering codes by Davydov, Giulietti, Marcugini and Pambianco \cite{davydov2011linear} and higgledy-piggledy line arrangements by Fancsali and Sziklai \cite{fancsali2014lines} (see also \cite{heger2021short} and \cite{fancsali2016higgledy}).
In this section we propose a generalization of this notion to subspace designs which turns out to be connected with minimal sum-rank metric codes, as we will see later.

\begin{definition}\label{def:cutdes}
Let $(U_1,\ldots,U_t)$ be an ordered set of $\fq$-subspaces in $V=V(k,q^m)$.
If for any $\F_{q^m}$-hyperplanes $H,H' \subseteq \F_{q^m}^k$ such that 
\[U_i \cap H \subseteq U_i \cap H'\,\, \mbox{for every}\,\, i\in\{1,\ldots,t\}\] 
implies $H=H'$, then $(U_1,\ldots,U_t)$ is called a \textbf{cutting design}.
\end{definition}

We now provide an easier to handle characterization of cutting designs, which follows the proof of \cite[Proposition 3.3]{Alfarano2020ageometric} (see also \cite[Theorem 3.5]{Bonini2021minimal}).

\begin{proposition}\label{prop:charcutt1}
Consider an ordered set $(U_1,\ldots,U_t)$ of $\fq$-subspaces in $V=V(k,q^m)$.
Then $(U_1,\ldots,U_t)$ is a cutting design if and only if for every $\F_{q^m}$-hyperplane $H$ of $V$ we have $\langle U_1 \cap H,\ldots,U_t \cap H \rangle_{\F_{q^m}}=H$.
\end{proposition}
\begin{proof}
Suppose that $(U_1,\ldots,U_t)$ is a cutting design and, by contradiction, suppose that there exists an $\F_{q^m}$-hyperplane $H$ of $V$ such that $\langle U_1 \cap H,\ldots,U_t \cap H \rangle_{\F_{q^m}}=X \subsetneq H$. Let $H'$ be an $\F_{q^m}$-hyperplane of $V$ containing $X$ different from $H$. So we have that $U_i \cap H \subseteq U_i \cap H'$, for every $i$. Since $H' \neq H$, this yields a contradiction.

Conversely, assume that for every $\F_{q^m}$-hyperplane $H$ we have $\langle U_1 \cap H,\ldots,U_t \cap H \rangle_{\F_{q^m}}=H$.
Let $H$ and $H'$ two hyperplanes such that 
\[U_i \cap H \subseteq U_i \cap H'\,\, \mbox{for every}\,\, i\in[t].\]
This implies that
\[ \langle U_1 \cap H,\ldots,U_t \cap H \rangle_{\F_{q^m}}\subseteq \langle U_1 \cap H',\ldots,U_t \cap H' \rangle_{\F_{q^m}}, \]
that is $H=H'$.
\end{proof}

\begin{remark}\label{rk:nondegcuttingdes}
By Proposition \ref{prop:charcutt1}, it follows that a cutting design $(U_1,\ldots,U_t)$ of $V(k,q^m)$ is also a non-degenerate $(k-1,A)_q$-design, for a certain $A\geq k-1$. 
Actually, it has a more strong property:  $\sum_{i=1}^t\dim_{\F_q}(U_i \cap H) \geq k-1$ for every $H$ hyperplane of $V$. 
\end{remark}

As the previous remark suggests, cutting designs are subspace designs with special pattern of intersections with hyperplanes. In particular, examples of cutting designs can be obtained from subspace designs with the property that the sum of the dimension of intersections with hyperplanes is a non-zero constant; this will allow us to give constructions of cutting designs (cf. Construction \ref{con:cuttingorbital}).

\begin{theorem}\label{th:costantimpliescutting}
Let $(U_1,\ldots,U_t)$ be an order set of $\F_q$-subspaces of $V=V(k,q^m)$. Suppose that there exists a positive integer $c$ such that $\sum_{i=1}^t\dim_{\F_{q}}(U_i \cap H)=c$ for every hyperplane of $V$. Then $c \geq k-1$ and $(U_1,\ldots,U_t)$ is a non-degenerate cutting design.
\end{theorem}
\begin{proof}
Let $n_i=\dim_{\F_q}(U_i)$. Let start by proving that $(U_1,\ldots,U_t)$ is a non-degenerate subspace design. Suppose by contradiction that $\langle U_1,\ldots,U_t\rangle_{\F_{q^m}} \subseteq H$, for a hyperplane $H$. This means that $U_i \subseteq H$ for every $i$, and hence $c=\sum_{i=1}^t \dim_{\F_q}(U_i \cap H)=\sum_{i=1}^t n_i$. Suppose, without loss of generality that $n_1>0$ and let $v \in U_1 \setminus \{0\}$. Let $H'$ be a hyperplane of $V$ such that $v \notin H'$. Then $ \sum_{i=1}^t\dim_{\F_q}(U_i \cap H') < \sum_{i=1}^t n_i=c$, a contradiction to the fact that $ \sum_{i=1}^t\dim_{\F_q}(U_i \cap H')=c$. Now we show that $(U_1,\ldots,U_t)$ is a cutting design.
Let $H$ and $H'$ be two hyperplanes of $V$ such that 
\[U_i \cap H \subseteq U_i \cap H'\,\, \mbox{for every}\,\, i\in[t].\] 
So, we get that $\dim_{\F_q}(U_i \cap H) \leq \dim_{\F_q}(U_i \cap H')$ for every $i$, and since
\[
c = \sum_{i=1}^t\dim_{\F_q}(U_i \cap H) \leq \sum_{i=1}^t\dim_{\F_q}(U_i \cap H') =c,
\]
it follows that $\dim_{\F_q}(U_i \cap H) = \dim_{\F_q}(U_i \cap H')$ and therefore $U_i \cap H=U_i \cap H'$ for every $i$. In particular, $U_i \cap H=U_i \cap H'=U_i \cap H \cap H' \subseteq H \cap H'$ for any $i$. Suppose that $H \neq H'$, this means that $H \cap H'$ is a $(k-2)$-dimensional $\F_{q^m}$-subspace of $V$. Let $U_i'=U_i \cap H=U_i \cap H'$ for every $i$. Therefore, there exists an $\fq$-subspace $U_i''$ of $V$ such that $U_i=U_i' \oplus U_i''$ for every $i$ and note that $U_i'' \cap H=U_i'' \cap H'=\{0\}$. Now, if all the $U_i''$ had dimension $0$ for every $i$ then $U_i \subseteq H$ for every $i$, a contradiction since $(U_1,\ldots,U_t)$ is a non-degenerate subspace design. Hence, without loss of generality, suppose that $U_1'' \neq \{0\}$ and let $u \in U_1'' \setminus \{0\}$. Let $H''$ be the hyperplane $(H \cap H') \oplus \langle u \rangle_{\F_{q^m}}$. As a consequence, $\dim_{\F_q}(U_1')=\dim_{\F_q}(U_1 \cap H \cap H') < \dim_{\F_q}(U_1 \cap H'')$.
Hence, we have
\[
c= \sum_{i=1}^t\dim_{\F_q}(U_i \cap H)=  \sum_{i=1}^t\dim_{\F_q}(U'_i) < \dim_{\F_q}(U_1 \cap H'')+ \sum_{i=2}^t\dim_{\F_q}(U'_i) \]\[\leq  \sum_{i=1}^t\dim_{\F_q}(U_i \cap H'')=c,
\]
a contradiction, and so $H=H'$.
\end{proof}

By Remark \ref{rk:nondegcuttingdes}, a cutting design also defines a system and clearly, if in an equivalence class of an $[ \mathbf{n},k]_{q^m/q}$ system there is a cutting design (that is the set of the entries of the $t$-ple is a cutting design), then all the elements of equivalence class are cutting design as well.
Therefore, an $[\mathbf{n},k]_{q^m/q}$ system with the property that the set of its entries forms a cutting design will be called a \textbf{cutting $[\mathbf{n},k]_{q^m/q}$ system}. This latter notion can be very useful when we need to enlighten the dimensions of the subspaces of a cutting design.

\subsection{Motivation}

We now will show how the notion of cutting design extends previously known notions of cutting blocking sets.
For a survey on blocking sets we refer to \cite{blokhuis2011blocking}.

\begin{definition}
Let $g,r,k$ be positive integers with $r<k$. A \textbf{$g$-fold $r$-blocking set} in $\PG(k,q)$ is a set $\mathcal{M} \subseteq \PG(k,q)$ such that for every $(k-r)$-subspace $\mathcal{S}$ of $\PG(k,q)$ we have 
\[
\lvert \mathcal{S} \cap \mathcal{M} \rvert \geq g.
\]
When $r=1$, we refer to $\mathcal{M}$ as \textbf{$g$-fold blocking set}. When $g=1$, we will refer to it as an \textbf{$r$-blocking set}. When $r=g=1$, $\mathcal{M}$ is simply a \textbf{blocking set}.
\end{definition}

A cutting $r$-blocking set is a blocking set that, roughly speaking, captures the structure of all the $(k-r)$-subspaces; see \cite[Definition 3.4]{Bonini2021minimal}.

\begin{definition}
Let $r,k$ be positive integers with $r<k$. An $r$-blocking set $\mathcal{M}$ in $\PG(k,q)$ is \textbf{cutting} if for every pair of $(k-r)$-subspaces $\mathcal{S}, \mathcal{S}'$ of $\PG(k,q)$ we have 
\[
\mathcal{M} \cap \mathcal{S} \subseteq \mathcal{M} \cap \mathcal{S}' \Longleftrightarrow \mathcal{S} = \mathcal{S}'.
\]
\end{definition}

Equivalently, 

\begin{proposition}(see \cite[Proposition 3.3]{Alfarano2020ageometric})\label{prop:charcuttingblockingset}
An $r$-blocking set $\mathcal{M} \subseteq \PG(k,q)$ is cutting if and only if for every $(k-r)$-space $\mathcal{S}$ of $\PG(k,q)$ we have $\langle \mathcal{M} \cap \mathcal{S}\rangle = \mathcal{S}$.
\end{proposition}

The $q$-analogue of the notion of a cutting blocking set is the linear cutting blocking set, introduced by Alfarano, Borello, Neri and Ravagnani in \cite[Definition 5.3]{alfarano2022linear}.

\begin{definition} \label{def:linearcutt}
An $[n,k]_{q^m/q}$ system $U$ is called a \textbf{linear cutting blocking set} if for any $\F_{q^m}$-hyperplanes $H,H' \subseteq \F_{q^m}^k$ we have $U \cap H \subseteq U \cap H'$ implies $H=H'$.
\end{definition}

If we need to underline the involved parameters of $U$,  we will also say that $U$ is a \textbf{linear cutting $[n,k]_{q^m/q}$ blocking set}.

\begin{remark}
If $U$ is a linear cutting $[n,k]_{q^m/q}$ blocking set, then $L_U$ is a cutting $1$-blocking set of $\PG(k-1,q^m)$.
\end{remark}

\subsection{Constructions and bounds}\label{sec:constrcutdes}

We now list examples of cutting designs which can be obtained from the known examples of cutting blocking sets. Some of them can be obtained easily, but for others more attention needs to be paid.
However, all of them are consequences of the following result where we show that the cutting design property can be required directly on the union of the linear sets associated with the design. 

\begin{proposition}\label{prop:charcutt2}
Consider an ordered set $(U_1,\ldots,U_t)$ of $\fq$-subspaces in $V=V(k,q^m)$.
Then $(U_1,\ldots,U_t)$ is a cutting design if and only if $L_{U_1} \cup \ldots \cup L_{U_t}$ is a cutting blocking set.
\end{proposition}
\begin{proof}
Suppose that $(U_1,\ldots,U_t)$ is a cutting design. Let $\mathcal{H}=\PG(W,\F_{q^m})$ be a hyperplane of $\PG(V,\F_{q^m}).$ By Proposition \ref{prop:charcutt1}, $U_1 \cap W \cup \ldots \cup U_t \cap W$ contains an $\F_{q^m}$-basis of $W$ and so $L_{U_1\cap W}\cup \ldots \cup L_{U_t\cap W}$ contains $k-1$ points of $\mathcal{H}$ such that their span is $\mathcal{H}$. Therefore,
\[
 \mathcal{H} \supseteq \langle  \mathcal{H}\cap (L_{U_1} \cup \ldots \cup L_{U_t}) \rangle =\langle  L_{U_1\cap W}\cup \ldots \cup L_{U_t\cap W} \rangle \supseteq \mathcal{H}.
\]
The assertion follows by Proposition \ref{prop:charcuttingblockingset}.
Conversely, suppose that $L_{U_1} \cup \ldots \cup L_{U_t}$ is a cutting blocking set. Let $W$ be an $\F_{q^m}$-hyperplane of $V$. Let $\mathcal{H}=\PG(W,\F_{q^m})$. By Proposition \ref{prop:charcuttingblockingset}, we have that $ L_{U_1\cap W}\cup \ldots \cup L_{U_t\cap W}=\mathcal{H}\cap (L_{U_1} \cup \ldots \cup L_{U_t}) $ contains a set of points $\{P_1=\langle v_1 \rangle, \ldots, P_{k-1}=\langle v_{k-1} \rangle\}$ of $\mathcal{H}$ such that $\mathcal{H}=\langle P_1,\ldots, P_{k-1} \rangle$ and $v_i \in U_1 \cup \ldots \cup U_t$ for every $i$. This implies that $ \langle  W \cap {U_1},  \ldots, W \cap {U_t} \rangle_{\F_{q^m}}=W$, and so the assertion follows by Proposition \ref{prop:charcutt1}.
\end{proof}

\begin{construction}[From a linear cutting blocking set]
If the ordered set $(U_1,\ldots,U_t)$ of $\fq$-subspaces in $V=V(k,q^m)$ is such that at least one of the $U_i$'s is a linear cutting blocking set, then $(U_1,\ldots,U_t)$ is a cutting design. 
\end{construction}

The above example provides a lot of constructions of cutting designs. However, interesting constructions of cutting designs are those in which the \emph{cutting property} is not satisfied individually by the elements of the design.
The next examples provides constructions of cutting designs not containing any linear cutting blocking sets.

\begin{construction}[From rational normal tangent set, \cite{fancsali2014lines}]\label{con:rnts}
Suppose that \\ $p=\mathrm{char}(\fq)\geq k$.
Choose $2k-1$ distinct points on the rational normal curve in $\PG(k-1,q^m)$ and consider the set $\{\ell_1,\ldots,\ell_{2k-1}\}$ of the tangent lines at these points.
For each line $\ell_i$ consider its sets of points $\{\la  {v}_{i,1} \ra_{\F_{q^m}},\ldots,\la  {v}_{i,q^m+1} \ra_{\F_{q^m}}\}$. Define
\[ U_{i,j}=\langle  {v}_{i,j} \rangle_{\fq}, \]
for every $i \in [2k-1]$ and $j\in [q^m+1]$.
The collection of the $U_{i,j}$'s is a cutting design.
\end{construction}

\begin{construction}[From tetrahedron, \cite{Alfarano2020ageometric,bartoli2021inductive,lu2021parameters}]
Consider $k$ points $P_1,\ldots,P_k$ in general position in $\PG(k-1,q^m)$ and let $\ell_{i,j}=\la P_i,P_j\ra$ for any $i,j \in [k]$ with $i\ne j$.
The collection of the $U_{i,j}$'s defined as in Construction \ref{con:rnts} is a cutting design.
\end{construction}

In the above two constructions all the subspaces of the design have dimension one and their number is quite large. In the next we will focus on showing examples of \emph{shorter} cutting design whose elements have higher dimension.
The following is an easy consequence of Theorem \ref{th:costantimpliescutting}.

Let $\mathcal{G}$ be a subgroup of $\GL(k,q^m)$, and consider the action $\phi_{\mathcal{G}}$ of $\mathcal{G}$ on $\F_{q^m}^k\setminus\{0\}$ induced by the one of $\GL(k,q^m)$, that is
$$ \begin{array}{rccl} \phi_{\mathcal{G}} : & \mathcal{G}\times (\F_{q^m}^k\setminus\{0\}) & \longrightarrow & \F_{q^m}^k\setminus\{0\}\\
&(A,v)& \longmapsto & vA. 
\end{array}$$
For any $n$ and $r$ such that $r$ divides $n$, this action naturally induces an action also on the $n$-dimensional $\F_{q^r}$-subspaces of $\F_{q^m}^k$ with kernel $ \mathcal{G} \cap \mathbb{D}_{q^r}$, where $\mathbb{D}_{q^r}=\{\alpha I_k \colon \alpha \in \F_{q^r}^*\}$. In this way, we can consider the action of the group $\overline{\mathcal{G}}=\mathcal{G}/(\mathcal{G} \cap \mathbb{D}_{q^r})$ on the $n$-dimensional $\F_{q^r}$-subspaces of $\F_{q^m}^k$, that we denote by $\phi_{\mathcal{G}}^{r,n}$. Finally, we say that $\mathcal{G} \leq \GL(k,q^m)$ is transitive if the action $\phi_{\mathcal{G}}^{m,1}$ is transitive.
(see \cite[Section 6.1]{neri2021geometry}).

\begin{construction}\label{con:cuttingorbital}[Orbital construction, \cite{neri2021geometry}]
Let $U$ be an $\F_q$-subspace of $\F_{q^m}^k$ with $\dim_{\F_q}(U)=n$. Let $\mathcal{G} \leq \GL(k,q^m)$ be a transitive subgroup and let $\mathcal{O}=(\phi_{\mathcal{G}}^{1,n}(A,U))_{A \in \overline{\mathcal{G}}}$ be the orbit (counting possible repetition) of the action of $\phi_{\mathcal{G}}^{1,n}$. When $\mathcal{G}$ is the Singer subgroup of $\mathrm{GL}(k,q^m)$ we call the orbit $\mathcal{O}$ an  $n$-\textbf{simplex}.
The orbit $\mathcal{O}=(U_1,\ldots,U_t)$ obtained in the previus way starting by a transitive subgroup $\mathcal{G}\le\GL(k,q^m)$ has the property that there exists a positive integer $c$ such that $\sum_{i=1}^t\dim_{\F_q}(U_i \cap H)=c$ for every hyperplane $H$ of $\F_{q^m}^k$, see \cite[Theorem 6.2]{neri2021geometry}. Theorem \ref{th:costantimpliescutting} implies that $(U_1,\ldots,U_t)$ is a cutting $[(n,\ldots,n),k]$ system. 
\end{construction}

\begin{construction} \label{con:cuttingsubgeometries}
The ordered set $(U_1,\ldots,U_t)$ obtained from Theorem \ref{th:subgeom1design} when $m=2$ gives a cutting design by Theorem \ref{th:projproperties}, as a consequence of Theorem \ref{th:costantimpliescutting}.
\end{construction}

\begin{remark}
We can prove that the ordered set $(U_1,\ldots,U_t)$ provided in Construction \ref{con:cuttingorbital} or Construction \ref{con:cuttingsubgeometries}, is a cutting design also by using Theorem Proposition \ref{prop:charcutt2}, since $L_{U_1}\cup \ldots \cup L_{U_t}=\PG(k-1,q^m)$ in both the cases. Indeed, in Construction \ref{con:cuttingorbital} gives a cutting design as the group used is transitive and Construction \ref{con:cuttingsubgeometries} gives a cutting design since the considered subgeometries form a partition of the entire projective space. It would be interesting to find an ordered set $(U_1,\ldots,U_t)$ as in the hypothesis of Theorem \ref{th:costantimpliescutting} such that $L_{U_1} \cup \cdots \cup L_{U_t} \neq \PG(k-1,q^m)$ or prove that this is not possible. 
\end{remark}

The following construction arises again by subgeometries, but in this case they do not give a partition of the projective space. 

\begin{construction}[From subgeometries, \cite{bartoli2020cutting}]
Let $\Sigma_1=L_{U_1}, \Sigma_2=L_{U_2}$ and $\Sigma_3=L_{U_3}$ three $q$-order subgeometries in $\PG(3,q^3)$ chosen as in \cite[Section 2.3]{bartoli2020cutting}.
Then $(U_1,U_2,U_3)$ is a cutting design.
\end{construction}

With the same spirit, more constructions can be obtained using those in \cite{alfarano2021three,alfarano2022linear,bartoli2021small,bartoli2020cutting,bartoli2021weight}.

\subsection{Connection with minimal sum-rank metric codes}

In this subsection, we extend the notion of minimal codes in the Hamming metric and the in rank metric to sum-rank metric such that this definition is consistent with that in the Hamming and in the rank metric.

Blakley and Shamir in 1979 \cite{blakley1979safeguarding,shamir1979share} introduced independently \emph{secret sharing schemes}, that are protocols for distributing a secret among a certain number of participants. Then McEliece and Sarwate presented in \cite{mceliece1981sharing} a more general construction based on linear codes (equipped with Hamming metric), where Reed-Solomon codes were used. In \cite{massey1993minimal}, Massey relates the secret sharing protocol to \emph{minimal codewords}: the minimal access structure in his secret sharing protocol is given by the support of the minimal codewords of a linear code $\mathcal{C}$ in $\fq^n$, having first coordinate equal to $1$. A codeword $x \in \mathcal{C}$ is said \textbf{minimal} if, for every $y \in \mathcal{C}$, such that $\supp(y)\subseteq \supp(x)$ we have that $y=\alpha x$ for some $\alpha \in \F_{q}$. However, finding the minimal codewords of a general linear code is a difficult task. For this reason, a special class of codes has been introduced: a linear code is said to be \textbf{minimal} if all its nonzero codewords are minimal. The analogue in the rank metric has been recently introduced by Alfarano, Borello, Neri and Ravagnani in \cite{alfarano2022linear}.
In \cite{Alfarano2020ageometric,lu2021parameters,tang2021full}, minimal codes in the Hamming metric are geometrically characterized as cutting blocking sets, whereas in \cite{alfarano2022linear} minimal codes in the rank metric are geometrically characterized as linear cutting blocking sets.
In this section we will extend these connections providing a characterization of cutting designs in terms of sum-rank metric codes, which coincides with the aforementioned links when the sum-rank metric is either the Hamming or the rank metric.

\subsection{Supports in the sum-rank metric}

The interested reader is referred to \cite{martinez2019theory} for a more detailed description of the theory of supports in the sum-rank metric.

The support of an element of $\Pi=\bigoplus_{i=1}^t \F_q^{m_i \times n_i}$ is defined as follows.

\begin{definition}
Let $X=(X_1,\dots, X_t)\in\Pi$. 
The \textbf{sum-rank support} of $X$ is defined as the space
$$\supp(X)=(\mathrm{colsp}(X_1), \mathrm{colsp}(X_2),\ldots, \mathrm{colsp}(X_t)) \subseteq \fq^\bfn,$$
where $\mathrm{colsp}(X_i)$ is the $\fq$-span of the columns of $X_i$ and $\mathbf{n}=(n_1,\ldots,n_t)$.
\end{definition}

The notion of support does not depend on the choice of the $\fq$-basis of $\F_{q^m}$, as shown by the following result (see e.g.\ \cite[Proposition 2.1]{alfarano2022linear}).

\begin{proposition}
Let $\Gamma=(\Gamma_1,\ldots,\Gamma_t),  \Lambda=(\Lambda_1,\ldots,\Lambda_t)$ be two tuples of $\fq$-bases of $\F_{q^m}$ and let $x \in \F_{q^m}^\bfn$. Then
$\supp(\Gamma(x))=\supp(\Lambda(x))$.
\end{proposition}

The above definition allows us to give the following definition of support for an element of $\F_{q^m}^{\bf n}$.

\begin{definition}
The \textbf{sum-rank support} of an element $x=(x_1, \ldots, x_t) \in \F_{q^m}^{\bfn}$ is the tuple
$$\supp_{\bfn}(x)=\supp(\Gamma(x)),$$
for any (and hence all) choice of $\Gamma=(\Gamma_1,\ldots, \Gamma_t)$, where $\Gamma_i$ is an $\fq$-basis of $\F_{q^m}$ for each $i \in[t]$.
\end{definition}

\subsection{Minimal sum-rank metric codes}

In this section we propose the notion of minimal code in the sum-rank metric, which extends the ones of minimal codes in both Hamming and rank metric. We will then characterize geometrically minimal codes as those associated with cutting designs, extending the previously known connections between minimal codes in the Hamming metric and rank.

\begin{definition}
Let $\mathrm{C}$ be an $[\mathbf{n},k]_{q^m/q}$ code. A codeword $x \in \mathrm{C}$ is said \textbf{minimal} if for every $y \in \mathrm{C}$ such that $\supp_{\mathbf{n}}(y)\subseteq \supp_{\mathbf{n}}(x)$ then $y=\alpha x$ for some $\alpha \in \F_{q^m}$. We say that $\mathrm{C}$ is \textbf{minimal} if all of its codewords are minimal. 
\end{definition}

The following characterization of the inclusion of the supports has been used in \cite{alfarano2022linear} to describe geometrically minimal rank metric codes.

\begin{theorem}(see \cite[Theorem 5.6]{alfarano2022linear}) \label{th: correspondenceminimalrank}
Let $G$ be a generator matrix of a non-degenerate $[n,k]_{q^m/q}$ code. Let $U$ be the $[n,k]_{q^m/q}$ system associated with $G$ and $u,v \in \F_{q^m}^k \setminus \{0\}$. Then,
\[
\supp_{n} (uG) \subseteq \supp_{n} (vG) \mbox{ if and only if }\langle u \rangle^{\perp} \cap U \supseteq \langle v \rangle^{\perp} \cap U.
\]
\end{theorem}

The above result naturally extends in the sum-rank metric.

\begin{theorem}\label{th:charsupp}
Let $G$ be a generator matrix of a non-degenerate $[\mathbf{n},k]_{q^m/q}$ sum-rank code. Let $(U_1,\ldots,U_t)$ be the $\Fmnk$ system associated with $G$ and $u,v \in \F_{q^m}^k \setminus \{0\}$. Then,
\[
\supp_{\mathbf{n}} (uG) \subseteq \supp_{\mathbf{n}} (vG) \mbox{ if and only if }\langle u \rangle^{\perp} \cap U_i \supseteq \langle v \rangle^{\perp} \cap U_i \mbox{ for any $i\in [t]$}.
\]
\end{theorem}

As a consequence of Theorem \ref{th:charsupp}, the geometric correspondence described in Section \ref{sec:vectsumrank} and Definition \ref{def:cutdes} give us a $1$-to-$1$ correspondence between classes of minimal sum-rank metric codes and classes of cutting designs.

\begin{corollary}
There is a $1$-to-$1$ correspondence between classes of minimal $[\mathbf{n},k]_{q^m/q}$ codes and classes of cutting $[\mathbf{n},k]_{q^m/q}$ systems.
\end{corollary}

When the sum-rank metric corresponds to the Hamming or to the rank metric, the above corollary coincides with \cite[Theorem 3.4]{Alfarano2020ageometric} (see also \cite{tang2021full}) and \cite[Corollary 5.7]{alfarano2022linear}, respectively.
Therefore, all the examples of cutting designs given in Section \ref{sec:constrcutdes} yield constructions of minimal sum-rank metric codes.

Finally, we observe that, as in the Hamming and in the rank metrics, all the one-weight sum-rank metric codes are minimal, see e.g. \cite{alfarano2022linear}.

\begin{proposition} \label{prop:oneweightimpliesminimal}
Let $\mathrm{C}$ be a non-degenerate $[\mathbf{n},k]_{q^m/q}$ code. If all the codewords of $\mathrm{C}$ have the same sum-rank metric weight then $\mathrm{C}$ is a minimal sum-rank metric code.
\end{proposition}
\begin{proof}
Let $G$ be a generator matrix of $\mathrm{C}$ and let $(U_1,\ldots,U_t)$ be an $\Fmnk$ system associated with $G$. Since all the codewords of $\mathrm{C}$ have the same sum-rank metric weight, by Theorem \ref{th:connection}, we have that $\sum_{i=1}^t\dim_{\F_q}(U_i \cap H)$ is constant for every hyperplane $H$ of $\F_{q^m}^k$. Theorem \ref{th:costantimpliescutting} can now be used to obtain the assertion.
\end{proof}

\subsection{Minimal sum-rank metric codes in the Hamming metric}

We first describe the connection between sum-rank metric codes and Hamming metric code in \cite[Section 5.1]{neri2021geometry} (see also \cite[Section 4]{alfarano2022linear}).
\noindent 

For a collection of multisets $(\cM_1,\mathrm{m}_{1}), \ldots, (\cM_t,\mathrm{m}_{t})$ of $\PG(k-1,q^m)$, we can define their \textbf{disjoint union} as
$$\biguplus_{i=1}^t (\cM_i,\mathrm{m}_{i})\coloneqq(\cM,\mathrm{m}),$$
where $\cM=\cM_1\cup\ldots\cup \cM_t$, and $\mathrm{m}(P)=\mathrm{m}_1(P)+\ldots+\mathrm{m}_t(P)$ for every $P\in \PG(k-1,q^m)$.
To every $n$-dimensional $\F_q$-subspace $U$ of $\F_{q^m}^k$, it is possible to associate the multiset $(L_U,\mathrm{m}_U)$, where  $L_U \subseteq \PG(k-1,q^m)$ is the $\F_q$-linear set defined by $U$ and $$\mathrm{m}_U(\langle v\rangle_{\F_{q^m}})\coloneqq\frac{q^{w_{L_U}(\langle v\rangle_{\F_{q^m}})}-1}{q-1}.$$
By \eqref{eq:pesivett}, this means that the multiset $(L_U,\mathrm{m}_U)$ of $\PG(k-1,q^m)$ has cardinality (counted with multiplicity) $\frac{q^n-1}{q-1}$. Consider now an $[\mathbf{n},k]_{q^m/q}$ system $(U_1,\ldots,U_t)$ and define the multiset \[
\Ext(U_1,\ldots,U_t)\coloneqq \biguplus\limits_{i=1}^t (L_{U_i},\mathrm{m}_{U_i}).
\]
Then $\Ext(U_1,\ldots,U_t)$ is a  projective $[\frac{q^{n_1}+\ldots +q^{n_t}-t}{q-1},k]_{q^m}$ system.

Hence, we can give the following definition.

\begin{definition}\label{def:hammsumrank}
Let $\mathrm{C}$ be a non-degenerate $[\mathbf{n},k]_{q^m/q}$ code. Let $(U_1,\ldots,U_t)$ be a system associated with $\mathrm{C}$. Any code $C \in \Psi^{H}(\Ext(U_1,\ldots,U_t))$, where $\Psi^H$ is defined as in Section \ref{sec:2weight}, is called an \textbf{associated Hamming-metric code} with $\mathrm{C}$.
\end{definition}

Minimal sum-rank metric codes also defines minimal codes in the Hamming metric via the codes in Definition \ref{def:hammsumrank}.

\begin{corollary}
Let $\mathrm{C}$ be a non-degenerate $[\mathbf{n},k]_{q^m/q}$ code. Then $\mathrm{C}$ is minimal if and only if any associated Hamming-metric code is minimal.
\end{corollary}
\begin{proof}
The proof directly follows from Proposition \ref{prop:charcutt2}.
\end{proof}

\section{Dimension expanders}\label{sec:dimexp}

The notion of dimension expander was introduced by Wigderson, see \cite{wigderson2004expanders}, where he also pointed out the problem of constructing these structures. A dimension expander is a collection of $m$ linear maps $\Gamma_j:\F^{\ell} \rightarrow \F^{\ell}$, where $\F$ is a field, such that for any subspace $U \subseteq \F^n$ of sufficiently small dimension, the subspace $\Gamma_1(U)+\ldots+\Gamma_m(U)$  has dimension significantly larger than $\dim(U)$. Their interest is related to the fact that they can be seen as the linear-algebraic analogue of expander graph and hence they play an important role in the theory of \emph{algebraic pseudorandomness}. See \cite[Section 1]{guruswami2021lossless} for a nice overview of the problem.
Constructions in zero characteristic were then given by Lubotzky and Zelmanov in \cite{lubotzky2008dimension} and by Harrow in \cite{harrow2007quantum}, but very recently some constructions were given when the characteristic is positive  by Guruswami, Resch and Xing in \cite{guruswami2021lossless}.

The definition over finite fields is the following. 

\begin{definition}
Let $\ell,m \geq 1$ be integers, $\eta >0$ and $\zeta>1$. Let $\Gamma_1,\ldots,\Gamma_m: \F_{q}^{\ell} \rightarrow \F_{q}^{\ell}$ be $\F_q$-linear maps. The collection $\{\Gamma_j\colon j \in[m]\}$ forms an $(\eta,\zeta)$-\textbf{dimension expander} if for all subpace $U \subseteq \F_q^{\ell}$ of dimension at most $\eta \ell$, 
\[
\dim_{\F_q} \left( \sum_{j=1}^m\Gamma_j(U) \right) \geq \zeta \dim_{\F_q} (U).
\]
The \textbf{degree} of the dimension expander is $m$.
If $\zeta=\Omega(m)$ then the dimension expander is said to be \textbf{degree-proportional}.
\end{definition}

In \cite{guruswami2021lossless}, the authors ably reduce the problem of constructing dimension expanders (with some constrains on the parameters) to the construction of subspace designs.
We will now describe the reduction and then we will use the $s$-designs to get interesting and new examples of dimension expanders.

Suppose $\ell=mk$. Let $U_1,\ldots,U_t$, with $t \mid \ell$ and $t \leq m$ be $\F_q$-subspaces of $\F_{q^{\ell}}$ with $\dim_{\F_q} (U_i) = \ell / t$. Consider
\[
\mathcal{D}=\left\{f(x)=\sum_{i=0}^{t-1}f_ix^{q^i} \colon f_i \in U_{i+1}, i \in\{0,\ldots,t-1\} \right\} \subseteq \mathcal{L}_{m,q}.
\]
Since $\F_{q^{\ell}} \cong \F_q^{\ell} \cong \mathcal{D} $ as $\F_q$-vector spaces and the maps of the dimension expanders are from $\F_q^{\ell}$ in itself, we can identify the domain of the maps of a dimension expander with $\mathcal{D}$ and the image space with $\F_{q^{\ell}}$.  Let $(\beta_1,\ldots,\beta_m)$ be an $\F_q$-basis of $\F_{q^m}$. Define
\begin{equation} \label{eq:formexpanders}
\begin{array}{rccl}
\Gamma_j: & \mathcal{D} & \longrightarrow & \F_{q^{\ell}} \\
& f(x)  & \longmapsto &  f(\beta_j),
\end{array} 
\end{equation}
for $j\in[m]$.

In \cite{guruswami2021lossless} the authors show that when $(U_1,\ldots,U_t)$ is a subspace design with certain parameters, then the construction described above of the maps $\{\Gamma_j \colon j \in [m]\}$ gives a dimension expander for which the parameters depend on the subspace design $(U_1,\ldots,U_t)$. 

\begin{theorem}(see \cite[Theorem 3.4]{guruswami2021lossless}) \label{th:subspaceexpander}
Let $\ell=mk$ with $m$ and $k$ positive integers. Let $t$ be a positive integer such that $t \leq m$.
Consider an ordered set $(U_1,\ldots,U_t)$ of $\fq$-subspaces in $\F_{q^m}^k$ such that
\[ \sum_{i=1}^t \dim_{\fq}(U_i\cap S)\leq As \]
for any $s$-dimensional $\F_{q^m}$-subspace $S$ in $\F_{q^m}^k$, for all $s \leq \mu \ell$, for some $0<\mu<1/m$ and $\dim_{\F_q}(U_i)=\ell/t$, for $i\in[t]$. Then the maps $\{\Gamma_j\colon j \in[m]\}$ defined as in \eqref{eq:formexpanders} form a $(\mu A,\frac{m-t+1}{A})$-dimension expander.
Moreover, if the subspaces $U_1,\ldots,U_t$ are explicit, then the dimension expander is explicit.
\end{theorem}

\begin{remark}
In the above result by explicit subspaces we mean that there is a polynomial time algorithm in $m$ and $k$ which gives a basis for each of those subspaces.
\end{remark}

Using the construction of $s$-designs provided in Theorem \ref{th:incollamento}, we obtain the following explicit dimension expanders.

\begin{theorem} \label{th:newexpander}
Let $m \geq q-1$. Suppose that $q-1$ divides $k$. Then there exists an explicit $\left( \frac{q-2}{mk},m-q+2 \right)$-dimension expander in $\F_q^{mk}$ of degree $m$. 
\end{theorem}
\begin{proof}
Choosing $s=q-2$, $t=q-1$, Theorem \ref{th:incollamento} ensures the existence of an (explicit) $s$-design $(U_1,\ldots,U_t)$, with $\dim_{\F_q}(U_i)=\frac{mk}{q-1}$. Let $\mu=\frac{q-2}{mk}$. Note that $\mu < 1/m$, since $q-1 \mid k$ and $k>q-2$. Moreover, by Proposition \ref{prop:diminuzione}, $(U_1,\ldots,U_t)$ is an $s'$-design for each $s' \leq \mu mk=q-2$. Hence by Theorem \ref{th:subspaceexpander} the assertion follows. Since the subspace design of Theorem \ref{th:incollamento} is explicit, then the constructed dimension expander is explicit as well.
\end{proof}

Special instances of the above result give examples of degree-proportional dimension exapanders.

\begin{corollary}\label{cor:asymdimexp}
Let $ \delta \in \left(0, \frac{q-2}{q-1} \right)$ such that $1/\delta \in \mathbb{N}$. Let $m=\frac{q-2}{\delta}$ and $k=\frac{q-1}{\delta}$. Then there exists an explicit construction of a $(\eta,\zeta)$-dimension expander in $\F_q^{mk}$ of degree $m$, where $\eta=\Omega \left( \frac{\delta m}{m}\right)$ and $\zeta=\Omega((1-\delta)m)$. In particular, it is also degree-proportional.
\end{corollary}
\begin{proof}
By choosing $m=\frac{q-2}{\delta}$ and $k=\frac{q-1}{\delta}$ in Theorem \ref{th:newexpander}, we obtain an explicit 
$\left( \frac{q-2}{mk},m-q+2 \right)$-dimension expander in $\F_q^{mk}$ of degree $m$. So
\[
\eta=\frac{q-2}{mk}=\frac{\delta(q-2)}{(q-1)m}=\Omega\left(\frac{\delta}{m}\right)
\]
and 
\[
\zeta=m-q+2=\frac{q-2}{\delta}-q+2=(1-\delta)m,
\] 
and hence the assertion follows.
\end{proof}

\begin{remark}
In \cite[Theorem 5.1]{guruswami2021lossless} the authors construct degree-proportional dimension expanders in which $q\geq (mk)^\delta$ for some $\delta>0$, whereas the construction in Corollary \ref{cor:asymdimexp} has no restrictions on $q$.
Furthermore, in \cite[Theorem 5.1]{guruswami2021lossless} further assumptions on the parameters are needed which in Corollary \ref{cor:asymdimexp} are not required.
\end{remark}

\section{Conclusions and Open problems}\label{sec:concl}

Motivated by their applications in list decoding for rank metric codes and Hamming metric codes, and to dimension expanders, in this paper we provided bounds and constructions of subspace designs in $V(k,q^m)$.
Moreover, we pointed out which subspace designs correspond to maximum sum-rank metric codes, i.e. the optimal subspace designs.
We then generalized the notion of $s$-scattered subspaces to the subspace designs, introducing the $s$-designs.
Subsequently, we proved that for $s=1$ and $s=k-1$ they correspond to the optimal subspace designs.
We paid particular attention to the case $s=1$ for which we provided examples and characterizations.
Then we also investigated cutting designs and we showed their connection with minimal sum-rank metric codes.

There are some problems that still remain open. We list some of them.

\begin{itemize}
    \item In Corollary \ref{cor:boundwait}, we proved that a maximum $s$-design in $V(k,q^m)$ is also a $(k-1,t(mk/(s+1)-m+s))_q$ subspace design. Nevertheless, when $s=1$ we proved that a maximum $1$-design is a $(k-1,tmk/2-tm+1))_q$-subspace design, improving what is proved in Corollary \ref{cor:boundwait}. This also suggests that a maximum $s$-design in $V(k,q^m)$ is also a $(k-1,t(mk/(s+1)-m)+s)_q$-subspace design, but the double counting arguments become too difficult to deal with.
    \item Construct maximum $1$-designs in the case where $t>1$, $k$ is odd and $m\ne 2$ is even.
    \item Is it true that maximum $h$-designs in $V(k,q^m)$ are optimal subspace designs also when $2\leq h\leq k-2$?
    \item In Corollary \ref{cor:8.9} we showed that maximum $1$-designs yield two intersection sets with respect to hyperplanes. In particular, this means that from a maximum $1$-design, we can define linear codes equipped with the Hamming metric with two nonzero weights (known also as \emph{two-weight codes}) and we can consider the associated strongly regular graph; see Section \ref{sec:2weight}. However, we do not know whether or not the parameters of these codes and of these graphs are new, except for the union of Baer subgeometries and the case in which the design has one element, see \cite{baker2000baer} and \cite{blokhuis2002two}. A first check with the available databases of two weight codes/strongly regular graphs (see \cite{onlinedatabases2wei,onlinedatabasessrg}) seems to suggest that the parameters are new.
    \item Classification of optimal subspace designs using known classification results on scattered spaces. This would clearly imply a classification for MSRD codes.
    \item Is it possible to give a geometric description of the Delsarte dual operation on subspace designs similarly to the Delsarte duality introduced in \cite[Section 3]{csajbok2021generalising}?
    \item Obtain new (non)-existence results of Cameron-Liebler sets by making use of the connection established in Theorem \ref{th:camerondesign}, the techniques developed in Sections \ref{sec:fromstrong1}, \ref{sec:fromstrong2}, \ref{sec:fromstrong3} and the Singleton bound for sum-rank metric codes.
    \item Show a non-trivial sum-rank metric analogue of the Ashikhmin-Barg condition \cite[Lemma 2.1]{ashikhmin1998minimal} for the minimality of a sum-rank metric code.
    \item An ordered set $(U_1,\ldots,U_t)$ of $\F_q$-subspaces of $\F_{q^m}^k$ with the property that there exists a positive integer $c$ such that $\sum_{i=1}^t\dim_{\F_{q}}(U_i \cap H)=c$ for every hyperplane of $\F_{q^m}^k$, corresponds to one-weight rank metric codes, see \cite{neri2021geometry} and Proposition \ref{prop:oneweightimpliesminimal}. As proved in Theorem \ref{th:costantimpliescutting} they also define a cutting design. It would be interesting to prove that if $L_{U_1} \cup \ldots \cup L_{U_t}$ cover the entire space or to find examples of cutting designs $(U_1,\ldots,U_t)$ such that $L_{U_1} \cup \ldots \cup L_{U_t} \neq \PG(k-1,q^m)$.
\end{itemize}

\section*{Acknowledgements}

The authors are very grateful to Sam Adriaensen and Olga Polverino for carefully reading the paper and giving several suggestions.
We would also like to thank the editor Gerard van der Geer and the anonymous referees for their suggestions, which improved the quality and the clarity of the paper.
Moreover, we thank Ferdinand Ihringer and Sam Mattheus for fruitful discussions on strongly regular graphs. 
The research was supported by the project ``VALERE: VanviteLli pEr la RicErca" of the University of Campania ``Luigi Vanvitelli'' and was partially supported by the Italian National Group for Algebraic and Geometric Structures and their Applications (GNSAGA - INdAM). The authors were supported by the project COMBINE of the University of Campania ``Luigi Vanvitelli”.

\bibliographystyle{abbrv}
\bibliography{biblio}

Paolo Santonastaso and Ferdinando Zullo,\\
Dipartimento di Matematica e Fisica,\\ 
Universit\`a degli Studi della Campania ``Luigi Vanvitelli'',\\ 
Viale Lincoln, 5,\\ 
I--\,81100 Caserta, Italy\\
E-mail: \{paolo.santonastaso,ferdinando.zullo\}@unicampania.it

\end{document}